\documentclass[11pt]{article}
\usepackage[margin=1in]{geometry}
\usepackage{graphicx} % Required for inserting images

\usepackage{setspace}
\usepackage{enumitem}

\usepackage{amsmath,amsfonts,amssymb,mathtools,amsthm,bm}
\usepackage{style} 
\usepackage{tikz}

\makeatletter
\renewcommand*\env@matrix[1][*\c@MaxMatrixCols c]{%
  \hskip -\arraycolsep
  \let\@ifnextchar\new@ifnextchar
  \array{#1}}
\makeatother

\usepackage{float}

\usepackage[usenames,dvipsnames]{xcolor}
\definecolor{mydarkblue}{rgb}{0,0.08,0.45} 
\definecolor{mydarkgreen}{rgb}{0,0.45,0.08} 
\usepackage[utf8]{inputenc} % allow utf-8 input
\usepackage[T1]{fontenc}    % use 8-bit T1 fonts
\usepackage{hyperref}       % hyperlinks 
\hypersetup{
    colorlinks=true,
    citecolor=mydarkgreen, 
    linkcolor=mydarkblue,
}

\usepackage{hyperref}
\crefname{assumption}{Assumption}{assumptions}

\newtheorem*{theorem*}{Theorem}
\newtheorem*{lemma*}{Lemma}

\usepackage{pgfplots}

\usepackage[strict]{changepage}

\usepackage{framed}

\definecolor{formalshade}{rgb}{0.95,0.95,1}

\newenvironment{formal}{%
  \MakeFramed{\advance\hsize-\width\FrameRestore}%
  \noindent\hspace{-4.55pt}% disable indenting first paragraph
  \begin{adjustwidth}{}{7pt}%
  \vspace{2pt}\vspace{2pt}%
}
{%
  \vspace{2pt}\end{adjustwidth}\endMakeFramed%
}

\usepackage{listings}

\definecolor{codegreen}{rgb}{0,0.45,0.08}
\definecolor{codegray}{rgb}{0.5,0.5,0.5}
\definecolor{codepurple}{rgb}{0,0.08,0.45}
\definecolor{backcolour}{rgb}{0.95,0.95,0.92}

\lstdefinestyle{mystyle}{
    backgroundcolor=\color{backcolour},   
    commentstyle=\color{codegreen},
    keywordstyle=\color{magenta},
    numberstyle=\tiny\color{codegray},
    stringstyle=\color{codepurple},
    basicstyle=\ttfamily\footnotesize,
    breakatwhitespace=false,         
    breaklines=true,                 
    captionpos=b,                    
    keepspaces=true,                 
    numbers=left,                    
    numbersep=5pt,                  
    showspaces=false,                
    showstringspaces=false,
    showtabs=false,                  
    tabsize=2
}
\DeclarePairedDelimiterX{\inp}[2]{\langle}{\rangle}{#1, #2}
\DeclarePairedDelimiterX{\norm}[1]{\lVert}{\rVert}{#1}

\newcommand{\Hb}{\ensuremath{H_B}}

\newcommand{\tatonnement}{t\^atonnement}
\newcommand{\RR}{\mathbb{R}}
\newcommand{\g}{\mathrm{g}}

\newcommand{\bfx}{\mathbf{x}}
\newcommand{\bfy}{\mathbf{y}}
\newcommand{\bfz}{\mathbf{z}}
\newcommand{\bfd}{\mathbf{d}}
\newcommand{\bfp}{\mathbf{p}}
\newcommand{\bfg}{\mathbf{g}}
\newcommand{\bfh}{\mathbf{h}}
\newcommand{\bfq}{\mathbf{q}}
\newcommand{\bfalpha}{\boldsymbol{\alpha}}
\newcommand{\bfB}{\mathbf{B}}

\title{T\^atonnement Dynamics for Fisher Markets with Chores}
\author{Bhaskar Ray Chaudhury$^{\dagger}$, Christian Kroer$^{\ddagger}$, Ruta Mehta$^{\dagger}$, Tianlong Nan$^{\ddagger}$}

\date{\vspace{-12pt}$^{\dagger}$University of Illinois Urbana-Champaign, $^{\ddagger}$Columbia University}

\begin{document}

\maketitle
% Optionally include a table of contents
% \vspace{1cm}
% \setcounter{tocdepth}{1} % adjust to 1 if desired
% \tableofcontents

\begin{abstract}

In this paper, we initiate the study of \emph{t\^atonnement dynamics in markets with chores}. T\^atonnement is a fundamental market dynamic, that captures how prices evolve when they are adjusted in proportion of their excess demand. While its convergence to a competitive equilibrium (CE) is well understood in goods markets for broad classes of utility functions, {no analogous results are known for chore markets}. 

Analyzing t\^atonnement in the chores market presents new challenges. Several elegant structural properties that facilitate convergence in goods markets---such as convexity of the equilibrium price set and monotonicity of excess demand under the t\^atonnement price updates---\emph{fail to hold in the chore setting}. Consistent with these difficulties, we first show that \emph{naïve t\^atonnement}, which adjusts prices proportional to the excess demand, {diverges even for the simplest case of linear disutilities}. To overcome this, we propose a modified process called \emph{relative t\^atonnement}, where prices are updated according to normalized excess demand. We prove {its {\em convergence to a CE}} 
under suitable step-size choices for a broad class of disutility functions, namely \emph{continuous, convex, and 1-homogeneous (CCH)} disutilities. This class includes many standard forms such as linear and convex CES disutilities. Our proof proceeds by showing that the relative t\^atonnement dynamics correspond to applying generalized gradient methods to a nonsmooth, nonconvex yet regular objective function---\emph{a generalization of the objective in the Eisenberg--Gale-type dual program} introduced by~\cite{chaudhury2024competitive}.  

For the case of \emph{CES disutilities}, where disutility is the weighted $p$-norm of the individual chore disutilities for $p \in (1, \infty)$, we show that relative t\^atonnement converges to an $\varepsilon$-CE in $\mathcal{O}(1/\varepsilon^2)$ iterations. This \emph{quadratic convergence rate} is established by proving smoothness of the associated objective function. We achieve this by interpreting the objective as the \emph{polar gauge (or gauge dual)} of the disutility function. Typically, smoothness of gauge dual is proven by proving strong convexity of the \emph{primal gauge}, (in this case, the disutility function). Although CES disutilities are neither strictly nor strongly convex, we are nonetheless able to prove smoothness of their gauge dual, thereby obtaining the desired rate of convergence.

Finally, following the framework of~\cite{arrow1958stability}, we analyze the \emph{stability of competitive equilibria} under the continuous-time counterpart of our relative t\^atonnement dynamics. We provide a complete characterization of local stability and show that the \emph{Nash welfare maximizing {CE}} is always locally stable when agents have \emph{linear disutilities}-- offering a new normative justification for their desirability \cite{bogomolnaia2017competitive}.

\end{abstract}

\section{Introduction} 
\label{sec:intro}

The notion of \emph{competitive equilibrium} (CE), also known as \emph{market equilibrium} (ME), has long been regarded as the crown jewel of microeconomics~\cite{nisan2007algorithmic}. In a \emph{Fisher market}~\cite{fisher1891value} with $n$ buyers and $m$ divisible items, each item has a fixed total supply, and each buyer seeks to maximize their utility subject to a budget constraint. At a CE, the price for each item is set such that the \emph{market clears}, i.e.,  when every buyer demands their utility-maximizing bundle at those prices, the aggregate demand for each item exactly equals its supply. It is a key solution concept in perfectly competitive markets, and its static and dynamics properties were central research areas in economic theory.  For the past few decades, CE has been utilized to design internet market mechanisms~\cite{eden2023platform}, develop online resource allocation algorithms~\cite{salehi2023competitive}, and has garnered significant attention within the computer science communities.
% and operations research communities.
% \RM{Not sure if the last sentence is needed.}

The most extensive study of CE has occurred for goods markets, i.e.,  markets where agents receive positive utilities from the items. Beyond the traditional static properties of a CE, such as existence, uniqueness, and welfare theorems, significant attention has also been given to its dynamics properties, including the convergence to an equilibrium and the stability of CE~\cite{arrow1958stability, samuelson1941stability, cheung2019tatonnement, cole2008fast}. The significance and objectives of studying the dynamics properties of a CE are effectively outlined in the seminal work of~\cite{arrow1958stability}.

\begin{formal}
    ``\emph{The task consists in constructing a formal dynamics model whose characteristics reflect the nature of the competitive process and in examining its stability properties, given assumptions as to the properties of the individual units or of the aggregate excess demand functions.}''
\end{formal}

A central dynamics concept studied since the early development of competitive equilibrium theory in the 19\textsuperscript{th} century is \emph{(Walrasian) t\^atonnement}~\cite{walras1874elements}. t\^atonnement represents a natural price-adjustment process driven by excess demand: prices are increased for goods whose aggregate demand exceeds supply, and decreased when supply surpasses demand, typically in proportion to the magnitude of the imbalance. In their seminal work, \cite{arrow1958stability} established sufficient conditions under which t\^atonnement converges to a competitive equilibrium. Since then, a rich body of research has explored the convergence, stability, and computational properties of t\^atonnement and related market dynamics in goods markets~\cite{codenotti2005market, cole2008fast, cheung2012tatonnement, cheung2019tatonnement}.

%\subsection{Chores Market} 

Recently, competitive markets involving \emph{chores}--non-disposable items that incur cost during consumption--have garnered significant attention due to their relevance in online labor markets (e.g., \emph{Upwork, Amazon's Mechanical Turk}), freelance (e.g., \emph{Fiverr, Toptal}) and gig platforms (e.g., \emph{TaskRabbit}). Despite their importance, the dynamic properties of CE in the context of chores are yet to be explored.

\begin{center}
  \emph{In this paper, we initiate the study on t\^atonnement dynamics in the chores market, drawing natural parallels and distinctions to the t\^atonnement literature in the goods market.}    
\end{center}

\paragraph{Difficulty in the chores market: \emph{non-monotone excess-demand}.} When adapting t\^atonnement to the chores market, we encounter a fundamental obstacle that makes proving convergence significantly more challenging than in the goods setting. For any fixed price vector, the excess demand for an item is defined as the difference between its aggregate demand and its available supply, given that all buyers best respond\footnote{demand their utility-maximizing bundles subject to their budget constraints.} to those prices. In the goods setting, the t\^atonnement dynamics updates prices by adding a term proportional to the excess demand: if the excess demand for a good is positive (demand exceeds supply), its price increases; if negative (supply exceeds demand), its price decreases.

The key reason this process converges in goods markets with fairly general buyer preferences, e.g., Weak Gross Substitute (WGS), homothetic, etc., lies in the monotonic behavior of excess demand under such price updates. Increasing the price of a good whose demand exceeds supply typically (i) reduces the consumption of that good by agents who were previously demanding it, and (ii) may shift some of their demand to other goods, further reducing aggregate demand for that good. A symmetric argument applies when supply exceeds demand. This monotonic response of excess demand stabilizes the t\^atonnement dynamics and underpins the convergence guarantees of classical Walrasian algorithms and dynamics~\cite{codenotti2005market}.

When we naturally extend the t\^atonnement process to the chores setting, we incorporate a negative correction term proportional to the excess demand, i.e., if demand exceeds supply, the corresponding price decreases, and if supply exceeds demand, the price increases. However, unlike in the goods setting, the excess demand function in chores markets fails to satisfy monotonicity-- even under simple linear disutility functions. The failure of this property is also at the heart of %\redit{the difficulty in deriving polynomial-time algorithm to find an exact CE for the chores setting.} 
why there are no polynomial-time algorithms known 
% \redit{so far} 
so far 
to determine a CE for the chores setting. 
Consider the following example,

\begin{example}
Consider a simple chores market with two agents and two chores. Each agent has a budget of \$1, and each chore has a total supply of one unit. The agents have \emph{linear separable disutilities} specified by the matrix
\[
d_{i,j} =
\begin{cases}
1 & \text{if } i = j,\\
2 & \text{if } i \neq j,
\end{cases}
\quad \text{for } i,j \in \{1,2\}.
\]
Here, $d_{i,j}$ represents the disutility incurred by agent~$i$ from performing one unit of chore~$j$. Consider the price vector $(p_1, p_2) = (\tfrac{5}{4}, \tfrac{3}{4})$. Given these prices, one can verify that the optimal (disutility-minimizing) demand bundles are
\[
\mathbf{x}^*_1 = (\tfrac{4}{5}, 0) \quad \text{and} \quad \mathbf{x}^*_2 = (\tfrac{4}{3}, 0).
\]
At these allocations, supply exceeds demand for chore~1 ($\tfrac{4}{5} < 1$), while demand exceeds supply for chore~2 ($\tfrac{4}{3} > 1$).  In a chores t\^atonnement process, prices are updated in the \emph{opposite} direction from the goods setting: prices are \emph{increased} for chores where supply exceeds demand and \emph{decreased} where demand exceeds supply. Thus, the update rule would increase $p_1$ and decrease $p_2$. However, this change actually amplifies the magnitude of excess demand—agents end up demanding even more of chore~2 and less of chore~1—before eventually correcting. Only when the prices reach $(p_1, p_2) = (\tfrac{4}{3}, \tfrac{2}{3})$ does the market clear exactly. 

This simple example highlights the central challenge: in chores markets, t\^atonnement price adjustments can initially drive the system \emph{away} from equilibrium—when distance is measured by the imbalance in excess demand—making the convergence analysis substantially more delicate.
\end{example}

\paragraph{Na\"ive adaptation fails.}As a natural starting point, we adapt the t\^atonnement dynamics from the goods setting, updating the price of each chore proportionally to its excess demand. Perhaps unsurprisingly, this na\"ive extension fails to converge, under a variety of disutilities. 
In particular, we construct a family of examples 
(Figure~\ref{fig:naive-chores-t\^atonnement}) 
featuring a unique competitive equilibrium at which t\^atonnement diverges—even when initialized arbitrarily close to equilibrium. This divergence is robust to all choices of step sizes (including adaptive ones) and tie-breaking rules.

A closer examination of this example reveals that na\"ive t\^atonnement systematically drives the price vector \emph{outside the price simplex}, i.e., the hyperplane where the sum of prices equals the sum of budgets. This observation suggests a natural remedy: constrain the price updates to remain within the simplex. To this end, we introduce \emph{relative t\^atonnement}, a variant in which the price of each chore is updated proportionally to its \emph{relative excess demand}—defined as the excess demand minus its average across all chores. This modification guarantees that the total price mass remains constant, ensuring that prices evolve within the simplex.

\paragraph{Convergence of \emph{relative t\^atonnement}.} We next turn to our main contribution: establishing convergence guarantees for \emph{relative t\^atonnement}. Surprisingly, despite the inherent non-monotonicity of excess demand in chores markets, relative t\^atonnement converges under remarkably general conditions on agent disutilities.

\begin{theorem*}[Informal]
    Relative t\^atonnement converges to a competitive equilibrium when agents have convex and 1-homogeneous disutility functions.
\end{theorem*}
%\RM{Utility functions are by default normalized. So lets not mention "normalized". It creates confusion that there is an additional constraint.}

This result closely parallels the classical convergence guarantees of (na\"ive) t\^atonnement in goods markets {for concave 1-homogeneous utility functions (homothetic markets), which,} %, extending them to the chores setting in full generality. 
to the best of our knowledge, represents one of the most general non-trivial classes of preferences for which t\^atonnement convergence is known in the goods setting. The only other broad class is that of \emph{weak gross substitutes (WGS)} preferences, whose definition itself is {{\em monotonic excess demand}-- a property that is inapplicable to the chores market.} 
%convergence analyses crucially rely on the monotonicity of excess demand—a property that fails to hold even for linear disutilities in chores markets.

Our proof proceeds by constructing a \emph{Lyapunov potential function}—a non-convex scalar function defined over the price simplex—that strictly decreases along the relative t\^atonnement dynamics. This potential can be viewed as a natural generalization of the Eisenberg–Gale (EG) dual program for chores introduced in~\cite{chaudhury2024competitive} for linear disutilities. We further show that the vector of relative excess demands coincides with the generalized gradient of this potential, implying that relative t\^atonnement corresponds to a first-order dynamical system on this generalized non-convex program. 
% \RM{Do we want to say that we show the result for both continuous and discrete dynamics, in 1-2 sentences?}

\paragraph{Polynomial convergence to $\varepsilon$-CE for CES (weighted $\ell_\rho$-norm) disutilty functions.}
% \textcolor{red}{TODO: Understand dependence on $u_{i,j}$ and $\rho$ for convergence of tatonnement in the goods setting-- IMP for comparing our convergence results.}

We further establish a polynomial convergence rate for a broad and economically meaningful class of \emph{CES disutilities}. In this model, the total disutility from a bundle of chores is given by the $\ell_{\rho}$-norm of the vector of individual disutility, for some parameter $\rho \ge 1$. This formulation captures varying degrees of \emph{substitutability} and \emph{complementarity} among chores. When $\rho = 1$, we recover \emph{linear disutilities}, corresponding to perfectly substitutable chores—each unit of effort can be reallocated freely across chores without changing the overall burden.  For $\rho = 2$, the model introduces a mild interdependence among chores: the total disutility is minimized when effort (or disutility) is distributed relatively evenly across the tasks and increases when it is concentrated on only a few. Put differently, agents perceive an imbalanced allocation of chores as more burdensome, reflecting a natural preference for spreading effort across multiple tasks rather than overloading a subset. As $\rho$ increases, this effect becomes stronger—agents become increasingly sensitive to imbalance, and the total disutility is progressively dominated by the most demanding chores. In this way, larger values of $\rho$ capture a continuum from nearly additive, independent disutilities to strongly “imbalance-averse” or complementary disutilities across chores.

\begin{theorem*}
    When agents have CES disutility functions with $\rho \in (1, \infty)$, relative t\^atonnement converges to an $\varepsilon$-competitive equilibrium in $\tilde{\mathcal{O}}(1/\varepsilon^2)$  iterations.\footnote{$\tilde{O}(\cdot) $ hides dependencies on $\rho$ and the budgets of the agents.}
\end{theorem*}

This establishes the first polynomial-time convergence guarantee for t\^atonnement under CES disutilities with $\rho \in (1, \infty)$, encompassing several important subclasses of preferences. Our analysis builds on a careful quantification of the “\emph{demand stability}” property of CES disutilities-- small changes in prices induce small changes in demand. This property can be intuitively visualized through the smooth, rounded contours of CES disutility functions (see Figure~\ref{fig:strongly-convex-set-and-smoothness}).

Formally, this stability is captured through the same Lyapunov potential that governs the t\^atonnement dynamics. This potential involves the logarithm of a quantity that can be interpreted as an agent’s \emph{maximum earning potential, given fixed prices}-- the maximum total payment an agent can obtain for performing chores, subject to a bound on their overall disutility. We interpret this expression as a the \emph{gauge dual} or equivalently the \emph{polar gauge}~\cite{friedlander2014gauge} with the agent's disutility function being the \emph{primal gauge}. 

To achieve the $\mathcal{O}(1/\varepsilon^2)$ convergence rate, the key technical challenge lies in establishing the smoothness of this Lyapunov potential. This, in turn, reduces to proving smoothness of the gauge dual—a nontrivial task, since CES disutility functions are not even strictly convex. We overcome these challenges through careful structural observations and tailored arguments, a brief overview of which is provided in Section~\ref{subsub:poly-convergence}.

\paragraph{Stability of Individual CE.}  Following the line of inquiry initiated by Arrow and Hurvicz~\cite{arrow1958stability}, we next examine the \emph{stability} of competitive equilibria (CE) under the \emph{continuous-time} analogue of the relative t\^atonnement dynamics. Informally, an equilibrium price vector $\bfp^*$ is said to be \emph{locally stable} with respect to a continuous-time dynamical system if there exists a non-empty feasible neighborhood $\mathbf{N}(\bfp^*)$ around $\bfp^*$ such that every trajectory starting in $\mathbf{N}(\bfp^*)$ converges to $\bfp^*$. 

We provide several equivalent and insightful characterizations of locally stable CEs (see Theorem~\ref{thm:characterize-locally-stable-chores-ce}) when agents have linear disutility functions. One characterization identifies locally stable equilibria as local minima of a Lyapunov potential function closely related to the objective of the Eisenberg--Gale-like dual program for chores introduced by~\cite{chaudhury2024competitive}. Another characterization is combinatorial: a CE is locally stable if and only if the \emph{minimum pain-per-buck (MPB)} graph%—defined as the bipartite graph between agents and chores with an edge $(i,j)$ whenever $d_{ij}/p_j \le d_{ik}/p_k$ for all $k \in [m]$—
is connected.  
Another {interesting} characterization is that the CE that maximizes \emph{Nash welfare}, i.e., $\prod_{i=1}^n d_i(\bfx_i)$, is always locally stable.

\begin{theorem*}
    When agents have linear disutility functions, the Nash welfare-maximizing competitive equilibria are locally stable.
\end{theorem*}

Interestingly,~\cite{bogomolnaia2017competitive} argue—``\emph{without any strong normative reason}''—that the Nash welfare-maximizing CE is the most desirable among the multiple possible equilibria that a chores market may admit. Our results offer such a normative justification: stability provides a natural dynamic rationale for why the Nash welfare-maximizing equilibrium should indeed be viewed as the desirable outcome.

\paragraph{Remark on Global Stability.} Arrow and Hurwicz \cite{arrow1958stability} also define {\em global-stability}, where a CE $p^*$ is said to be globally stable if starting from any initial feasible prices $p^0$, any trajectory $p(t)$ converges to $p^*$. They show that in the weak-gross-substitutes goods market, the unique CE is globally stable. This follows from the fact that the dynamics has to converge to the unique CE of the market. The weak-gross-substitute property subsumes markets with linear utilities in the goods case. However, such a result is impossible in the case of chores market, even with linear disutilities, because of a multiplicity of equilibria. We show that this holds even if we {\em relax} the definition to require that starting from a uniformly random feasible $p^0$, the dynamics converges to $p^*$ with probability one. In particular, we design instances (see Figure~\ref{fig:instances-stability-NW}) that have multiple locally-stable equilibria, and thereby none of these can be globally-stable even under the relaxed definition.

%\RM{Should we re-emphasize that even though none is globally stable, all locally stable are NW maximizing, so starting uniformly at random, tatonnement converges to this ``desired equilibrium'' with probability one.}

\begin{figure}[t]
    \centering
    \includegraphics[width=1\linewidth]{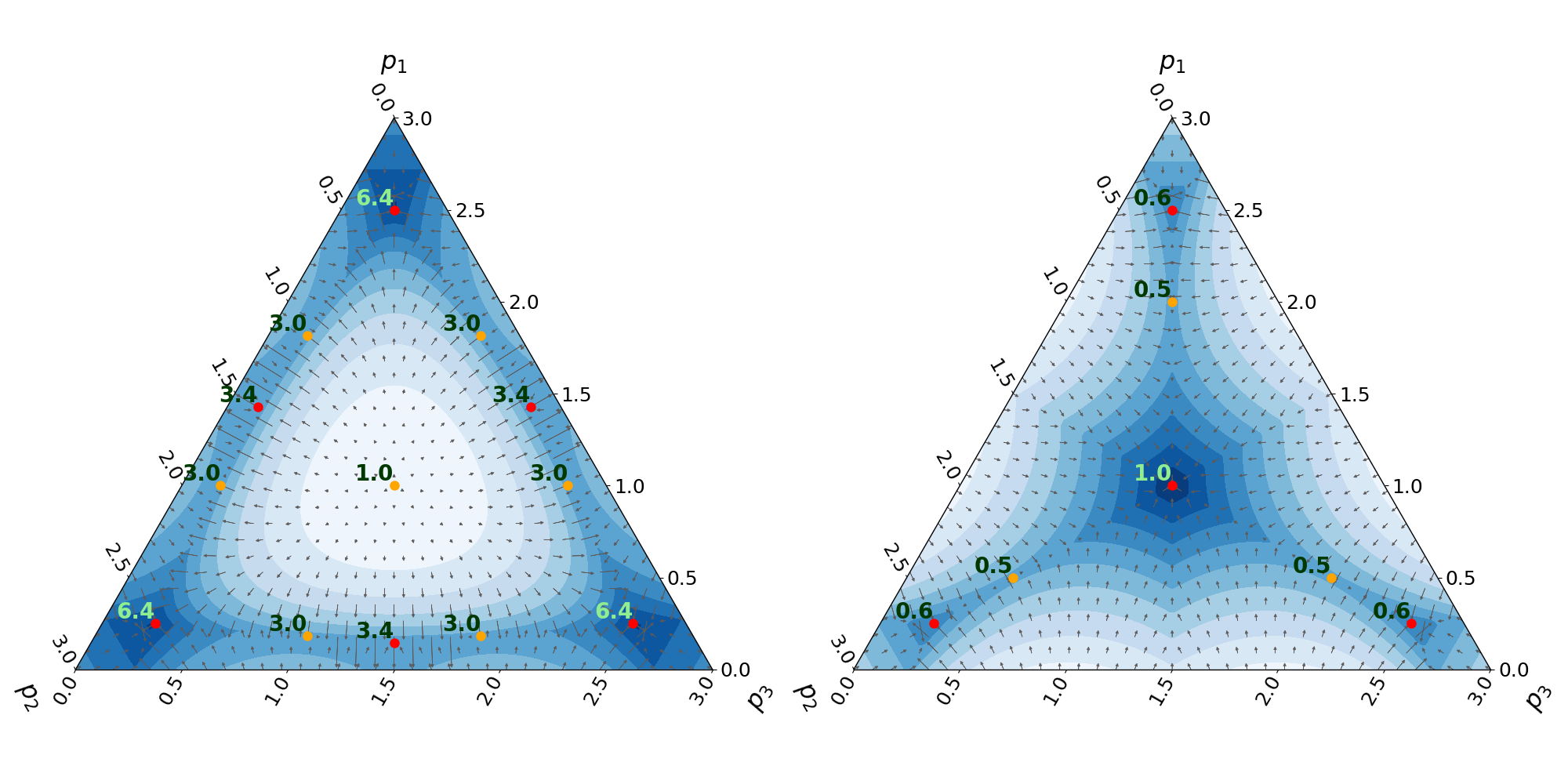}
    \caption{Stability and associated Nash welfare of chores competitive equilibrium visualized as ternary plots (also known as simplex plots). The details on the instances can be found in~\cref{sec:intro}. 
    We make the plots in the price simplex consisting of nonnegative prices that sum to the sum of budgets. 
    We use blue color to denote the level sets of the objective values of our potential function (see \cref{eq:potential}). The darker the color is, the lower the objective value is. 
    We use red and yellow dots to denote ``locally stable'' and ``unstable'' chores CE, respectively. 
    We use bold texts to denote the values of ``Nash welfare'' associated with each CE, and the maximum Nash welfare CE is denoted in light green text.
    The gray arrows (flows) denote the relative excess demand vector. If there are multiple relative excess demands, we pick the one with the minimal Euclidean norm.}
    \label{fig:instances-stability-NW}
\end{figure} 

\subsection{Technical Highlights}
We study t\^atonnement in Fisher markets with chores, where a set of $n$ agents are deciding on what chores to take on from a set of $m$ divisible chores in order to earn money. 
%In a Fisher market with chores, we have $n$ agents and $m$ divisible chores.
Each chore has a unit supply.
An allocation $\bfx \in \mathbb{R}^{n \times m}_{+}$, $x_{ij}$ represents the amount of chore $j$ allocated to agent $i$. 
Each agent $i$ incurs a disutility $d_i (\bfx_i)$ for a bundle $\bfx_i = (x_{i1}, x_{i2}, \dots, x_{im})$ of chores that is assigned to her. 
We consider each $d_i$ to be a general continuous, convex, $1$-homogeneous (CCH) disutility function.
Agent $i$ has an earning requirement of $B_i$. %, i.e., at a given set of prices for the chores, each agent $i$ needs to earn $B_i$ units of money from the chore bundle $x_i$. 

Each agent $i$ has an earning requirement of $B_i$. Given prices $(p_1,\dots,p_m)$ of chores, where $p_j$ represents the payment-per-unit of chore $j$ done, agent $i$ will demand a bundle that minimizes her disutility subject to earning at least $B_i$ money. When, for each chore, its aggregated demand from all the agents meets its supply, the prices are said to be at CE. That is, $(\bfp,\bfx)$ are at CE iff they satisfy the following conditions:

%The prices are said to be at equilibrium if 

%A CE is an allocation $\bfx$ and a price vector $\bfp \in \mathbb{R}^m_+$, where $p_j$ represents price of chore $j$, satisfying the following conditions:
\begin{itemize}
    \item \emph{Optimal bundles:} Each agent $i$ gets a disutility-minimizing chore bundle subject to her earning requirement. Letting $x^*_i(\bfp) = \textnormal{argmin}_{\bfy_i \in \mathbb{R}^m_+} \left\{  d_i(\bfy_i) \mid \inp*{\bfp}{\bfy_i} \geq B_i \right.\}$, this means that $\bfx_i \in x^*_i(\bfp)$ for all $i \in [n]$; 
    \item \emph{Market clearing:} Every chore is completely allocated, i.e., $\sum_{i \in [n]} x_{ij} = 1$ for all $j \in [m]$.
\end{itemize}

Throughout the paper, the summations of vectors are component-wise operations. 
$\bm{1}_m$ and $\bm{0}_m$ represent the $m$-dimensional all ones and all zeros vectors, respectively.
% \paragraph{Economic Dynamics~\cite{arrow1958stability}.} An economic dynamical system captures how prices evolve based on their current values. %Arrow and Hurwicz express the system as a set of simultaneous differential inclusions

\subsubsection{Na\"ive t\^atonnement and its Divergence}

T\^atonnement is an economic dynamical system that captures how prices evolve based on the current demand and supply. 
A price trajectory is described by a time-based parameterization of the price vector: $\bfp(t)$ denotes the price vector at time $t > 0$. Arrow and Hurwicz~\cite{arrow1958stability} consider a continuous-time evolution of the price vector $\bfp(t)$, while a substantial body of work in the EconCS community has considered the evolution of $\bfp(t)$ in discrete time steps\footnote{Primarily to accommodate dynamics in discrete-time decision-making processes in online markets, as well as the possibility of implementation on a computer.}. 
We consider both the continuous and discrete-time settings. In the discrete-time setting, we model the evolution of prices as a sequence of price vectors $\{\bfp^t\}_{t \ge 0}$, where each $\bfp^t$ denotes the market prices at iteration~$t$.

Given any valid price vector $\bfp \in \mathbb{R}^m_{+}$, we define the (aggregated) \emph{excess demand} $Z(\bfp)$ as the \emph{correspondence} $Z(\bfp) = \{\sum_{i \in [n]} \bfx_{i} -\bm{1}_m \mid \bfx_i \in x^*_i(\bfp) \text{ for all } i \}$.
Note that, a price vector $\bfp^*$ is a CE if $\bm{0}_m \in Z(\bfp^*)$.
A t\^atonnement process reads signals from one excess-demand vector in $Z(\bfp)$ and updates prices in response, with the goal of approaching an equilibrium price vector.

As an initial approach, we adapt \emph{na\"ive t\^atonnement} from the goods setting to the chores setting, where the price of a chore is {\em additively} adjusted proportional to the excess demand.\footnote{Typically, this would just be called \emph{t\^atonnement}, we add the modifier \emph{na\"ive} due to its failure in the chores setting, as we demonstrate.}
%the price is reduced when aggregate demand exceeds supply, and increased when supply exceeds demand.  
Given a price vector $\bfp^t$ and the excess demand set $Z(\bfp^t)$, na\"ive t\^atonnement updates the prices as follows \cite{codenotti2005market}: 
\begin{align*}
    \bfp^{t + 1} = \bfp^t - \eta^{t} \bfz^t, \hspace{20pt} \bfz^t \in Z(\bfp^t), 
\end{align*}
where $\eta^t > 0$ is the stepsize. 
We begin by observing the non-convergent behavior of this na\"ive dynamics. 
In Section~\ref{sec: naive unstable}, we present an example with two chores (see Figure~\ref{fig:naive-chores-t\^atonnement}) and a unique CE price vector $\bfp^*=(\frac{1}{2},\frac{1}{2})$. We show that, starting at any $\bfp^0\neq \bfp^*$, the sum of the prices $(p^t_1+p^t_2)$ strictly increases under the na\"ive t\^atonnement dynamics, implying its divergence (Lemma \ref{lem:divergence}). This shows divergence in a strong sense: even when starting {\em arbitrarily} close to the CE, and furthermore, for any choice of (possibly adaptive) step-sizes and tie-breaking rules. Next, 
we extend this divergence result to any class of convex CES disutility functions with $\rho \in [1,\infty)$. 
This stands in stark contrast to the goods setting, where in CES markets (for a full spectrum of $\rho$) na\"ive additive t\^atonnement converges to a CE from any arbitrary initial point~\cite{codenotti2005market,nan2024convergence}.
% We remark that the foregoing contrast is not very surprising, considering the well-known structural irregularities associated with the CE in the chores market (e.g., the disconnectivity in the set of CE in the chores market).
Finally, %In~\cref{sec: naive unstable}, 
we also provide an economic explanation for the price divergence under na\"ive t\^atonnement. 

\subsubsection{Convergence of Relative t\^atonnement}

The divergence of na\"ive t\^atonnement necessitates finding an alternate dynamical system which is natural, simple and guaranteed to converge. To this end, we introduce \emph{relative t\^atonnement}. %in the chores market. 
Given a price-vector $\bfp$, define the \emph{relative excess demand correspondence} $\tilde{Z}(\bfp) = \{\bfz - \frac{1}{m} \bm{1}^\top_m \bfz \mid \bfz \in Z(\bfp)\}$. Relative t\^atonnement is then defined by replacing $Z(\bfp^t)$ with $\tilde{Z}(\bfp^t)$ in the na\"ive t\^atonnement. 
% Secondly, we also study a multiplicative variant of t\^atonnement, which does not explicitly use relative excess demand, but where we show that, with proper initialization, it enforces price relativity correctly.
Formally, the relative t\^atonnement dynamics is as follows:
\begin{equation*}
    \bfp^{t+1} = \bfp^{t} - \eta^t \tilde{\bfz}^t, \hspace{20pt} \tilde{\bfz}^t \in \tilde{Z}(\bfp^t). 
\end{equation*}
% \begin{itemize}
%     \item Relative t\^atonnement: 
%     $p^{t+1} = p^{t} - \eta^t \tilde{\zeta}^t$, where $\tilde{\zeta}^t \in \tilde{z}(p^t)$
%     % , and 
    
%     % \item Multiplicative t\^atonnement: 
%     % $p^{t+1} = p^{t} (1 - \eta^t \zeta^t)$ where $\zeta^t \in z(p^{t})$. 
% \end{itemize}
In Sections \ref{subsec:Relative chores t\^atonnement is stable}, we prove that surprisingly, with such a mild modification, foregoing relative t\^atonnements converge to a CE in the chores market!  Note that, as a neat fix to the previous diverging issue, our price adjustment mechanism responds to the excess demands compared to the averaging level, 
and maintain the total prices at a constant level:
$$\bfp^t \in H_B := \{ \bfp \in \RR^m \mid \sum_{j=1}^m p_j = \sum_{i=1}^n B_i \} \mbox{ for all $t \geq 0$.}$$
This circumvents the problematic scenarios where the total price of all the chores diverges to infinity or zero over time, which consistently happens for na\"ive t\^atonnement dynamics (Lemma \ref{lem:divergence}). The result is first proved  for the continuous-time dynamics and then extended to discrete in Section \ref{subsec:CCH-Discrete}. %~\cref{fig:naive-chores-t\^atonnement}.
\medskip

\noindent{\bf Sketch of the Convergence Proof.}
We first consider the continuous-time dynamics, defined by $\frac{d \bfp}{d t} \in -\tilde{Z}(\bfp)$, and show that the set of its stationary points (restricted to $H_B$) exactly corresponds to the set of CE (Lemma \ref{key-lem:stationary-point-in-affine-hull-is-CE}). This implies that if the dynamics converge point-wise, then it has to converge to a CE.

%Further, as we show in~\cref{key-lem:stationary-point-in-affine-hull-is-CE}, the intersection of the set of the stationary points of the relative t\^atonnement dynamics and the affine hull $H_B$ exactly corresponds to the set of CE.

% First, we show this process is well defined itself - the prices stay in the positive orthant. To show this, we identifies a market-specific constant and build a barrier along the boundary of the positive orthant. 
% As a result, the continuous trajectory will never touch the boundary of the positive orthant.
% \begin{lemma*}
%     Any price trajectory $\bfp(t)$ of \textnormal{(\ref{chores-tatonnement})} starting from any initial point $\bfp(0) \in \Delta_B$ is well-defined and stays in $\Delta_B$ for $t \geq 0$. 
% \end{lemma*}
So then the next step is to show point-wise convergence. To do this, %, to show the point-wise convergence, 
for both continuous and discrete dynamics, 
we introduced a novel {\em Lyapunov potential function} applicable for the general CCH disutility functions: 
\begin{equation}
    \mbox{The restriction $f|_{\Hb}(\bfp):\Delta_B\rightarrow \RR$  where, }  f(\bfp) = - \sum_{j = 1}^m p_j + \sum_{i = 1}^n B_i \log\left( \max_{\bfx_i \geq 0: d_i(\bfx_i) \leq 1} \inp{\bfp}{\bfx_i} \right), 
    \label{tech-highlight:eq:potential}
\end{equation}
where $\Delta_B=H_B \cap \RR^m_+$ is the price simplex. Next we show one-to-one mapping between the generalized subdifferential of $f|_{\Hb}$ and the excess demand correspondence (see Lemma \ref{lem:subdifferentially-regular-and-generalized-subdifferential}):
\begin{lemma*}
    $\partial f\vert_{H_B}(\bfp) = \tilde{Z}(\bfp)$ for any $\bfp \in \Delta_B$. 
\end{lemma*}

The above lemma shows that the relative t\^atonnement dynamics can be seen as a subgradient dynamical system associated with the unconstrained optimization problem $\min_{p \in \mathbb{R}^m} f\vert_{H_B}$.  
Therefore, convergence of the first-order method to a local-minima of $f\vert_{H_B}$ implies point-wise convergence of the relative T\^atonnement dynamics to a CE (see the proof of Lemma \ref{lem:descent}). Now, our task reduces to show the convergence of the first-order method on $f\vert_{H_B}$ where the tricky part is to handle non-convexity and non-smoothness at the same time. 

To prove the above lemma, Lemma \ref{lem:subdifferentially-regular-and-generalized-subdifferential} first shows that, even as a nonconvex nonsmooth function, $f|_{\Hb}$ is locally Lipschitz and subdifferentially regular.
To show the regularity property, intuitively, we need to show that 
there are no  \emph{downward facing cusps} in the graph of the function.
This property is crucial to prove the convergence of the first-order method (Lemma \ref{lem:descent}), because, in the absence of regularity, the downward facing cusps can cause problems for the first-order method to converge. If they exist even at an exact stationary point, there can be a descent direction that prevents the dynamics from terminating.
In the chores market language, 
our potential function ensures that the relative t\^atonnement dynamics is stable at any exact CE.
We illustrate this point in~\cref{fig:subdifferential-regularity}.
Furthermore, this regularity guarantees strict descent at any non-stationary point, establishing the foundations for convergence of the continuous-time dynamics, formally proved in Theorem \ref{thm:relative-tatonnement-stable}. %and its discretizations. 

Extending the result to discrete-time dynamics is trickier (Section \ref{subsec:CCH-Discrete}), since in general, it is common for continuous-time gradient
descent procedures to have better convergence properties than their discrete-time analogues. Following the general framework of stochastic approximation ~\cite{davis2020stochastic}, 
we isolate conditions under which the discrete-time sequence can be seen as an approximation to the continuous-time trajectory of the dynamical system. Building on this, we show the convergence of discrete-time relative t\^atonnement dynamics when %we have to be careful about the step size ($\eta^t$ at time $t$), in particular that 
the step-size sequence ($\eta^t, t=1,2,...$) is nonnegative and square summable, but not summable (Theorem \ref{thm:discrete-time-relative-tatonnement-convergence}).

%\RM{What extra challenges we face to extend the convergence result to the discrete time dynamics? Can we highlight those in a short para here? Thanks!}

\begin{figure}[t]
    \centering
    \includegraphics[width=0.75\linewidth]{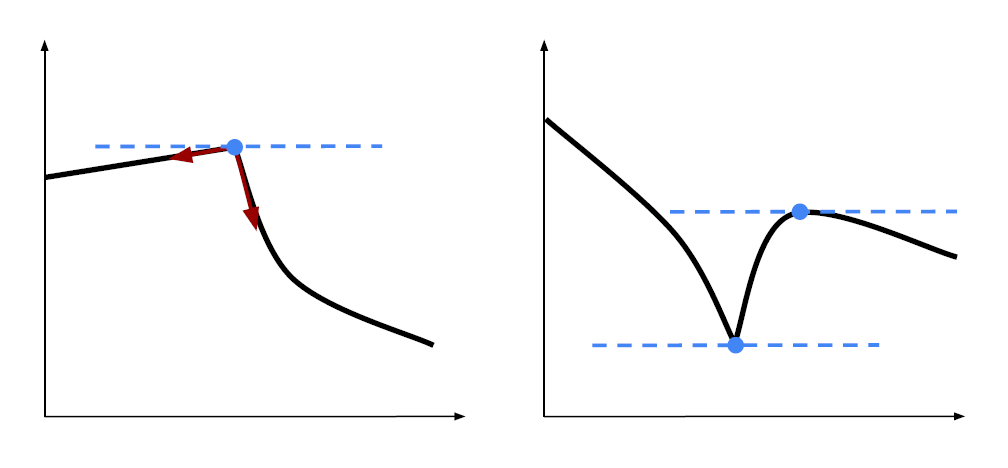}
    \caption{The left is a function with ``downward facing cusps'' and hence is {\em not regular}, while the right is a ``regular'' function.
    Both of them are nonsmooth nonconvex. The left function is not compatible with first-order dynamics because of the descent directions (denoted by red arrows) at the unique stationary point.
    In contrast, a ``regular'' function is nice because all first-order generalized gradients pointing to \emph{non-descent} directions at any stationary point.}
    \label{fig:subdifferential-regularity}
\end{figure}

% \RM{Talk about KKT point of $min_{\sum_j p_j = \sum_i B_i} f(p)$ mapping to CE for CCNH disutility functions.}
\medskip

\noindent{\bf CE as KKT points for markets with CCH disutility functions.} For the case of CCH disutility functions, \cite{bogomolnaia2017competitive} characterized CE as KKT points of a mathematical formulation--minimize product of disutilities subject to total allocation equals supply--{\em under additional condition that disutility of no agent should be zero (non-pole KKT points)}.\footnote{Note that, the global minimizer of the product of disutilities will always set one of the agent's disutility to zero. Thus strikingly this characterization requires to avoid the global minimizers!}  This latter (open) constraint turns out to be one of the trickiest to handle in designing efficient algorithms \cite{BoodaghiansCM22}. For the case of linear disutility functions \cite{chaudhury2024competitive} circumvents this through a novel formulation. As a by product of our potential function, we generalize this result to general CCH disutility functions. 
In particular, our potential function naturally provides a descent program that captures all the equilibrium prices as its KKT points. 

Interestingly, by embedding the “\emph{redundant constraint}” (sum of prices equal sum of earnings requirements) in Chaudhury et al. (2024) and the non-negativity constraint, the program circumvents the major “\emph{poles}” issue while preserving the one-to-one mapping between competitive equilibria and KKT points. In this setting, a \emph{pole} is a direction in the domain along which the objective function can go to infinity-- for instance, letting certain variables approach a boundary causes the objective to blow up-- so standard first-order methods may be drawn towards these pathological “\emph{infinite valleys}” rather than meaningful equilibria. By removing these directions via the redundant constraint, the modified program becomes free of poles. This significantly expand the boundary of usage of the optimization methods in solving CE in nonconvex chores markets.
\begin{equation}
    \begin{aligned}
        \min_{\bfp \geq 0} \quad & - \sum_{j = 1}^m p_j + \sum_{i = 1}^n B_i \log\left( \max_{\bfx_i \geq 0: d_i(\bfx_i) \leq 1} \inp{\bf p}{\bf x_i} \right) \\ 
        \textnormal{s.t.} \quad & \sum_{j=1}^m p_j = \sum_{i = 1}^n B_i. 
    \end{aligned}
    \tag{General EG Dual}
    \label{tech-highlight:pgm:general-eg-dual}
\end{equation}

\begin{theorem*}[Informal]
    There is a one-to-one mapping between the set of equilibrium prices in the chores Fisher market and the set of primal points to~\eqref{tech-highlight:pgm:general-eg-dual} that satisfies the KKT conditions.
    Moreover,~\eqref{tech-highlight:pgm:general-eg-dual} has a finite lower bound.
\end{theorem*}
% We remark that this program works for a board class of functions- much broader than the previous linear disutilities.

% \RM{Also need to show that $f(p)$ satisfies a set of nice properties -- no non-smoothed point towards -ve side.}

% In our proof of convergence, we show that both processes can be viewed as approximations of their continuous-time (dynamical system) counterparts, and establish stability and convergence of these continuous-time systems. 
% \textcolor{red}{BRC: We need to write a few words about the techniques}
% \tianlong{ 
% To leverage the advanced convergence results, for example those in~\cite{davis2020stochastic,ding2024stochastic}, we need to 
% \begin{itemize}
%     \item show our dynamical system is well defined and well behaving so we can get rid of unwanted cases, e.g., unbounded norm of generalized gradients; 
%     \item show our potential function has nice properties such as subdifferential regularity. 
% \end{itemize}}

\subsubsection{Polynomial-time Convergence to an Approximate CE}
\label{subsub:poly-convergence}
Next, in Section \ref{sec:poly-time}, we show efficient convergence of relative t\^atonnement for the markets under CES disutility functions with $\rho\in(1,\infty).$  The main challenge in proving (fast) convergence of discrete t\^atonnement dynamics  is linked to the dramatic demand changes that can be triggered by a very small change of price; a hallmark of markets under linear disutilities for example. As such, the signal of the normalized excess demand can be very large, even at a price vector which is very close to a stationary point. This makes the dynamics hard to detect how close a stationary point is.  Even more, with the non-convex property of the potential function $f\vert_{H_B}$, it often breaks the monotonicity of the potential function value, which we usually use to guarantee any non-asymptotic convergence. 

Given that we have already established that relative tâtonnement can be interpreted as a \emph{generalized first-order method} on the potential $f\vert_{H_B}$, an effective approach to proving its fast convergence is to establish the \emph{smoothness} of this potential function. To this end, we focus on the second term $\log\!\left( \max_{\bfx_i \ge 0 : d_i(\bfx_i) \le 1} \langle \bfp, \bfx_i \rangle \right)$. The inner expression,
\[
\max_{\bfx_i \ge 0 : d_i(\bfx_i) \le 1} \langle \bfp, \bfx_i \rangle,
\]
represents the \emph{maximum earning potential} of agent~$i$ at a given price vector~$\bfp$, subject to a unit upper bound on total disutility. From a convex analysis perspective, this quantity can be viewed as the \emph{gauge dual} $d^\circ(\bfp) = \max_{\bfx_i \ge 0 : d_i(\bfx_i) \leq 1} \langle \bfp, \bfx_i \rangle$ of the \emph{primal} disutility function~$d_i(\cdot)$, i.e., the \emph{support function} of the unit disutility ball $\{\bfx_i : d_i(\bfx_i) \le 1\}$.\footnote{Given any convex set $K$, the support function $\sigma_K(\cdot)$ is defined as $\sigma_K(\bf p) = \max_{{\bfx} \in K} \langle \bf p, \bf x \rangle$.}  Hence, our first step is to analyze the \emph{smoothness of the gauge dual} $d^{\circ}(\cdot)$ itself, since smoothness of this dual implies smoothness of the potential function~$f\vert_{H_B}$ after applying additional structural arguments (Theorem~\ref{lem:smooth-f-HB}).

\paragraph{Geometric intuition.} The smoothness of the gauge dual can be understood geometrically via the curvauture of the boundary of the unit disutility ball $\{\bfx_i : d_i(\bfx_i) \leq 1\}$. At any price vector, the gauge dual maps to the point on the boundary, where a hyperplane with the normal vector equal to the price vector forms a supporting hyperplane. The way this contact point moves as prices change depends on the local curvature of the boundary: highly rounded regions produce gently varying maximizers, while flatter regions lead to larger jumps, reflecting reduced smoothness of the dual.
\cref{fig:strongly-convex-set-and-smoothness} illustrates this idea: interpreting the arrows as price vectors, each black dot represents the optimal bundle~$x_i$ that maximizes the agent’s earnings under a disutility constraint. For a CES disutility with $\rho = 2$, the contour is smoothly rounded, so the optimal demand~$\bf x(\bf p)$ varies continuously and gently as prices change—indicating a \emph{smooth} gauge dual. In contrast, when $\rho = 3$, the contour flattens near the axes ($x_1 = 0$ or $x_2 = 0$), implying sharper transitions in the demand correspondence and thus reduced smoothness.

\paragraph{Formal argument.} Typically, smoothness of a dual function is obtained by proving the \emph{strong convexity} of the primal convex function---in our setting, the disutility function~$d_i(\cdot)$. However, CES disutilities fail to satisfy this condition; indeed, they are not even \emph{strictly convex}. 
A significant technical bulk of our work in this section lies in proving that, \emph{despite} this lack of strong convexity, the gauge dual remains sufficiently smooth at least in a local region to ensure polynomial-time convergence of the tâtonnement process. 
We establish this by analyzing two distinct regimes of CES disutilities:
\begin{itemize}
    \item \textbf{Case $\rho \leq 2$:} In this case, 
    we notice that the square of the disutility function is strongly convex (Lemma~\ref{lem:squared_rho=2}), and the corresponding unit sublevel set coincides with that of the original disutility. 
    The strong convexity of this level set then guarantees smoothness properties (Corollary~\ref{coro:rho=2}).
    We establish such a smoothness globally for a broader class of CCH disutilities, including many other interesting disutility functions, e.g., $d_i(\bfx_i) = \norm{A \bfx_i}$ where $A$ is a full rank matrix.
    
    \item \textbf{Case $\rho > 2$:} This case is considerably more subtle. The level sets become increasingly flat along coordinate axes, and strong convexity arguments for the first case no longer apply. 

    To handle this difficulty, we identify a \emph{locally smooth region} of the dual domain-- specifically, the region in which all prices are bounded away from zero (lower bounded), and show that the disutility function is well behaved in this domain (Lemma~\ref{lem:smooth-region}). Our analysis proceeds in two stages. In the first stage, starting from an arbitrary initialization that may lie outside the smooth region, we employ a carefully chosen small constant stepsize to guarantee that the tâtonnement iterates enter the locally smooth region within polynomially many iterations (Lemma~\ref{lem:non-smooth-to-smooth}). In the second stage, once the iterates enter the smooth region, we maintain them there using a discrete-time thresholding argument, ensuring that subsequent updates remain stable (Lemma~\ref{lem:non-smooth-to-smooth}). Together, these arguments establish that the tâtonnement dynamics reach an approximate CE after polynomially many iterations.
     
%We remark that, in handling the above two cases, we provide the first instance to build a connection between the strong-convexity-type properties of disutilities and the smoothness of the guage dual. We believe these connections may be useful tools in designing optimization algorithms for $p$-norm objectives in general.
    
    %In this case, we first establish a ``local smooth region'' where all the prices are strictly lower bounded.
    %Then, we consider a two-phase process: first, equipped with a carefully designed constant small stepsize, 
    %all prices will move into the smooth local region in a polynomial time, even initialized in a nonsmooth region; 
    %second, with a proper ``discrete-time threshold'', 
    %we ensure the prices are controlled within the smooth region and thereby find an approximate CE in a polynomial time.
\end{itemize}

We believe that our approach for analyzing the convergence of first-order methods on nonconvex, nonsmooth CES functions via the gauge dual—and the structural properties established in the process—may be of independent interest for the design and analysis of optimization algorithms for general $p$-norm objectives.

%these general properties we establish may be of independent interest in analyzing guage duals. 
%We remark that, in handling the above two cases, we provide the first instance to build a connection between the strong-convexity-type properties of disutilities and the smoothness of the guage dual. We believe these connections may be useful tools in designing optimization algorithms for $p$-norm objectives in general.

 %in the problems such as market equilibrium computation, for example, by exploiting local smoothness. 

% \RM{Non-smoothness is the hindrance, and there may also be a counterexample where it is not possible to reach a stationary point in finite time. }

% \ck{Point out that the new programs in the previous section act as the potential (after reformulation to a price-only program).}

\begin{figure}
    \centering
    \includegraphics[width=1\linewidth]{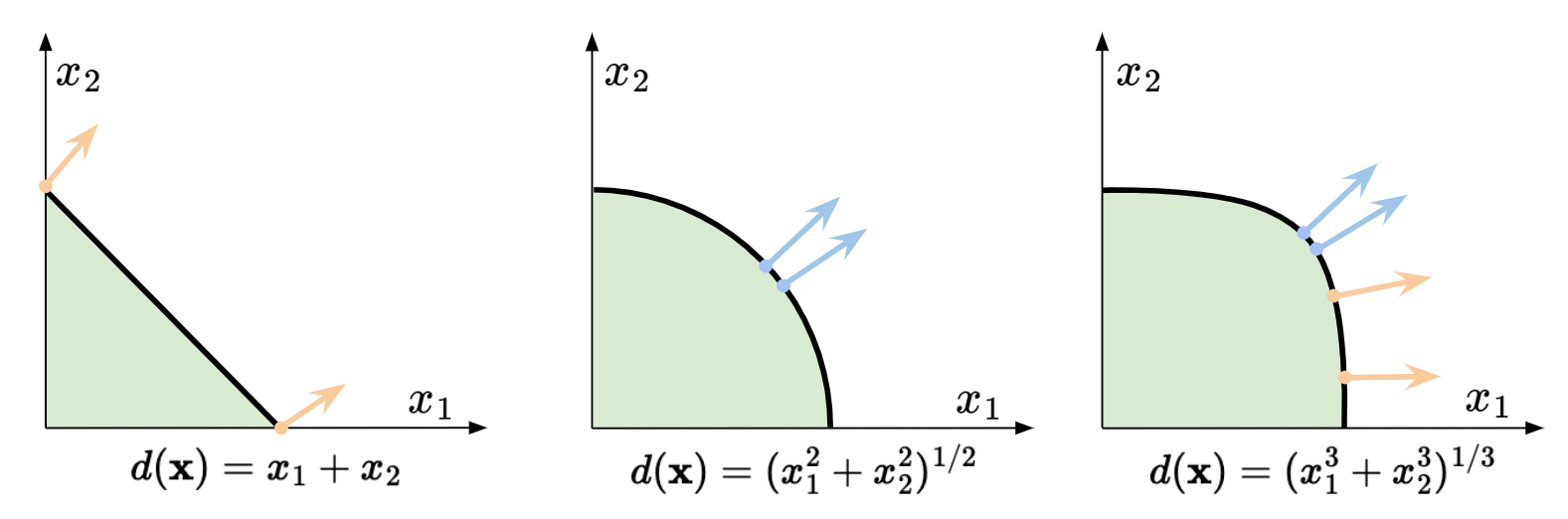}
    \caption{CES disutility functions: $1$-level sets and their smoothness.}
    \label{fig:strongly-convex-set-and-smoothness}
\end{figure}

\subsection{Related Work}
%- As introduced by~\cite{arrow1958stability}, a formal dynamic model whose characteristics reflect the nature of the competitive process and its stability properties are central in the study of competitive equilibrium: 

\textbf{Computation of CE.} Computation of CE has been a fundamental part of economics and computation since its outset. 
Earlier work derived various convex programming and linear complementarity problem (LCP) formulations that capture the set of CE in different economic models~\cite{eisenberg1959consensus,nenakov1983one, eaves1975finite}. \cite{devanur2008market} introduced the first polynomial-time algorithm for linear Fisher markets, using a primal-dual approach. Later, \cite{Orlin10} presented the first strongly polynomial-time algorithm for the same setting. A wide range of algorithms have since been explored for computing CE in more general economic models beyond the linear Fisher market, such as the exchange and Arrow-Debreu markets, including interior point methods~\cite{jain2007polynomial,ye2008path}, 
% ellipsoid algorithms~\cite{ye2008path}, 
and combinatorial algorithms~\cite{duan2015combinatorial}.

The computation of Competitive Equilibria (CE) in the chores market presents a fundamentally different challenge. This difference primarily stems from the disconnectedness of the set of CEs in contrast to the convexity of CEs in the goods market. In their seminal paper, Bogomolnaia, Moulin, Sandomirskiy, and Yanovskaya~\cite{bogomolnaia2017competitive} remark that \emph{``we expect computational difficulties in problems with many agents
and/or items''}. Therefore, the complexity of finding an exact CE in the chores market is still open. However, there has been a line of work that gives arbitrary good approximations of CE in polynomial time in the chores market: ~\cite{BoodaghiansCM22, ChaudhuryGMM22, chaudhury2024competitive}.~\cite{BranzeiS24} provide a polynomial time algorithm to compute a CE in linear fisher markets when there are constantly many agents or chores.

\vspace{0.15cm} 

\textbf{Economic Dynamics in the Goods Market.}
%Economists~\cite{arrow1958stability, samuelson1941stability} have studied intuitive market dynamics in continuous-time. 
There has been substantial study on discrete market dynamics from the EconCS community. \cite{codenotti2005market} proved the convergence of discrete t\^atonnement 
to an approximate equilibrium in an exchange economy satisfying \emph{weak gross substitutability (WGS).} \cite{cole2008fast} considered multiplicative {\tatonnement} with artificially upper-bounded excess demands and showed convergence in non-linear \emph{constant elasticity of substitution} (CES) markets satisfying WGS. 
Thereafter, \cite{cheung2012tatonnement} extended this analysis to some non-WGS markets. 
% \cite{avigdor2014convergence,cheung2012tatonnement} considered multiplicative {\tatonnement} and showed convergence to an approximate equilibrium for Fisher markets with nested non-linear CES-type utilities.
Recently, \cite{nan2024convergence} showed convergence results for additive t\^atonnement in Fisher markets with linear utilities.
There is also research~\cite{cheung2019tatonnement, cole2019balancing,goktas2023tatonnement} studying the convergence of \emph{entropic  t\^atonnement} and its modifications.

Another popular dynamic explored in the goods market is \emph{proportional response dynamics}~\cite{zhang2011proportional}. Here, in every time-step, agents update their \emph{bids} on each good, \emph{proportional} to the utility they derived from the good in the previous round. The goods are then distributed in proportion to the bids, and the per-unit price of a good is set to be the sum of bids on that good. Proportional response dynamics are known to converge to a CE for CES utilities~\cite{zhang2011proportional}.~\cite{birnbaum2011distributed} interpret proportional allocation as mirror descent applied to the convex program for Fisher markets proposed by~\cite{shmyrev2009algorithm}. The full survey of proportional response dynamics is well beyond the scope of the paper and we refer the reader to~\cite{branzei2021proportional} for a detailed survey.

\section{Preliminaries}

% Throughout, we will mostly use standard notation. CK: no need to say this.
Throughout, $\norm{\cdot}$ denotes the Euclidean norm (i.e., $\ell_2$ norm) and $\inp*{\cdot}{\cdot}$ denotes the Euclidean inner product. 
$\mathbb{R}^m_+$ denotes the $m$-dimensional nonnegative orthant and $\mathbb{R}^m_{++}$ denotes the $m$-dimensional positive orthant. 
We use $\textnormal{argmin}$ to denote the set of minimizers. The convex hull is denoted as $\textnormal{Conv}(\cdot)$. We use $[m]$ to denote $\{ 1, \ldots, m \}$. 
For conciseness, we sometimes omit the index set (e.g., $[m]$) in the sum if the sum is taken over the full index set. 
We use $\mathcal{B}_r(x)$ to denote the $\ell_2$ ball centered at $x$ with radius $r$. 
We use $\mathbf{1}_d$ and $\mathbf{0}_d$ to denote $d$-dimensional all-one and all-zero vectors, respectively. 
We use $\lvert \cdot \rvert$ to denote the vector of component-wise absolute values of a vector.

\subsection{Fisher Market with Chores} 

A chores Fisher market consists of $n$ agents $m$ \emph{divisible} chores. Without loss of generality, the supply (available quantity) of each chore is assumed to be one.
In this market, each agent $i$ needs to \emph{earn} a ``budget'' of $B_i > 0$ (which we also refer to as an \emph{earning requirement}). 
Without loss of generality, each chore $j$ is assumed to have unit supply.
The prices for the chores are denoted by a price vector $\bfp \in \mathbb{R}^m_+$.
The disutility of agent $i$ is a convex function $d_i: \mathbb{R}^m_+ \rightarrow \mathbb{R}_{++}$ of her allocated bundle of chores $\bfx_i = (x_{i1}, \ldots, x_{im}) \in \mathbb{R}^m_+$. 
\medskip

\noindent{\em Disutility Functions.} In this paper, we assume that each disutility function $d_i: \mathbb{R}^m_+ \rightarrow \mathbb{R}$ is continuous, convex, and homogeneous  degree one (CCH). Additionally, we will require the following form of strict non-satiation for each coordinate of the disutility function:
\begin{assumption}
    The disutility functions are strictly increasing in each coordinate.
    \label{assump:strictly-increasing}
\end{assumption}
% \begin{assumption}
%     The disutility functions satisfy
%     (1) they are increasing in each coordinate, 
%     (2) for any $\bfx_i \geq 0$, they are strictly increasing in at least one coordinate and 
%     (3) for any $\bfx_i \geq 0$ there exists $C \geq 0$ such that $d_i(\bfx_i)$ is strictly increasing in any coordinate $j$ if $x_{ij} \geq C$.
% \end{assumption}

% \begin{assumption}
%     The disutility functions are strictly non-satiating if for each $d_i: \mathbb{R}^m_+ \rightarrow \mathbb{R}$, 
%     (1) it is increasing in each coordinate, 
%     (2) for any $\bfx_i \geq 0$, it is strictly increasing in at least one coordinate and 
%     (3) for any $\bfx_i \geq 0$ there exists $C \geq 0$ such that $d_i(\bfx_i)$ is strictly increasing in any coordinate $j$ if $x_{ij} \geq C$.
% \end{assumption}

The most common CCH disutility functions are \emph{convex CES disutility} functions: 
\begin{equation*}
    d_i(\bfx_i) = \left( \sum\nolimits_{j = 1}^m d_{ij} x_{ij}^\rho \right)^{\frac{1}{\rho}}, 
\end{equation*}
where $d_{ij} > 0$ for all $i \in [n], j \in [m]$, and $\rho \in [1, \infty)$ is a parameter characterizing the agent's preference curvature.\footnote{If $d_{ij} = 0$ for some $i$, chore $j$ can be allocated to agent $i$ without affecting her disutility. In this case, the chore $j$ can be removed from the market.}
% The following three examples give the most interesting cases of convex CES disutilities:
% \begin{itemize}
%     \item Linear disutility: $d_i(\bfx_i) = \sum\nolimits_{j = 1}^m d_{ij} x_{ij}$;
%     \item  $\ell_2$ disutility: $d_i(\bfx_i) = \sqrt{\sum\nolimits_{j = 1}^m d_{ij}^2 x_{ij}^2}$; 
%     \item Maximum-burden disutility: $d_i(\bfx_i) = \max_{j \in [m]}\{ d_{ij} x_{ij} \}$. 
% \end{itemize}

In~\cref{fig:ces}, we plot iso-disutility curves for three convex CES functions with different parameter $\rho$. These plots illustrate that each function captures a different behavioral pattern in how disutility arises from bundles of chores.
\begin{figure}
    \centering
    \includegraphics[width=0.325\linewidth]{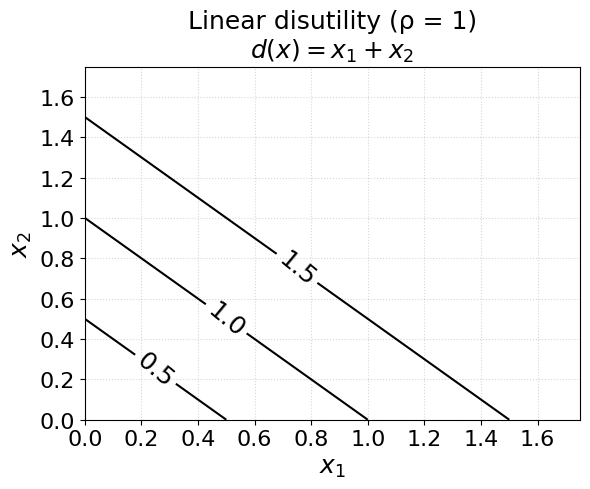}
    \includegraphics[width=0.325\linewidth]{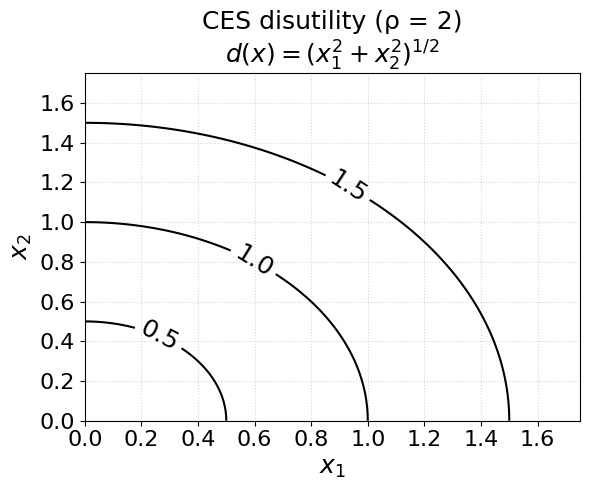}
    \includegraphics[width=0.325\linewidth]{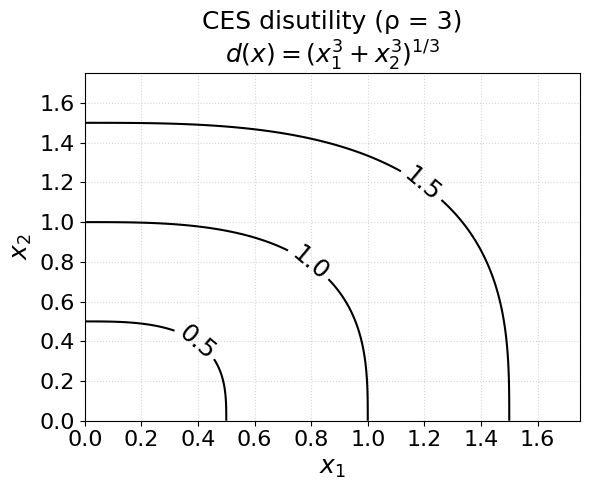}
    
    \caption{Iso-disutility curves for convex CES disutility functions.}
    % \tianlong{This can also go to the introduction}}
    \label{fig:ces}
\end{figure}

% A chores linear Fisher market consists of $n$ agents and $m$ divisible chores. In this market, each agent $i$ needs to earn a ``budget'' (earning requirement) of $B_i$. Without loss of generality, each chore $j$ is set to be of unit supply. The disutility of agent $i$ is a linear function $d_i: \mathbb{R}^m_+ \rightarrow \mathbb{R}_{++}$ of her allocated bundle of chores $x_i = (x_{i1}, \ldots, x_{im}) \in \mathbb{R}^m_+$: 
% % \RM{change $u_i$ to $d_i$ everywhere.}
% \begin{equation*}
%     d_i(x_i) = \sum\nolimits_{j = 1}^m d_{ij} x_{ij}, 
% \end{equation*}
% where $d_{ij} > 0$ for all $i, j$\footnote{If $d_{ij} = 0$ for some $i$, chore $j$ can be allocated to agent $i$ without affecting her disutility. Then, the chore $j$ can be removed from the market.}. The prices of chores are denoted by a price vector $p \in \mathbb{R}^m_+$. 
\medskip

\noindent{\em Competitive Equilibrium.}
At any given prices $\bfp$, agents want to minimize their disutility subject to their earnings requirement. 
Agent $i$'s \emph{demand correspondence} $X_i: \mathbb{R}^m_+ \rightrightarrows \mathbb{R}^m_+$ is thus defined by 
\begin{equation*}
    X_i(\bfp) := \textnormal{argmin}_{\bfx_i \in \mathbb{R}^m_+} \left\{ d_i(\bfx_i) \left| \, \inp*{\bfp}{\bfx_i} \geq B_i
    \right. \right\}. 
\end{equation*}

% For the maximum-burden disutility ($\rho = \infty$), multiplicity can arise when some chore is priced at zero. A natural tie-breaking selection is the bundle proportional to $( \frac{1}{d_{i1}}, \frac{1}{d_{i2}}, 
% \ldots, \frac{1}{d_{im}})$; it matches the unique demand for $\bfp\in\RR_{++}^m$.

As state in~\cref{sec:intro}, we say a pair of prices and allocation $(\bfp^*, \bfx^*)$ is a competitive equilibrium if it satisfies 
\begin{enumerate}[label=(\roman*)]
    \item \emph{optimal bundles}: $\bfx^*_i \in X_i(\bfp^*)$ for each agent $i \in [n]$ and 
    \item \emph{market clearance}: $\sum_{i=1}^n x^*_{ij} = 1$ for each chore $j \in [m]$.
\end{enumerate}
\medskip

\noindent{\em CES Demand.} For convex CES functions, the demand correspondence may be single- or multi-valued, depending on $\rho$. 
In the linear case ($\rho=1$), an agent may have multiple optimal demand bundles at a given price vector. 
In particular, she selects a convex combination of the chores with \emph{minimum pain-per-buck} $\min_{j \in [m]} d_{ij} / p_j$.
For $\rho \in (1, \infty)$, each agent has a unique demand given any $\bfp \in \RR^m_+ \setminus \{ \mathbf{0}_m \}$. 
\medskip

\noindent{\em Excess Demand.} A crucial quantity for studying market dynamics is the \emph{excess demand} correspondence $Z: \mathbb{R}^m_+ \rightrightarrows \mathbb{R}^m_+$: 
\begin{equation*}
    Z(\bfp) := \left\{ \sum_{i = 1}^n \bfx_{i} - \mathbf{1}_m \, \Big| \, \bfx_{i} \in X_i(\bfp) \; \forall\, i \in [n] \right\}.  
\end{equation*} 
It can be seen that $\bfp^*$ is an equilibrium price if and only if $\bm{0}_m \in Z(\bfp^*)$.

\subsection{Economic Dynamical System, Stability and T\^atonnement} 

In this work, we study t\^atonnement as economic dynamical systems in the context of chores. In this section, we start with the definitions analogous to those introduced by the classical work of ~\cite{arrow1958stability} for the goods case, which are for continuous-time dynamics. The extension to discrete-time dynamics is defined and studied in Section \ref{sec:discrete}.

An economic dynamical system describes a price adjustment process which can be expressed as the following \emph{differential inclusion}: 
\begin{equation}
    \frac{\partial \bfp}{\partial t} \in D(\bfp). 
    \label{econ-differential-inclusion-vec}
\end{equation} 
where $D: \mathbb{R}^m_+ \rightrightarrows \mathbb{R}^m$ is some \emph{set-valued mapping}.

A \emph{price arc} $\bfp: \mathbb{R}_+ \rightarrow \mathbb{R}^m$
is a function (mapping time $t$ to a price vector) that is \emph{absolutely continuous} if 
% , i.e., there exists a map $\bfy: \mathbb{R}_+ \rightarrow \mathbb{R}^m$ that is integrable on any compact interval that satisfies 
\begin{equation*}
    \bfp(t) = \bfp(0) + \int_0^t \dot{\bfp}(s) ds \quad \forall\; t \geq 0, 
\end{equation*}
where $\dot{\bfp}(t) = d \bfp(t) / d t$ for almost all $t \geq 0$. 
A price arc $\bfp(t)$\footnote{$\bfp(t)$ is also called a \emph{solution} to the system of the differential inclusion.} is called a continuous-time \emph{trajectory}/\emph{path} of~\cref{econ-differential-inclusion-vec} with $\bfp^0$ as the \emph{initial point} if 
\begin{equation}
    \dot{\bfp}(t) \in D(\bfp(t)) \quad \text{for almost all $t \geq 0$} \quad \text{and} \quad \bfp(0) = \bfp^0. 
    \label{econ-differential-inclusion-solution-trajectory}
\end{equation}
A point $\bfp^* \in \mathbb{R}^m_+$ is called a \emph{stationary point} of~\cref{econ-differential-inclusion-vec} if 
$\bm{0}_m \in D(\bfp^*)$. 

Analogous to the definitions in~\cite{arrow1958stability}, we define the stability of chores CE based on an economic dynamical system. 
Before introducing the concepts of convergence and stability, we first define the affine hull $H_B$ where the sum of all prices is constrained to equal the total budgets in the market, i.e., 
\begin{equation*}
    \Hb = \left\{ \bfp \in \mathbb{R}^m \; \left\vert \; (\mathbf{1}_m)^\top \bfp = \lVert \bfB \rVert_1  \right. \right\}. 
\end{equation*}
We also define $\Delta_B = H_B \cap \RR^m_+$, which we refer to as the \emph{price simplex}.
A fact for Fisher markets is that equilibrium prices always lie in $\Delta_B$. 
In the chores case, we will consider a slightly weaker stability notion based on $\Delta_B$, compared to analogous definitions in~\cite{arrow1958stability}. 
In particular, we consider any trajectory of the dynamical system starting from any initial point within $\Delta_B$, instead of any price in the nonnegative orthant. 

\begin{definition}[Local stability of chores CE]
\label{def:stability}
    An equilibrium price $\bfp^*$ is \emph{locally stable} w.r.t. dynamics~\cref{econ-differential-inclusion-vec} if there exists a non-empty neighborhood $N(\bfp^*) \subset \mathbb{R}^m_+$ of $\bfp^*$ such that for \emph{any} initial point $\bfp(0) \in N(\bfp^*) \cap \Delta_B$ every trajectory of~\cref{econ-differential-inclusion-vec} converges to $\bfp^*$, i.e., 
    $
        \lim_{t \rightarrow \infty} \bfp(t) \rightarrow \bfp^*. 
    $
\end{definition} 

An equilibrium price is said to be locally unstable if it is not locally stable. 
Recall that one of the main characteristics of chores CE is that there may exist multiple disconnected equilibria. 
As such, we should not expect for there to be a single equilibrium that all prices will converge to, but the dynamical system as a whole may still be endowed with the following important stability property. 

\begin{definition}[Stability of dynamical system]
    A dynamical system~\cref{econ-differential-inclusion-vec} is \emph{stable} if 
    for every $\bfp(0) \in \Delta_B$ and every trajectory of~\cref{econ-differential-inclusion-vec}, there is some stationary point $\bfp^*$ such that $\lim_{t \rightarrow \infty} \bfp(t) \rightarrow \bfp^*$. 
\end{definition}
In short, a dynamical system is said to be stable if the continuous-time dynamics converge for \emph{any} initial point.

\medskip

\noindent{\bf T\^atonnement.} The most natural choice for a price adjustment process is setting $D$ in~\cref{econ-differential-inclusion-vec} proportional to the \emph{(aggregate) excess demand} $Z(\bfp)$. 
This type of price adjustment process is commonly known as t\^atonnement. 

In chores markets, items become more desired by agents as prices increase. 
Correspondingly, we shall decrease (resp. increase) the price of an item if there is positive excess demand (resp. negative excess demand or excess supply). 
Thus, a natural economic dynamics is 
\begin{equation}
    \frac{d \bfp}{d t} \in -Z(\bfp). 
    \label{naive-chores-tatonnement}
    \tag{Na\"ive t\^atonnement}
\end{equation}
We refer to these dynamics as \emph{na\"ive t\^atonnement}. 
In words, the dynamics simply decrease prices proportional to excess demand, and vice versa for excess supply in an additive manner. 
Because agents are indifferent among their optimal bundles given a price vector, we will sometimes specify the dynamics for a specific tie-breaking rule.
By definition, it is direct to show that any stationary point of~\eqref{naive-chores-tatonnement} corresponds to a competitive equilibrium.

\subsection{Nonsmooth Analysis Basics}
As mentioned in~\cref{sec:intro}, our analysis requires the use of some basic tools from nonsmooth analysis. 

A function $f: \mathbb{R}^n \rightarrow \mathbb{R}$ is called \emph{locally Lipschitz continuous} (or \emph{locally Lipschitz}) if for any compact set $U \subset \mathbb{R}^n$ there exists a (local) Lipschitz constant $L_U$ such that 
\begin{equation*}
    \lvert f(\bfx) - f(\bfy) \rvert \leq L_U \lVert \bfx - \bfy \rVert \quad \text{for any } \bfx, \bfy \in U. 
\end{equation*}
Rademacher's theorem guarantees that a locally Lipschitz function is differentiable \emph{almost everywhere}. 
For a locally Lipschitz function $f$, the \emph{Clarke subdifferential} of $f$ at any point $\bfx$ is 
\begin{equation*}
    \partial f(\bfx) := \textnormal{Conv}\left( \left\{ \lim_{i \rightarrow \infty} \nabla f(\bfx_i): \bfx_i \rightarrow \bfx, \bfx_i \notin \Omega, \bfx_i \notin \Omega_f \right\} \right), 
\end{equation*}
where $\Omega$ is an arbitrary measure-zero subset of $\mathbb{R}^n$, and $\Omega_f$ is the set of points in $\mathbb{R}^n$ at which $f$ is not differentiable~\cite[Theorem 8.1]{clarke2008nonsmooth}. The clarke subdifferential is known to be a nonempty, compact, convex set for each $\bfx \in \mathbb{R}^n$. 
An element in $\partial f(\bfx)$ is called a \emph{Clarke generalized gradient} of $f$ at $\bfx$. 

Next we define \emph{subdifferentially regular} functions. We say a locally Lipschitz function $f$ is \emph{subdifferentially regular} at a point $\bfx$ if every Clarke generalized gradient $\mathbf{g} \in \partial f(\bfx)$ yields an affine minorant of $f$ up to first-order~\cite[Definition 5.3]{davis2020stochastic}: 
\begin{equation*}
    f(\bfy) \geq f(\bfx) + \inp*{\mathbf{g}}{\bfy - \bfx} + o\left( \norm{\bfy - \bfx} \right) \quad \text{as } \bfy \rightarrow \bfx. 
\end{equation*} 
A function $f$ is said to be subdifferentially regular if it satisfies subdifferential regularity at every point in its domain. 
A nice property of a subdifferentially regular function is that it does not contain \emph{downward facing cusps} in the graph of the function; Figure \ref{fig:subdifferential-regularity} provides a pictorial explanation of downward (and upward) facing cusps, and (non) regularity. As simple examples, $f(x) = \lvert x \rvert$ is subdifferentially regular and $f(x) = -\lvert x \rvert$ is not. 
The subdifferentially regular functions enjoy many attractive properties~\cite{li2020understanding}.

\section{Relative t\^atonnement}  

% As we mentioned, even though recent boosted understanding towards CE with chores, the convergence of classic {\em tatonnement} dynamics and ``stability'' of CE with chores has not been studied. In this section, we initialize some basic definitions and establish the basis for studying those dynamic properties of chores CE. 

% \subsection{Economic dynamics}

% \subsection{Stability of chores CE} 

In this section, we first show that the typical, and most natural, t\^atonnement dynamics given in (\ref{naive-chores-tatonnement}) does not converge, in a strong sense, for the case of chores markets. We then propose new (still additive) t\^atonnement dynamics by replacing the excess demand with the relative excess demand, and show that this new dynamics converge, point-wise, to a competitive equilibrium. 
In the process, we obtain a new program whose KKT points correspond to chores CE for general CCH disutility functions, which is a result of broader interest.

\subsection{Non-convergence of Na\"ive t\^atonnement}
\label{sec: naive unstable}

We first study the continuous-time na\"ive t\^atonnement dynamics. 
The following lemma shows that the na\"ive t\^atonnement dynamics can diverge away from the set of CE, no matter which tie-breaking rule we use. 
In the subsequent section, we prove that the discrete-time version of the na\"ive t\^atonnement inherits this divergence behavior (see~\cref{lem:discrete-time-naive-ttm-divergence}).

% \RM{The following lemma is for the discrete tatonnement, which is not defined yet! Please move this to section 4, and rewrite it for the continuous tatonnement, which is defined in section 2.2.}

% \RM{In the following lemma, do we really need $p^0 \in \Delta_B$? Does the lemma hold for $p^0 \in \RR^m_+ \setminus \{p^*\}$?}
\begin{lemma}
    For any $\rho \in [1, \infty)$, 
    there is a $\rho$-CES chores Fisher market with a unique CE price vector $\bfp^*$, 
    such that for every initial price vector $\bfp^0 \in \{ \bfp \in \RR^m_+ \mid \sum_{j=1}^m p_j \geq \sum_{i=1}^n B_i, \bfp \neq \bfp^* \}$, the continuous-time na\"ive t\^atonnement dynamics equipped with \emph{any} tie-breaking rule diverges from $\bfp^*$.
    \label{lem:continuous-time-naive-ttm-divergence}
\end{lemma}

\begin{proof}
    Consider an instance with a single agent and 2 chores. 
    The agent's disutility is a convex CES function of the two chores: 
    $d_1(\bfx_1) = (x_{11}^\rho + x_{12}^\rho)^{1/\rho}$.
    The agent has an earning requirement of $B_1 = 1$.
    It is easy to verify that this market has a unique CE with prices $\bfp^* = (\frac{1}{2}, \frac{1}{2})$ and allocation $\bfx^*_1 = (1, 1)$.
    We will show that the sum of prices diverges to $+\infty$.

    We first consider $\rho = 1$, i.e., a linear chores Fisher market. 
    Consider a price vector $\bfp \in \RR^2_+$.
    The excess demand mapping at $\bfp$ is 
    \begin{equation*}
        \begin{aligned}
            Z(\bfp) = \begin{cases}
            \{ \big( \frac{1}{p_1} - 1, -1 \big) \} & p_1 > p_2 \geq 0 \\ 
            \{ \lambda \big( \frac{1}{p} - 1, - 1 \big) + (1 - \lambda) \big( - 1, \frac{1}{p} - 1 \big) \mid \lambda \in [0, 1] \} & p_1 = p_2 = p > 0 \\ 
            \{ \big( - 1, \frac{1}{p_2} - 1 \big) \} & p_2 > p_1 \geq 0. 
        \end{cases}
        \end{aligned}
    \end{equation*}

Then, one can show that the sum of prices follows a dynamical system as follows: 
\begin{equation*}
    \frac{d}{d t} (p_1 + p_2) = 2 - \frac{1}{\max\{ p_1, p_2 \}}. 
\end{equation*}

Note that when $p_1 = p_2$, $\frac{d}{d t}(p_1 + p_2)$ turns out to be invariant under the choice of tie-breaking rule.
Conditioning on $p_1 + p_2 \geq 1$, we have $\max\{ p_1, p_2 \} \geq \frac{1}{2}\left( \max\{ p_1, p_2 \} + \min\{ p_1, p_2 \} \right) = \frac{1}{2}\left( p_1 + p_2 \right) \geq \frac{1}{2}$, and the equality holds only if $p_1 = p_2 = \frac{1}{2}$, i.e., $\bfp = \bfp^*$.
Hence, 
\begin{equation*}
    \frac{d}{d t} (p_1 + p_2) > 0 \hspace{30pt} \forall\; \bfp \in \{ \bfp \in \RR^2_+ \,|\, p_1 + p_2 \geq 1, \bfp \neq \bfp^* \}.
\end{equation*}
This implies that the continuous-time na\"ive t\^atonnement dynamics diverge away from $\bfp^*$.

We now consider $1 < \rho < \infty$. 
In this case, the excess demand set is a singleton for any $\bfp \in \RR^2_{+}$. By calculation, we have that the prices follow a dynamical system as follows: 
% at $\bfp \in \RR^2_{+}$ are 
\begin{equation*}
    \frac{d}{d t} (p_1) = 1 - \frac{p_1^{1/(\rho - 1)}}{p_1^{\rho / (\rho - 1)} + p_2^{\rho / (\rho - 1)}}, \hspace{10pt} \frac{d}{d t} (p_2) = 1 - \frac{p_2^{1/(\rho - 1)}}{p_1^{\rho / (\rho - 1)} + p_2^{\rho / (\rho - 1)}}
\end{equation*}
and hence 
\begin{equation*}
    \frac{d}{d t}(p_1 + p_2) = 2 - \frac{p_1^{1/(\rho - 1)} + p_2^{1/(\rho - 1)}}{p_1^{\rho / (\rho - 1)} + p_2^{\rho / (\rho - 1)}}. 
\end{equation*}

Next, we show that 
\begin{equation*}
    \frac{p_1^{1/(\rho - 1)} + p_2^{1/(\rho - 1)}}{p_1^{\rho / (\rho - 1)} + p_2^{\rho / (\rho - 1)}} < 2 \hspace{30pt} \forall\; \bfp \in \{ \bfp \in \RR^2_+ \,|\, p_1 + p_2 \geq 1, \bfp \neq \bfp^* \}, 
\end{equation*}
which is equivalent to 
\begin{equation}
    (2 p_1 - 1) p_1^{1/(\rho - 1)} + (2 p_2 - 1) p_2^{1/(\rho - 1)} > 0 \hspace{30pt} \forall\; \bfp \in \{ \bfp \in \RR^2_+ \,|\, p_1 + p_2 \geq 1, \bfp \neq \bfp^* \},  
    \label{inequality-to-show-continuous-time}
\end{equation}
To show~\cref{inequality-to-show-continuous-time}, we discuss the three cases:
(1) if $2p_1 - 1 \geq 0$ and $2p_2 - 1 \geq 0$, then \cref{inequality-to-show-continuous-time} trivially holds. 
Otherwise, (2) if $2p_1 - 1 < 0$, then we have $2p_2 - 1 > 0$ and $| 2p_2 - 1 | \geq | 2p_1 - 1 |$ because $2p_1 - 1 + 2p_2 - 1 \geq 0$. 
Also, we have $p_2^{1/(\rho - 1)} > p_1^{1/(\rho - 1)} \geq 0$ because $p_2 > p_1 \geq 0$.
Hence, \cref{inequality-to-show-continuous-time} holds in this case;  
(3) if $2p_2 - 1 < 0$, \cref{inequality-to-show-continuous-time} also holds by a symmetric argument. 
Therefore, $\frac{d}{d t}(p_1 + p_2) > 0$ for all $\bfp \in \{ \bfq \in \RR^2_+ \, | \, q_1 + q_2 \geq 1, \bfq \neq \bfp^* \}$. 
Therefore, similar to the case of $\rho = 1$, 
the continuous-time na\"ive t\^atonnement dynamics diverge away from $\bfp^*$.
\end{proof}

\begin{figure}[t]
    \centering
    \includegraphics[width=0.7\linewidth]{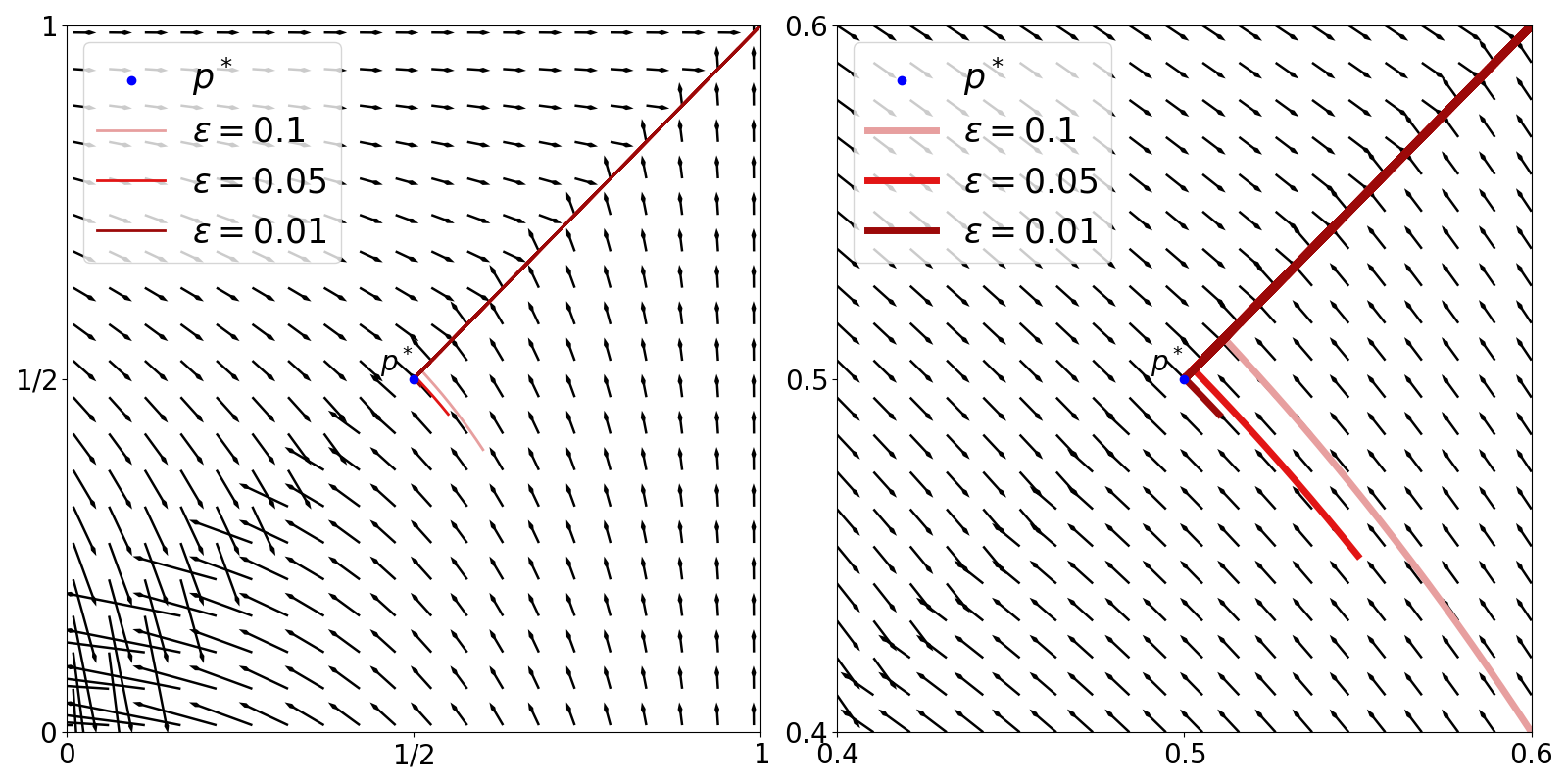}
    \caption{
    % \tianlong{I think we can remove the diagrams for non-convergence of naive t\^atonnement as reviewers can ask about what happens initialized from an alternative point, etc.} 
    A trajectory of na\"ive t\^atonnement starting from a neighborhood of a CE. In the right plot, we zoom into the neighborhood. This shows na\"ive t\^atonnement is unstable in the linear chores market. }
    \label{fig:naive-chores-t\^atonnement}
\end{figure}
In Section \ref{sec:discrete} (Lemma \ref{lem:discrete-time-naive-ttm-divergence}) we show a similar (strong) divergence result for the discrete-time counterpart of \eqref{naive-chores-tatonnement}. Note that, Lemmas \ref{lem:continuous-time-naive-ttm-divergence} and \ref{lem:discrete-time-naive-ttm-divergence} imply divergence of \eqref{naive-chores-tatonnement} in a strong sense:  divergence of both continuous and discrete-time dynamics, no matter the step-sizes, no matter the starting point, and no matter if point-wise or time-average ({\em ergodic}).
% } 
These results show that, unlike the case of Fisher markets with goods, t\^atonnement dynamics are unstable in the chores case. 
This difference can be interpreted as follows. 
In the goods case, a lower price leads to higher excess demand, as buyers end up buying more of the goods. 
As such, the excess demand serves as a ``correct'' indicator for the direction of price adjustment in the goods case. 
However, for chores a lower price can lead to either a higher or lower excess demand, depending on the situation. 
On the one hand, lowering the price may cause an agent to want to perform a different chore, thus decreasing demand as expected. On the other hand, if the price is lowered, but the buyer still prefers to perform the chore, then their demand will go \emph{up}, because now they need to do more of the chore in order to satisfy their earning requirement.
It follows from the above that excess demand increases can be triggered by both price decreases and price increases. 
As a result, the pure excess demand might be a ``deceptive'' signal for the chores t\^atonnement process. 
% A high price could possibly decrease the excess demand of a chore and thus trigger even more price increase. 
% In particular, if the price of one chore is greater than $\sum_i B_i$, then its excess demand will stay negative and its price will keep increasing. 
% This difference is crucial for t\^atonnement, as the pure excess demand might be a ``deceptive'' signal. 

Another insight from our examples is that the excess demand cannot calibrate prices if the price vector deviates from the budget-constrained price space $H_B$. Let $\bfp^*$ be any CE in an arbitrary chores market instance. 
Let us consider a slightly inflated price $\bfp^\uparrow = (1 + \delta) \bfp^*$ for some $\delta > 0$. 
Then, one possible excess demand vector at $\bfp^\uparrow$ (there may be others due to tie-breaking) is the vector where every entry is $-\frac{\delta}{1 + \delta}$, which leads to a price-adjustment direction pointing away from $\bfp^*$. 
Similarly, the excess demand of the prices $\bfp^\downarrow = (1 - \delta) \bfp^*$ can indicate the wrong direction as well. 
This explains why we need the initial price to be within the linear manifold $H_B$. 
However, as shown by the above instance, even starting from a $\bfp^0 \in H_B$, the trajectory of na\"ive chores t\^atonnement can deviate from $H_B$. 
In the next section, we introduce a modified excess demand function which will enable t\^atonnement to avoid this issue. 
% Let us consider two extreme cases. First, if prices for all chores are less than 

\subsection{Convergence of Relative T\^atonnement} 
\label{subsec:Relative chores t\^atonnement is stable}

As shown above, na\"ive chores t\^atonnement is not stable and generally does not converge to CE in chores markets. 
We now show that it is possible to find a new dynamical system that is stable, while still being a fairly natural price-adjustment process for a chores market.

Our new t\^atonnement process will be based on the \emph{relative excess demand} $\tilde{Z}(\bfp): \mathbb{R}^m_+ \rightrightarrows \mathbb{R}^m$ defined as 
\begin{equation*}
    \tilde{Z}(\bfp) := \left\{ \bfz - \frac{1}{m} \mathbf{1}_m^\top \bfz \; \Big\vert \; \bfz \in Z(\bfp) \right\}.   
\end{equation*} 
The benefit of employing this notion as a price-adjustment indicator is clear: the net excess demand of a chore only reflects the relative level of the excess demand compared to the averaged excess demand of the other chores. This will make the scenario in \cref{sec: naive unstable}, where both chores are overdemanded yet the price goes up, impossible. 
Mathematically, this ensures that the price-adjustment signal lies in the subspace parallel to $\Hb$, thereby ensuring that if the prices are already in $\Hb$ then they remain in $\Hb$ after the price update.

By replacing $Z(\bfp)$ in~(\ref{naive-chores-tatonnement}) with $\tilde{Z}(\bfp)$, we have a new dynamical system: 
\begin{equation}
    \frac{d \bfp}{d t} \in -\tilde{Z}(\bfp). 
    \label{chores-tatonnement}
    \tag{Relative t\^atonnement}
\end{equation}
We call these dynamics (relative) chores t\^atonnement. 
We next show that this dynamical system captures chores CE as its stationary points and is stable in the sense of \cref{def:stability}. 

First, we show that any stationary point of \textnormal{(\ref{chores-tatonnement})} that lies in the affine hull $\Hb$ corresponds to a chores CE. 
\begin{lemma}
    Any price vector $\bfp^* \in \Hb$ which is a stationary point of \textnormal{(\ref{chores-tatonnement})} corresponds to a chores CE. Conversely, any equilibrium price vector in a chores CE is a stationary point of \textnormal{(\ref{chores-tatonnement})} in $\Hb$.
    \label{key-lem:stationary-point-in-affine-hull-is-CE}
\end{lemma} 
\begin{proof}[Proof of~\cref{key-lem:stationary-point-in-affine-hull-is-CE}]
    If $\bfp^*$ corresponds to a CE, then we have $\mathbf{0}_m \in Z(\bfp^*)$ and thus $\mathbf{0}_m \in \tilde{Z}(\bfp^*)$. 
    Conversely, if $\mathbf{0}_m \in \tilde{Z}(\bfp^*)$ and $\bfp^* \in \Hb$, 
    then there exists $\bfz = (z_1, z_2, \ldots, z_m)$ such that $z_j - \frac{1}{m} \mathbf{1}_m^\top \bfz = 0$ for all $j \in [m]$, which implies $z_1 = \cdots = z_m = c$ where $c$ is a constant. 
    By the definition of excess demand and agent optimality, there exist nonnegative $\{ b_{ij} \}_{i \in [n], j \in [m]}$ such that 
    $\sum_{j = 1}^m b_{ij} = B_i \; \forall\, i \in [n]$ and 
    $\sum_{i=1}^n b_{ij} = (1 + c) p^*_j \; \forall\, j: p^*_j > 0$. 
    The latter follows from rearranging $\frac{\sum_{i=1}^n b_{ij}}{p_j^*} - 1 = z_j = c \; \forall\, j$ for $j$ such that $p^*_j > 0$.
    Also, $b_{ij} = 0$ for all $j$ implies $p^*_j = 0$ for all $i \in [n]$. 
    This leads to 
    \begin{equation*}
        \sum_{i=1}^n B_i = \sum_{i=1}^n \sum_{j=1}^m b_{ij} = \sum_{j: p^*_j > 0} \sum_{i=1}^n b_{ij} = (1 + c) \sum_{j: p^*_j > 0} p^*_j = (1 + c) \sum_{j=1}^m p^*_j. 
    \end{equation*}
    Since $\bfp^* \in H_B$, we must have $c = 0$. 
    Hence, we have $\bm{0}_m \in Z(\bfp^*)$, 
    which means that $\bfp^*$ is an equilibrium price vector. 
\end{proof}

% Since $\tilde{Z}(\bfp)$ is well-defined\footnote{We say an excess demand set is well defined if it contains at least one finite element.} only when $\bfp \in \mathbb{R}^m_+$ and $\max_{j \in [m]} p_j > 0$, first we need to verify that the dynamical system \textnormal{(\ref{chores-tatonnement})} is well-defined. 
Next, we show that any \textnormal{(\ref{chores-tatonnement})} trajectory that is initialized from $\Delta_B$ stays within $H_B$.
In fact, for any continuous-time trajectory, we can ensure that it stays within $\Delta_B$.
To show this result, 
we define the following useful modulus for each agent $i$: 
\begin{equation}
    \nu_i := \min \left\{\frac{d'_{ij}}{d'_{ij'}}\ \bigg|\ {j, j' \in [m]},  \tilde{x}_{ij} \geq  \frac{1}{2 n m}, \tilde{x}_{ij'} \leq \frac{2m B_i}{\norm*{\bfB}_1},  \bfd'_i \in \partial d_i(\tilde{\bfx}_i) \right\}
    % \nu_i := \min_{j, j' \in [m]} \min_{\tilde{\bfx}_i \geq 0: \tilde{x}_{ij} \geq \frac{1}{2 n m}, \tilde{x}_{ij'} \leq \frac{2m B_i}{\norm*{\bfB}_1}} \min_{ \bfd'_i \in \partial d_i(\tilde{\bfx}_i)} \left( \frac{d'_{ij}}{d'_{ij'}} \right), 
    \label{eq:nu defn}
\end{equation}
where $\partial d_i(\tilde{\bfx}_i)$ is the subdifferential of $d_i$ at $\tilde{\bfx}_i$.
In short, this modulus $\nu_i$ measures the minimum ratio of the marginal disutilities across chores, when the agent's allocation for each chore is lower and upper bounded. 
% We show these facts in the following lemma. 
% We defer its proof to~\cref{app:subsec:proof:lem:well-defined-equilibrium}.  

% We make the following mild assumption on the excess demand separately to facilitate the proof. 
% \tianlong{
% This assumption is satisfied by a large range of general disutility functions. 
% For CES disutility, we should be able to compute the closed-form characteristics for this property. 
% I will replace this assumption with an assumption on the disutility functions soon.}
% \begin{assumption}
%     $\max_{\zeta \in z(p)} \zeta_j \rightarrow -1$ if $p \in \Delta_B$ and $p_j \rightarrow 0$. 
%     \label{lem:property-no-demand-for-zero-payment}
% \end{assumption}

% \begin{proposition}
%     If $d_i$ is a CCH disutility function for each $i \in [n]$, 
%     then scaling up all marginal disutilities $\partial d_i(\bfx_i) / \partial \bfx_i$ by the same factor does not affect the set of CE.
% \end{proposition}

For each agent $i$,  $\nu_i > 0$ is guaranteed if $d_i$ satisfies~\cref{assump:strictly-increasing}. 
For example, 
for linear disutilities $d_i(\bfx_i) = \sum_{j=1}^m d_{ij} x_{ij}$, $\partial d_i(\bfx_i)$ is independent of $\bfx_i$ and $\nu_i$ is simply $\min_{j, j' \in [m]} \frac{d_{ij}}{d_{ij'}} > 0$. 
For CES disutility function with $\rho \in (1, \infty)$, i.e., $d_i(\bfx_i) = (\sum_{j=1}^m d_{ij} x_{ij}^\rho)^{1/\rho}$, this modulus is 
\begin{equation}
    \nu_i = \min_{j, j' \in [m]} \frac{d_{ij}}{d_{ij'}} \left( \frac{\norm*{\bfB}_1}{4 n m^2 B_i} \right)^{\rho - 1} > 0. 
    \label{eq:nui defn CES}
\end{equation}
% Note that, for the maximum-burden utility (i.e., CES with $\rho = \infty$) in the form of $\max_{j \in [m]}\{ d_{ij} x_{ij} \}$, $\nu_i = 0$ since some coordinates of a subgradient can be zero.
% We discuss this case separately.
% In this case, we gives an alternative regularization condition: The disutility satisfies that: 
% \begin{equation}
%     \label{an alternative condition for tentative presentation}
%     \textnormal{there exist $\delta > 0$ s.t. for any $\bfx_i$ there is a $C > 0$ such that $d'_{ij} \geq \delta$ for any $\bfd'_i \in \partial d_i(\bfx_i)$ if $x_{ij} \geq C$.}
% \end{equation}
% \begin{proof}
%     The gradient of the disutility functions is 
%     \begin{equation*}
%         \left( \nabla d_i(\bfx_i) \right)_j = \big(\sum\nolimits_{j=1}^m d_{ij} x_{ij}^\rho\big)^{\frac{1}{\rho} - 1} d_{ij} x_{ij}^{\rho - 1}, 
%     \end{equation*}
%     then we have $\left( \nabla d_i(\bfx_i) \right)_j / \left( \nabla d_i(\bfx_i) \right)_{j'} = (d_{ij} / d_{ij'}) (x_{ij} / x_{ij'})^{\rho - 1}$. 
%     The above $\nu_i$ then serves as a lower bound of this quantity.
% \end{proof}

Let 
\begin{equation}
    \ell_0 := \frac{\norm*{\mathbf{B}}_1}{3m} \min_{i \in [n]} \nu_i.
    \label{eq:def-ell-0}
\end{equation}
In a chores Fisher market with CCH disutilities satisfying~\cref{assump:strictly-increasing}, $\ell_0 > 0$.

\begin{lemma}
    In a chores Fisher market with CCH disutilities satisfying~\cref{assump:strictly-increasing}, given any $\bfp \in \mathbb{R}^m_+$ such that $\max_{j \in [m]} p_{j} \geq \frac{\norm*{\mathbf{B}}_1}{2m}$, for any chore $j$, if $p_j \leq \ell_0$ then 
    $z_j < \frac{1}{2m} - 1 < 0$ for all $\bfz \in Z(\bfp)$.
    Additionally assuming $\max_{j \in [m]} p_{j} \leq \tfrac{3}{2} \norm{\bfB}_1$, 
    if $p_j \leq \ell_0$ then $\tilde{z}_j < -\frac{1}{6m} < 0$ for all $\tilde{z} \in \tilde{Z}(\bfp)$.
    \label{lem:property-no-demand-for-zero-payment}
    \label{lem:excess-demand-barrier}
\end{lemma}
\begin{proof}
    Given the conditions in the first part of the statement, 
    suppose that there is an excess demand vector $\bfz \in Z(\bfp)$ such that 
    $z_j \geq \frac{1}{2m} - 1$, then there has to be an agent $i$ who demands at least $\frac{1}{2 n m}$ of chore $j$, i.e., there exists $i \in [n]$ and 
    \begin{equation*}
        \bfx'_i \in \textnormal{argmin}_{\bfy_i \in \mathbb{R}^m_+} \left\{ d_i(\bfy_i) \left| \, \inp*{\bfp}{\bfy_i} \geq B_i
        \right. \right\}
    \end{equation*}
    such that $x'_{ij} \geq \frac{1}{2 n m}$. 
    Let $\bar{j}$ be the chore with the maximum price $p_{\bar{j}} > \frac{\norm*{\mathbf{B}}_1}{2m}$. 
    Observe that $x'_{i \bar{j}} \leq \frac{B_i}{p_{\bar{j}}} < \frac{2m B_i}{\norm*{\bfB}_1}$. 
    Then, we have for any $\bfd'_i \in \partial d_i(\bfx'_i)$, $d'_{ij} > 0$ and $d'_{i\bar{j}} < \infty$ since $\min_{i \in [n]} \nu_i > 0$, and 
    \begin{equation*}
        \frac{p_j}{d'_{ij}} \leq \frac{\ell_0}{d'_{ij}} = \frac{\frac{\norm*{\mathbf{B}}_1}{3m} \min_{i \in [n]} \nu_i}{d'_{i\bar{j}} \cdot ({d'_{ij}}/{d'_{i\bar{j}}})} \overset{\textnormal{(a)}}{<} \frac{\norm{\bfB}_1}{2m d'_{i\bar{j}}} \leq \frac{p_{\bar{j}}}{d'_{i\bar{j}}}, 
    \end{equation*}
    where (a) is by the definition of $\nu_i$ in \cref{eq:nu defn}.
    This contradicts the optimality of agent $i$. 

    For the second part, if additionally we have $\max_{j \in [m]} p_{j} \leq \tfrac{3}{2} \norm{\bfB}_1$, 
    then for each $\bfz \in Z(\bfp(t))$, 
    letting $b_{ij}$ denotes the money agent $i$ earns from chore $j$ corresponding to $\bfz$, 
    the mean of the excess demand is lower bounded by 
    \begin{equation}
        \frac{1}{m} \sum_{j=1}^m z_{j} = 
        \frac{1}{m} \sum_{j = 1}^m \left( \sum_{i=1}^n \frac{b_{ij}}{p_j} - 1 \right) \geq \frac{1}{m} \sum_{j = 1}^m \left( \frac{2}{3\norm{\bfB}_1} \sum_{i=1}^n b_{ij} - 1 \right) = \frac{2\sum_{i=1}^n\sum_{j = 1}^m b_{ij}}{3m\norm{\bfB}_1}  - 1 = \frac{2}{3m} - 1.  
        \label{eq:excess-demand-mean-lower-bound}
    \end{equation}
    This implies that for any $\tilde{\bfz}\in \tilde{Z}(\bfp(t))$, $\tilde{z}_j = z_j - \frac{1}{m}\sum_{j'=1}^m z_{j'} < (\frac{1}{2m} - 1) - (\frac{2}{3m} - 1) = - \frac{1}{6m} < 0$ for all $j \in [m]$, where $\bfz$ is an excess demand vector supporting $\tilde{\bfz}$.
\end{proof}

\begin{lemma}
    In a chores Fisher market with CCH disutilities satisfying~\cref{assump:strictly-increasing}, 
    for any $\bfp \in \mathbb{R}^m_+$ such that $\max_{j \in [m]} p_{j} \geq \frac{\norm*{\mathbf{B}}_1}{2m}$, we have $\norm{\tilde{\bfz}} < \infty$ for all $\tilde{\bfz} \in \tilde{Z}(\bfp)$.
    \label{lem:finite-rate}
\end{lemma}
\begin{proof}
    Since $\ell_0 > 0$, 
    % which means $\nu_i > 0\; \forall\, i \in [n]$ 
    by~\cref{lem:excess-demand-barrier},
    if $p_j \leq \ell_0$ then 
    $-1 \leq z_j < \frac{1}{2m} - 1$ thus $| z_j | < \infty$ for all $\bfz \in Z(\bfp)$. 
    On the other hand, if $p_j \geq \ell_0$, then $z_j \leq \frac{\norm{\bfB}_1}{p_j} \leq \frac{\norm{\bfB}_1}{\ell_0} < \infty$.
    It then follows that $\norm{\tilde{\bfz}} < \infty$ for all $\tilde{\bfz} \in \tilde{Z}(\bfp)$. 
    \qedhere
    % In the $\infty$-CES case, 
    % if $p_j > 0$ for all $j \in [m]$ then $\bfz$ corresponds to a multiple of a fixed bundle. 
    % By the agent optimality, we have $\norm{\tilde{\bfz}} < \infty$.
    % If $p_j = 0$ for some chore $j$, any agent cannot demand more chores than the maximum-disutility bundle, which is finite. 
    % Therefore, we also have $\norm{\tilde{\bfz}} < \infty$ in this case.
\end{proof}

% \cref{lem:excess-demand-barrier} implies a barrier along the boundary of nonnegative orthant, which helps us prove that any trajectory starting from $\Delta_B$ is.
\begin{lemma}
    In a chores Fisher market with CCH disutilities satisfying~\cref{assump:strictly-increasing}, 
    \emph{any} trajectory $\bfp(t)$ of the dynamical system \textnormal{(\ref{chores-tatonnement})} starting from any initial point $\bfp(0) \in \Delta_B$ satisfies that 
    $\bfp(t) \in \Delta_B$ for all $t \geq 0$. 
    % \begin{itemize}
    %     \item $p(t) \in \Delta_B$ for $t \geq 0$; 
    %     \item $p_j(t) \geq \min\{ \ell_0, \min_j p_j^0 \}$ where $\ell_0 = \frac{\norm{B}_1}{2m} ( \min\nolimits_{i, j, j'} \frac{d_{ij}}{d_{ij'}} ) $ for $t \geq 0$; 
    %     \item $\norm{\dot{p}(t)} \leq \infty$ for $t \geq 0$. 
    % \end{itemize}
    \label{lem:well-defined-equilibrium}
\end{lemma} 
\begin{proof}
    % [Proof of Lemma 1 for general convex disutilities]
    % Now, we consider a class of disutility functions that satisfy 
    % \textbf{continuous, convex, $1$-homogeneous} and~\cref{lem:property-no-demand-for-zero-payment}. 
    Consider a trajectory $\bfp(t)$ of \textnormal{(\ref{chores-tatonnement})} starting from any initial point $\bfp(0) = \bfp^0 \in \Delta_B$. 
    We first show that the trajectory will stay within the affine hull $\Hb$ by contradiction.
    Suppose that there is a $t_0$ such that the price trajectory leaves $H_B$ for the first time, i.e., 
    \begin{equation*}
        t_0 := \inf\{ t > 0 \mid \bfp(t) \notin H_B \}. 
    \end{equation*}
    
    Then, by~\cref{lem:finite-rate}, $\norm{\dot{\bfp}(t)} < \infty$ for $t \in [0, t_0]$. 
    % because the excess demand is finite when the price is close to $0$\footnote{Note that this statement holds under the condition that there is at least a chore with sufficiently large price, which is usually ensured by $\norm{p}_1 = \norm{B}_1$.}. 
    As $\bfp(t)$ is an absolutely continuous curve, we have 
    \begin{equation} 
        \sum_{j \in [m]} p_j(t_0) = \sum_{j \in [m]} p_j(0) + \sum_{j \in [m]} \int_0^{t_0} \dot{p}_j(s) d s = \sum_{j \in [m]} p_j^0 + \int_0^{t_0} \sum_{j \in [m]} \dot{p}_j(s) d s = \norm{\bfB}_1, \; \forall\, t \in [0, t_0] 
        \label{eq:trajectory-in-HB-general-convex}
    \end{equation} 
    where the last equality follows from $\sum_{j=1}^m \tilde{z}_j = 0 \;\forall\, \tilde{z} \in \tilde{Z}(\bfp(t))$, and $\bfp^0 \in H_B$. 
    This means $\bfp(t_0) \in H_B$, which contradicts the definition of $t_0$.

    Next we show that $\bfp(t) \geq 0$ for all $t \geq 0$. 
    This can be seen from the following argument. 
    Consider any $t^0_j$ such that $p_j(t^0_j) = 0$.
    By~\cref{lem:excess-demand-barrier,lem:finite-rate}, 
    $p_j(t)$ is differentiable and $\dot{p}_j(t^0_j) > 0$, which implies that 
    $p_j(t^0_j + h) > p_j(t^0_j) = 0$ for any sufficiently small $h > 0$.
    Hence, the trajectory cannot cross from nonnegative to negative values.
    Since $\bfp(0) \geq 0$, we can conclude that $\bfp(t) \geq 0$ for all $t \geq 0$. 
    \qedhere

\end{proof}

Next, we study the convergence of the dynamical system \textnormal{(\ref{chores-tatonnement})} by using a Lyapunov potential function. 
% In general, we can analyze the t\^atonnement process for the general continuous, convex and 1-homogeneous (CCH) disutility functions. 
% Let $d_i: \mathbb{R}^m_+ \rightarrow \mathbb{R}_+$ be such a CCH function for each agent $i$. 
Consider the following nonsmooth nonconvex function: 
\begin{equation}\label{eq:potential-CCH}
\begin{array}{c}
    f(\bfp) = - \displaystyle\sum_{j = 1}^m p_j + \sum_{i = 1}^n B_i \log\left( \max_{\bfx_i \geq 0: d_i(\bfx_i) \leq 1} \inp{\bfp}{\bfx_i} \right) = - \sum_{j = 1}^m p_j + \sum_{i = 1}^n B_i \max_{\bfx_i \geq 0: d_i(\bfx_i) \leq 1} \log\left( \inp{\bfp}{\bfx_i} \right). 
    % \\
    % \mbox{{Chores potential function} $f\vert_{H_B}: \Delta_B \rightarrow \mathbb{R}$, a restriction of $f$ to the affine hull $H_B$.}
    \end{array} 
\end{equation}
We denote the restriction of function $f$ to the affine hull $H_B$ by $f\vert_{H_B}$. 
% Again, this function is nonsmooth nonconvex in general, thus we need Clarke subdifferential when we study their ``gradients''. 
%and the restriction of $f$ to the affine hull $H_B$. We refer to the restricted function $f\vert_{H_B}: \Delta_B \rightarrow \mathbb{R}$ as the \emph{chores potential function}. 
We call $f\vert_{H_B}: \Delta_B \rightarrow \mathbb{R}$ \emph{Chores potential function}.
The idea for this potential function comes from the ``redundant dual'' of the Eisenberg-Gale nonconvex program for chores~\cite{chaudhury2024competitive}. 
% Our potential function corresponds to reformulating away the $\beta$ variables in this ``dual'' program. 
Our potential function significantly generalizes the range of disutility functions. 
For convex CES functions, \cref{eq:potential-CCH} reduces to the following closed-form expressions.

For the linear disutility case, 
\begin{equation}
    f(\bfp) := - \sum_{j = 1}^m p_j + \sum_{i = 1}^n B_i \log\left( \max_{j \in [m]} \frac{p_j}{d_{ij}} \right) 
    = - \sum_{j = 1}^m p_j + \sum_{i = 1}^n B_i \max_{j \in [m]} \left\{ \log\left( \frac{p_j}{d_{ij}} \right) \right\}. 
    \label{eq:potential}
\end{equation}

For the CES disutility function with $\rho \in (1, \infty)$ 
\begin{equation*}
    f(\bfp) := - \sum_{j = 1}^m p_j + \sum_{i = 1}^n \frac{B_i}{\sigma} \log\left( \sum_{j=1}^m d_{ij}^{1-\sigma} p_j^\sigma \right), \; \text{ where } \sigma = \frac{\rho}{\rho - 1}. 
\end{equation*}

We will need the following classical results. 
\begin{proposition}[{\cite[Example 7.28, Exercise 8.31, Theorem 10.31]{rockafellar2009variational}}]
    % The pointwise maximum of a (possibly infinite) collection of smooth functions
    Let 
    % $f^I(\bfx) = \max_{i \in I} f_i(\bfx)$ where $I$ is a finite index set and $f_i$ is smooth for all $i \in I$, 
    % and 
    $f^T(\bfx) = \max_{t \in T} f_t(\bfx)$ where $T$ is a compact (i.e., closed and bounded) index set and $f_t$ is smooth for all $t \in T$). 
    Then, we have 
    \begin{itemize}
        \item 
        % $f^I(\bfx)$ and 
        $f^T(\bfx)$ is subdifferentially regular; 
        \item 
        % $\partial f^I(\bfx) = \textnormal{Conv}\left(\{ \nabla f_i(x) \vert f_i(x) = \max_{i' \in I} f_{i'}(x) \}\right) 
        % $
        % and  
        $\partial f^T(\bfx) = \textnormal{Conv}\left(\{ \nabla f_t(\bfx) \mid f_t(\bfx) = \max_{t' \in T} f_{t'}(\bfx) \}\right) 
        $. 
    \end{itemize}
    \label{prop:clarke-subdifferential-max}
\end{proposition}

Another tool we will use is related to consumer theory, which has been well established for goods. 
% \tianlong{[CITE]} 
\cite{goktas2021consumer} exploits two closely related problems: the utility
maximization problem and the expenditure minimization problem (EMP) in consumer theory
to study Fisher markets with goods.
% continuous, concave, homogeneous (CCH) utility functions representing \emph{locally non-satiated}\footnote{If a buyer’s utility function represents locally non-satiated preferences, there always exists a better bundle for that buyer if their budget increases.} preferences. 
Here, 
we build up such relations for Fisher markets with chores, with self-contained proofs. 

Given a price vector $\bfp \in \mathbb{R}^m_+$, for the $i$-th agent, we consider the disutility minimization problem (DMP)
\begin{equation}
    \begin{aligned}
        \min_{\bfx_i \geq 0: \inp{\bfp}{\bfx_i} \geq B_i} d_i(\bfx_i)
    \end{aligned}
    \tag{DMP}
    \label{eq:DMP}
\end{equation}
and the earning maximization problem (EMP)
\begin{equation}
    \begin{aligned}
        \max_{\bfx_i \geq 0: d_i(\bfx_i) \leq D_i} \inp{\bfp}{\bfx_i} 
    \end{aligned}
    \tag{EMP}
    \label{eq:EMP}
\end{equation} 
where $B_i$ is the known earning requirement and $D_i$ is 
some disutility upper bound (unknown in the Fisher market setting).
% defined by $D_i = \min_{\bfx_i \geq 0: \inp{p}{x_i} \geq B_i} d_i(x_i)$. 
Regarding $B_i$ and $D_i$ as input parameters for the two problems, we define the following terms: 
\begin{equation*}
    g_i(\bfp, B_i) = \min_{\bfx_i \geq 0: \inp{\bfp}{\bfx_i} \geq B_i} d_i(\bfx_i) \hspace{30pt} \text{and} \hspace{30pt} x_i^*(\bfp, B_i) = \textnormal{argmin}_{\bfx_i \geq 0: \inp{\bfp}{\bfx_i} \geq B_i} d_i(\bfx_i); 
\end{equation*}
\begin{equation*}
    h_i(\bfp, D_i) = \max_{\bfx_i \geq 0: d_i(\bfx_i) \leq D_i} \inp{\bfp}{\bfx_i} \hspace{30pt} \text{and} \hspace{30pt} y_i^*(\bfp, D_i) = \textnormal{argmax}_{\bfx_i \geq 0: d_i(\bfx_i) \leq D_i} \inp{\bfp}{\bfx_i}. 
\end{equation*}

Next, we present the following intuitive but useful identities. 
\begin{proposition}
    If the disutilities are continuous and strictly increasing in some coordinate, then we have 
    \begin{enumerate}
        \item 
        If $D_i = \min_{\bfx_i \geq 0: \inp{\bfp}{\bfx_i} \geq B_i} d_i(\bfx_i)$, then 
        $B_i = \max_{\bfx_i \geq 0: d_i(\bfx_i) \leq D_i} \inp{\bfp}{\bfx_i}$; 
        \item $y_i^*(\bfp, D_i) = x_i^*(\bfp, B_i)$; 
        \item If the disutilities are $1$-homogeneous, $h_i(\bfp, D_i)$ and $y_i^*(\bfp, D_i)$\footnote{A set-valued function $F(\bfx)$ is said to be $1$-homogeneous if $F(a \bfx) = \{ a \mathbf{f} | \mathbf{f} = F(\bfx) \}$ for any $a > 0$.} are $1$-homogeneous in $D_i$. 
        % \RM{$y^*$ is a set. So 1-homogeneous for $y^*$ is undefined.}
    \end{enumerate}
    \label{prop:customer-theory}
\end{proposition}
\begin{proof}
    For the first point, 
    $\max_{\bfx_i \geq 0: d_i(\bfx_i) \leq D_i} \inp{\bfp}{\bfx_i} \geq B_i$ trivially holds because any $\bfx'_i \in x_i^*(\bfp, B_i)$ by definition satisfies $d_i(\bfx_i') \leq D_i$ and $\inp{\bfp}{\bfx_i'} \geq B_i$. 
    Suppose that there is a $\bfx_i' \in x_i^*(\bfp, B_i)$ such that $\inp{\bfp}{\bfx_i'} > B_i$, then there exists a $\bar{\bfx}_i = a \bfx_i'$ for some $a \in (0, 1)$ such that $\inp{\bfp}{\bar{\bfx}_i} = B_i$. 
    Because the disutility function is strictly increasing in some $x_{ij}$, 
    it holds that $d_i(\bar{\bfx}_i) < d_i(\bfx_i') \leq D_i$. 
    This contradicts the definition of $D_i$. Therefore, $B_i = \max_{\bfx_i \geq 0: d_i(\bfx_i) \leq D_i} \inp{\bfp}{\bfx_i}$. 

    For the second point, we first show that $y_i^*(\bfp, D_i) \subseteq x_i^*(\bfp, B_i)$. 
    For any $\bfy_i' \in y_i^*(\bfp, D_i)$, we have $d_i(\bfy_i') \leq D_i$ because of the feasibility of~\eqref{eq:EMP} and $\inp{\bfp}{\bfy_i'} = B_i$ by the first point. 
    Combining this with $D_i = \min_{\bfx_i \geq 0: \inp{\bfp}{\bfx_i} \geq B_i} d_i(\bfx_i)$, 
    we have $\bfy_i'$ has to be a minimizer of~\eqref{eq:DMP} hence $y_i^*(p, D_i) \subseteq x_i^*(p, B_i)$. 
    Similarly, any $\bfx_i' \in x_i^*(\bfp, B_i)$ satisfies $\inp{\bfp}{\bfx_i'} \geq B_i$ by the feasibility of~\eqref{eq:DMP} and $d_i(\bfx_i') = D_i$ by definition. Combining this with the first point, we have $\bfx_i'$ has to be a maximizer of~\eqref{eq:EMP}. 
    Therefore, $y_i^*(\bfp, D_i) \supseteq x_i^*(\bfp, B_i)$.

    For the third point, 
    first we have for any $a > 0$ 
    \begin{equation*}
        h_i(\bfp, a) = 
        \max_{\bfx_i \geq 0: d_i(\bfx_i) \leq a} \inp{\bfp}{\bfx_i} = a \max_{\bfx_i \geq 0: d_i\left(\frac{\bfx_i}{a}\right) \leq 1} \inp*{p}{\frac{\bfx_i}{a}} = a \max_{\bfy_i \geq 0: d_i(\bfy_i) \leq 1} \inp{\bfp}{\bfy_i} = a h_i(\bfp, 1). 
    \end{equation*}
    Then, we show that $y^*(\bfp, a) = a \odot y^*(\bfp, 1)$ for all $a > 0$, where $\odot$ denotes the element-wise scalar multiplication for a set. 
    This can be shown as follows. Given $a > 0$ and any $\bfy'(1) \in y^*(\bfp, 1)$, $a \bfy'(1)$ is a feasible solution to $\max_{\bfx_i \geq 0: d_i(\bfx_i) \leq a} \inp{\bfp}{\bfx_i}$ that gives the maximum value $\inp{\bfp}{a \bfy'(1)} = a h(\bfp, 1) = h(\bfp, a)$, 
    which means $a \bfy'(1) \in y^*(\bfp, a)$. Conversely, for any $\bfy'(a) \in y^*(\bfp, a)$, $\frac{1}{a}\bfy'(a)$ is a maximizer of $\max_{\bfx_i \geq 0: d_i(\bfx_i) \leq 1} \inp{\bfp}{\bfx_i}$. 
    This leads to $\frac{1}{a} \bfy'(a) \in y^*(\bfp, 1)$. 
\end{proof}

\begin{lemma}
    If $d_i$ is a CCH function satisfying \cref{assump:strictly-increasing} for all $i \in [n]$, 
    and 
    $f: \mathbb{R}^m_+ \rightarrow \mathbb{R}$ and $f\vert_{H_B}: \Delta_B \rightarrow \mathbb{R}$
    are defined as in~\cref{eq:potential-CCH}. 
    Then, $f$ and
    $f\vert_{H_B}$ 
    are locally Lipschitz functions that are subdifferentially regular. 
    Moreover, 
    % $\partial f(\bfp) = Z(\bfp)$ and  
    $\partial f\vert_{H_B}(\bfp) = \tilde{Z}(\bfp)$ for any $\bfp \in \Delta_B$.
    % in their domains. 
    \label{lem:subdifferentially-regular-and-generalized-subdifferential}
\end{lemma} 
\begin{proof}
    % Locally Lipschitz ... 
    Let $\g_i(\bfp) = B_i \max_{\bfx_i \geq 0: d_i(\bfx_i) \leq 1} \log\big( \inp{\bfp}{\bfx_i} \big)$. 
    By~\cref{prop:clarke-subdifferential-max}, we have $\partial \g_i(\bfp)$ is subdifferentially regular (hence it admits a chain rule) and 
    \begin{equation*}
        \partial \g_i(\bfp) = \left\{ \frac{B_i}{\inp{\bfp}{\bfy_i}} \bfy_i \Big| \bfy_i \in y_i^*(\bfp, 1) \right\} = \left\{ \frac{B_i}{\inp{p}{D_i \bfy_i}} D_i \bfy_i \Big| D_i \bfy_i \in y_i^*(\bfp, D_i) \right\}, 
    \end{equation*} 
    where the second equality follows by the third point in~\cref{prop:customer-theory}. 
    Replacing $D_i \bfy_i$ with $\bfx_i$, we have 
    \begin{equation*}
        \partial \g_i(\bfx) = \left\{ \frac{B_i}{\inp{\bfp}{\bfx_i}} \bfx_i \,\Big|\, \bfx_i \in y_i^*(\bfp, D_i) \right\} = x_i^*(\bfp, B_i), 
    \end{equation*}
    where the last equality follows by the first two points in~\cref{prop:customer-theory}. 
    % Denote $e^*(1) = h(1)$ and $d^*(B_i) = \min_{x_i \geq 0: \inp{p}{x_i} \geq B_i} d_i(x_i)$. 
    % \begin{equation*}
    %     \textnormal{argmax}_{x_i \geq 0: d_i(x_i) \leq d^*(B_i)} \inp*{p}{x_i} = \textnormal{argmin}_{x_i \geq 0: \inp{p}{x_i} \geq B_i} d_i(x_i). 
    % \end{equation*}
    We next use the following facts about the subdifferentially regular functions~\cite{li2020understanding}: 
    \begin{itemize}
        \item A convex function is subdifferentially regular; 
        \item The subdifferentially regular functions satisfy the \emph{sum rule}: %\footnote{The sum rule does not always hold in general.}: 
        Suppose that $\g = \sum_{i \in [n]} \g_i$, where $\g_i, \, \forall\, i \in [n]$ are locally Lipschitz functions. 
        Given a point $\bfp$, if $\g_1, \ldots, \g_m$ are subdifferentially regular at $\bfp$, then $\g$ is subdifferentially regular as well and 
        $\partial \g(\bfp) = \partial \g_1(\bfp) + \cdots + \partial \g_m(\bfp)$, 
        where the $+$ denotes \emph{Minkowski addition}. 
    \end{itemize} 
    Then, it follows that 1) $f$ is subdifferentially regular and 2) the Clarke subdifferential of $f$ at $\bfp$ is 
    \begin{equation*}
        \partial f(\bfp) = 
        -\mathbf{1}_m + \sum_{i = 1}^n x_i^*(\bfp, B_i) = Z(\bfp).
    \end{equation*}

    Following the derivations in~\cite[$\S 10.3$]{mikhalevich2024methods}, we have that the Clarke subdifferential of $f\vert_{H_B}$ at $\bfp$ is\footnote{\cite{mikhalevich2024methods} show this for a more general family of pseudogradient mappings. As argued in \cite{mikhalevich2024methods}, the Clarke subdifferential is the \emph{minimal inclusion} mapping with respect to this family of pseudogradient mappings, and it follows from there that the formula must apply to the Clarke subdifferential. See~\cref{app:sec:generalized-subdifferential}.}
    \begin{equation*}
        \partial f\vert_{H_B}(\bfp) = \left\{ \bfg_0 \in H_0 \, \left\vert \, \bfg_0 = \Pi_{H_0} (\bfg), \bfg \in \partial f(\bfp) \right. \right\}, 
    \end{equation*}
    where $H_0$ is the linear subspace associated with $H_B$.
    By first-order optimality conditions, one can check that $\Pi_{H_0}(\bfx) =  \textnormal{argmin}_{\bfx'} \big\{ \norm*{\bfx' - \bfx} \, | \, A \bfx = 0 \big\} = (I - A^\top (A A^\top)^{-1} A) \bfx$ where $A = \mathbf{1}_m^\top \in \mathbb{R}^{1 \times m}$ and $I$ is an $m$-dimensional identity matrix. 
    By calculation, $\Pi_{H_0}(\bfx) = \bfx - \left( \frac{1}{m} \sum_{j=1}^m x_j \right) \cdot \mathbf{1}_m$. 
    Therefore, we have that $\partial f(\bfp) = Z(\bfp)$ and hence $\partial f\vert_{H_B}(\bfp) = \tilde{Z}(\bfp)$ for any $\bfp \in \mathbb{R}^m_{+}$. 
    
    The subdifferential regularity of $f\vert_{H_B}$ can be proven by using similar arguments as in~\cite[$\S 10.3$]{mikhalevich2024methods}. 
    Any Clarke generalized gradient $\bfg \in \partial f(\bfp)$ can be decomposed into $\Pi_{H_0} \bfg$ and the orthogonal complement to $\Pi_{H_0}(\bfg)$, that is, $\bfg = \Pi_{H_0}(\bfg) + \Pi_{H_0^\perp}(\bfg)$. 
    By the subdifferential regularity of $f$ we have for any $\bfx, \bfy \in H_B$ 
    \begin{align*}
        f\vert_{H_B}(\bfy) &\geq f\vert_{H_B}(\bfx) + \inp*{\bfg}{\bfy - \bfx} + o\left( \norm{\bfy - \bfx} \right) \\ 
        &= f\vert_{H_B}(\bfx) + \inp*{\Pi_{H_0}(\bfg) + \Pi_{H_0^\perp}(\bfg)}{\bfy - \bfx} + o\left( \norm{\bfy - \bfx} \right) \\ 
        &= f\vert_{H_B}(\bfx) + \inp*{\Pi_{H_0}(\bfg)}{\bfy - \bfx} + o\left( \norm{\bfy - \bfx} \right) \quad \text{as } \bfy \rightarrow \bfx, 
    \end{align*}
    where the last equality follows because $\Pi_{H_0}(\bfg)$ is orthogonal to vectors in $H_B$. 
    It is trivial to show $f\vert_{H_B}$ is locally Lipschitz continuous. 
    Therefore, the subdifferential regularity of $f\vert_{H_B}$ is proven. 
\end{proof}

% We note that the function $f$ can be seen as the potential function of naive chores t\^atonnement. 

% Although our dynamical system is very close to the general setup in~\cite{davis2020stochastic}, our corresponding optimization problem is constrained to $\mathbb{R}^m_+$, compared to their unconstrained minimization problem. 
\begin{remark}
    \cref{lem:well-defined-equilibrium} shows that any trajectory of \textnormal{(\ref{chores-tatonnement})} will not escape the domain of $f\vert_{H_B}$. Therefore, 
    % it is safe to treat~
    \textnormal{(\ref{chores-tatonnement})} can be seen as a subgradient dynamical system associated with the unconstrained optimization problem $\min_{p \in \mathbb{R}^m} f\vert_{H_B}$. 
\end{remark}

To show the stability of the dynamical system~\eqref{chores-tatonnement}, we present the following descent lemma. 
\cite{davis2020stochastic} showed that such a descent lemma holds under some sufficient conditions, e.g., subdifferential regularity.
For completeness,  
we provide a simple proof for a more general setting in~\cref{app:sec:Omitted proofs}, which is from~\cite{ding2024stochastic}. 
\begin{lemma}[Descent]
    In a chores Fisher market with CCH disutilities satisfying~\cref{assump:strictly-increasing}, 
    whenever $\bfp: \mathbb{R}_+ \rightarrow \mathbb{R}^m_+$ is a trajectory of the dynamical system \textnormal{(\ref{chores-tatonnement})}, and $\bfp(0) \in H_B$ is not a stationary point of \textnormal{(\ref{chores-tatonnement})}, then for all $t \geq 0$, $\bfp(t) \in \Delta_B$, and there exists a real number $T > 0$ satisfying $f\vert_{H_B}(\bfp(T)) < \sup_{t \in [0, T]} f\vert_{H_B}(\bfp(t)) \leq f\vert_{H_B}(\bfp(0))$.
    % \begin{equation}
    %     f\vert_{H_B}(\bfp(T)) < \sup_{t \in [0, T]} f\vert_{H_B}(\bfp(t)) \leq f\vert_{H_B}(\bfp(0)). 
    %     \label{eq:descent}
    % \end{equation}
    \label{lem:descent}
\end{lemma} 
% \ck{Put proof in appendix and simplify via \cite{ding2024stochastic}}
% \begin{proof}
%     As shown in~\cref{lem:well-defined-equilibrium}, the dynamical system is well defined and for any $t \geq 0$, $p_j(t) \geq \min(p^0_j,\ell_0)$ for each $j \in [m]$, where $\ell$ is defined in~\cref{eq:price lb}. 
%     % Then, similar to the proof of~\cref{lem:well-defined-equilibrium} we can show 
%     % \begin{equation*}
%     %     \sum_{j \in [m]} p_j(t) = \sum_{j \in [m]} p_j(0) - \sum_{j \in [m]} \int_0^{t} \tilde{z}(s) d s = \sum_{j \in [m]} p_j(0) - \int_0^{t} \sum_{j \in [m]} \tilde{z}(s) d s = \norm{B}_1 \quad \forall\; t \geq 0. 
%     % \end{equation*} 
%     % Hence, $p(t) \in \textnormal{dom}\; f\vert_{H_B}$ for all $t \geq 0$. 
%     The descent property can be verified by Lemma 5.4 and 5.2 in~\cite{davis2020stochastic}. 
%     In short, since we have shown $f\vert_{H_B}$ is a locally Lipschiz function that is subdifferentially regular, by~\citet[Lemma 5.4]{davis2020stochastic} $f\vert_{H_B}$ admits a chain rule. One can prove~\cref{eq:descent} by using~\citet[Lemma 5.2]{davis2020stochastic}. See \tianlong{Appendix []} for more details on this proof. (e.g., the definition of the chain rule...). 
% \end{proof}

% \footnote{One can also think of optimizing over an extended function $\tilde{f}: H_B \rightarrow \mathbb{R}$ where $\tilde{f}\vert_{H_B} = f\vert_{H_B}$ whenever $p \in H_B$ and some other differentiable (e.g., quadratic function) such that locally Lipschitz continuity and subdifferential regularity still hold.}. 

From~\cref{lem:descent} and the fact that the stationary points of \textnormal{(\ref{chores-tatonnement})} correspond to equilibrium prices and thus are disconnected, we know any trajectory of \textnormal{(\ref{chores-tatonnement})} converges to some stationary point. We include a proof for a more general theorem implying this result in~\cref{app:sec:Omitted proofs}. 
Since by~\cref{lem:subdifferentially-regular-and-generalized-subdifferential} any stationary point corresponds to a CE, we can conclude convergence of~\eqref{chores-tatonnement}. 

\begin{theorem}[Continuous-time convergence]
    In a chores Fisher market with CCH disutilities satisfying~\cref{assump:strictly-increasing}, 
    any trajectory of the dynamical system \eqref{chores-tatonnement} converges to a CE. 
    \label{thm:relative-tatonnement-stable}
\end{theorem}

\subsection{A EG-type Program for General CCH Disutilities}

Our potential functions from the previous section suggest a new program to compute chores CE with general CCH disutilities. We formally write the new program as follows, which we call the general EG dual program. 
\begin{equation}
    \begin{aligned}
        \min_{p \geq 0} \quad & - \sum_{j = 1}^m p_j + \sum_{i = 1}^n B_i \log\left( \max_{x_i \geq 0: d_i(x_i) \leq 1} \inp{p}{x_i} \right) \\ 
        \textnormal{s.t.} \quad & \sum_{j=1}^m p_j = \sum_{i = 1}^n B_i. 
    \end{aligned}
    \tag{General EG Dual}
    \label{pgm:general-eg-dual}
\end{equation}

We formally show that the set of equilibrium prices in the chores Fisher market corresponds to the set of solutions to~\eqref{pgm:general-eg-dual} that are KKT points.\footnote{We say that a primal point to a program is a KKT point if there is a tuple of dual variables such that the primal point coupled with the dual variables satisfies the KKT conditions.}
\begin{theorem}
    If $d_i$ is a CCH function satisfying \cref{assump:strictly-increasing} for all $i \in [n]$, 
    then there is a one-to-one mapping between the set of equilibrium prices in the chores Fisher market and the set of KKT points of~\eqref{pgm:general-eg-dual}.
\end{theorem}
\begin{proof}
    Let $\mu$ be the dual variable corresponding to the constraint. 
    The KKT conditions are 
    \begin{align*}
        & \bfalpha \in \partial f(\bfp) + \mu \cdot \mathbf{1}_m \\ 
        & \alpha_j \geq 0 \hspace{100pt} \forall j \in [m] \\ 
        & p_j \cdot \alpha_j = 0 \hspace{100pt} \forall j \in [m] \\
        & \sum_{j=1}^m p_j = \sum_{i = 1}^n B_i \\ 
        & p_j \geq 0 \hspace{100pt} \forall j \in [m],  
    \end{align*}
    where $\bfalpha$ denotes the $m$-dimensional dual vector to $p\geq 0$. 
    Let $\tilde{\bfp}$, $\tilde{\mu}$, $\tilde{\bfalpha}$ be a set of variables satisfying the KKT conditions. 
    By~\cref{lem:subdifferentially-regular-and-generalized-subdifferential}, we have that 
    \begin{equation*}
        \tilde{\bfalpha} - \tilde{\mu} \cdot \mathbf{1}_m \in \partial f(\tilde{\bfp}) = Z(\tilde{\bfp}). 
    \end{equation*}
    Thus, there exists a $\tilde{\bfz} \in Z(\tilde{\bfp})$ such that $\tilde{\bfalpha} - \tilde{\mu} \cdot \mathbf{1}_m = \tilde{\bfz}$. 
    Since $\sum_{j=1}^m \tilde{p}_j = \sum_{i = 1}^n B_i$, $\tilde{p}_j \cdot \tilde{\alpha}_j = 0$ for all $j \in [m]$ and $\inp{\tilde{\bfp}}{\tilde{\bfz}} = 0$ (by the agent optimality and $\sum_{j=1}^m \tilde{p}_j = \sum_{i = 1}^n B_i$), we have 
    \begin{equation*}
        - \tilde{\mu} \sum_{i=1}^n B_i = - \tilde{\mu}\inp{\tilde{\bfp}}{\mathbf{1}_m} = \inp{\tilde{\bfp}}{\tilde{\bfalpha} - \tilde{\mu} \cdot \mathbf{1}_m} = \inp{\tilde{\bfp}}{\tilde{\bfz}} = 0, 
    \end{equation*}
    which leads to $\tilde{\mu} = 0$. 
    Hence, $\tilde{\bfz} = \tilde{\bfalpha} \geq 0$. 
    By~\cref{lem:property-no-demand-for-zero-payment} and 
    $\tilde{\bfp} \in \Delta_B$, 
    % $\max_{j \in [m]} \tilde{p}_j \geq \frac{1}{m} \sum_{j=1}^m \tilde{p}_j \geq \frac{\norm{\mathbf{B}}_1}{2m}$, 
    we attain that $\tilde{p}_j > 0$ for all $j \in [m]$ (otherwise $\tilde{z}_j < 0$ by~\cref{lem:property-no-demand-for-zero-payment}). 
    Since $\tilde{p}_j > 0$ for all $j \in [m]$, 
    we have $\tilde{\bfz} = \tilde{\bfalpha} = \mathbf{0}_m$.
    % \ck{The below does not follow. We either need to change the definition of equilibrium, or we need to restrict the class of disutilities to ones that have weakly increasing disutility in each coordinate.}
    Because $\mathbf{0}_m \in Z(\tilde{\bfp})$, $\tilde{\bfp}$ is an equilibrium price vector.
    
    % Suppose $\tilde{p}_j = 0$ for some $j \in [m]$, and let $\tilde \bfx$ be the allocation corresponding to $\tilde \bfz$.
    % Then we have $\sum_{i=1}^n \tilde x_{ij} \geq 1$.
    % By agent optimality, for every agent $i$ such that $\tilde x_{ij} > 0$ we have that their marginal disutility from chore $j$ is zero, and thus we can decrease $\tilde x_{ij}$ while preserving the agent optimality. 
    % Therefore,  we can construct an allocation such that supply exactly equals demand.

    % Combining this with $\tilde{p}_j \cdot \tilde{\alpha}_j = 0$ for all $j \in [m]$, we have $\tilde{\bfz} = \tilde{\bfalpha} = \mathbf{0}_m$. 

    The converse is immediate: for any $\bfp^*$ at a CE, taking $\bfalpha^* = \mathbf{0}_m$ and $\mu^* = 0$ satisfies all KKT conditions. 
\end{proof}

Next, we show that our new program avoids the common ``poles'' issues in computing CE with chores, where certain directions lead to an unbounded objective. 
We define 
% \begin{equation*}
%     D_i := \max\left( \max_{j \in [m]} \max_{\bfx_i \geq 0}\left\{ \frac{\partial d_i(\bfx_i)}{\partial x_{ij}} \Big|\, \bfx_i \leq 1 \right\}, \frac{1}{m} \right). 
% \end{equation*}

\begin{equation*}
    \delta_i := \sup\left\{ \delta > 0 \mid d_i\left( \delta \cdot \mathbf{1}_m \right) \leq 1  \right\}.
\end{equation*}

% One can show that $\delta_i > 0$ 

\begin{lemma}
    For any $\bfp$ that satisfies $\sum_{j=1}^m p_j = \sum_{i = 1}^n B_i$, the objective of~\eqref{pgm:general-eg-dual} is lower bounded. In particular, 
    \begin{equation}
        - \sum_{j = 1}^m p_j + \sum_{i = 1}^n B_i \log\left( \max_{x_i \geq 0: d_i(x_i) \leq 1} \inp{p}{x_i} \right) \geq \norm{\bfB}_1 \log\left( \frac{\min_{i \in [n]} \delta_i \norm{\bfB}_1}{e} \right) > -\infty.
    \end{equation}
\end{lemma}
\begin{proof}
    % By Euler's homogeneous function theorem, we have 
    % \begin{equation*}
    %     d_i\left( \frac{1}{m D_i} \cdot \mathbf{1}_m \right) = \sum_{j=1}^m \frac{\partial d_i(\frac{1}{m D_i} \cdot \mathbf{1}_m)}{\partial x_{ij}} \frac{1}{m D_i} \leq \sum_{j = 1}^m D_i \frac{1}{m D_i} = 1. 
    % \end{equation*}
    % It follows that, 
    For any price vector $\bfp \in \RR^m_+$ that satisfies $\sum_{j=1}^m p_j = \sum_{i=1}^n B_i$, we have 
    \begin{equation*}
        \max_{\bfx_i \geq 0: d_i(\bfx_i) \leq 1} \inp{\bfp}{\bfx_i} \geq \inp*{\bfp}{\delta_i \cdot \mathbf{1}_m} = \delta_i \norm{\bfB}_1. 
    \end{equation*}
    Therefore, the objective of~\eqref{pgm:general-eg-dual} is lower bounded by $- \sum_{i=1}^n B_i + \sum_{i=1}^n B_i \log(\delta_i \norm{\bfB}_1) \geq \norm{\bfB}_1 \log\left( \delta \norm{\bfB}_1 / e \right)$ where $\delta = \min_{i \in [n]} \delta_i$. 
\end{proof}

\section{Discrete-Time Relative t\^atonnement}\label{sec:discrete}

In the computer science literature, a major focus has been on proving convergence guarantees for the discrete-time variant of t\^atonnement. 
In general, it is common for continuous-time gradient descent procedures to have better convergence properties than their discrete-time analogues, and thus it is natural to ask whether the discrete-time variant of relative t\^atonnement still retains nice convergence properties~\cite{absil2005convergence}.
The discrete-time version of~(\ref{chores-tatonnement}) is as follows: 
starting from any initial price vector $\bfp^0 \in \Delta_B$, the iterates are updated according to  
\begin{equation}
    \begin{aligned}
        &\mbox{Pick } \tilde{\bfz}^k \in \tilde{Z}(\bfp^k) \\ 
        &\bfp^{k + 1} = \bfp^k - \eta^k \tilde{\bfz}^k
    \end{aligned}, 
    % p^{k + 1} = p^k - \eta^k \tilde{\zeta}^k, \quad \tilde{\zeta}^k(p^k) \in \tilde{z}^k \hspace{20pt} \forall\; k = 0, 1, 2, \ldots, 
    \tag{Discrete-time relative t\^atonnement}
    \label{discrete-time-relative-chores-tatonnement}
\end{equation} 
where $\{ \eta^k \}_{k = 0, 1, 2, \ldots}$ is a sequence of nonnegative stepsizes. 
% Throughout we assume that $p^0 \in \Hb$.
% \tianlong{\\ 
% \underline{Please check if we can/should say this}: 
% \\
% Note that our discrete-time dynamics does not consider any price projection operation, which is a relatively faithful price-adjustment process in response to excess demand signals.}
Note that, apart from the fact that we use relative demand, our t\^atonnement process completely avoids projection, thus rendering it a fairly natural price-adjustment process. We do not need projection onto the positive orthant because we will show that prices naturally stay away from zero.

For discrete-time na\"ive t\^atonnement dynamics, one can show that the divergence in~\cref{sec: naive unstable} is preserved for any choice of stepsizes.
\begin{lemma}\label{lem:divergence}
    For any $\rho \in [1, \infty)$, 
    there is a $\rho$-CES chores Fisher market with a unique CE price vector $\bfp^*$, 
    such that for every initial price vector $\bfp^0 \in \{ \bfp \in \RR^m_+ \mid \sum_{j=1}^m p_j \geq \sum_{i=1}^n B_i, \bfp \neq \bfp^* \}$, the discrete-time na\"ive t\^atonnement dynamics equipped with \emph{any} sequence of positive stepsizes $\{ \eta^t \}_{t \geq 0}$ and \emph{any} tie-breaking rule diverges from $\bfp^*$.
    \label{lem:discrete-time-naive-ttm-divergence}
\end{lemma}

% \textcolor{red}{Have only the continuous time and have a pointer to the discrete time variant in Section 4 -- so re-discuss the example in the discrete time variant in Section 4 and then talk about stepsize selection.}

\begin{proof}
    For each $\rho \in [1, \infty)$, we consider the same instance as in~\cref{lem:continuous-time-naive-ttm-divergence}.

    For $\rho = 1$, consider a price vector $\bfp \in \RR^2_+$ and an updated price vector $\bfp^+$ in the next step.
    If the current price satisfies $p_1 + p_2 \geq 1$ and $\bfp \neq \bfp^*$, then 
    \begin{equation*}
        \begin{aligned}
            \Delta(p_1 + p_2) := \frac{1}{\eta^t} (p^+_1 + p^+_2 - p_1 - p_2) = - (z_1 + z_2) = \begin{cases}
            2 - \frac{1}{p_1} & p_1 > p_2 \geq 0 \\ 
            2 - \frac{1}{p} & p_1 = p_2 = p > 0 \\ 
            2 - \frac{1}{p_2} & p_2 > p_1 \geq 0, 
        \end{cases}
        \end{aligned}
    \end{equation*}
    where $(z_1, z_2) \in Z(\bfp)$.
    Equivalently, $\Delta(p_1 + p_2) = 2 - \frac{1}{\max\{ p_1, p_2 \}}$.
    By the same argument as in the continuous-time counterpart, $p^t_1 + p^t_2 > p^{t - 1}_1 + p^{t - 1}_2 > \cdots > p^0_1 + p^0_2 \geq 1$ for all $t \geq 1$ for any sequence of positive stepsizes $\{ \eta^t \}_{t \geq 0}$. This implies that the na\"ive t\^atonnement dynamics diverge away from $\bfp^*$.

    For $1 < \rho < \infty$, by defining  
    $\Delta(p_1) := \frac{1}{\eta^t} (p^+_1 - p_1) $ and 
    $\Delta(p_1) := \frac{1}{\eta^t} (p^+_2 - p_2)$, 
    we have 
    \begin{equation*}
        \Delta(p_1 + p_2) = 2 - \frac{p_1^{1/(\rho - 1)} + p_2^{1/(\rho - 1)}}{p_1^{\rho / (\rho - 1)} + p_2^{\rho / (\rho - 1)}}. 
    \end{equation*}
    Then, the divergence of na\"ive t\^atonnement follows by the same argument as in~\cref{lem:continuous-time-naive-ttm-divergence}.
\end{proof}

First, we will show that discrete-time relative t\^atonnement converges to CE with a sequence of shrinking stepsizes for CCH and strictly increasing disutilities. 
Then, we show that for a large class of CCH disutility functions, including convex CES disutility functions with $\rho \in (1,\infty)$, one can get an $\mathcal{O}\left( \frac{1}{\epsilon^2} \right)$ convergence rate. 
% This iteration complexity matches the best-known complexity for general convex disutilities. 

\subsection{General CCH Disutilities}\label{subsec:CCH-Discrete}

First, we show that discrete-time relative t\^atonnement converges to CE for strictly increasing CCH disutilities, once we choose appropriate stepsizes. 

\begin{lemma} 
    Let $\{ \bfp^k \}_{k = 0, 1, 2, \ldots}$ be a sequence of iterates generated by~(\ref{discrete-time-relative-chores-tatonnement}) with any initial point $\bfp^0 \in \Delta_B$ 
    and $\eta^k \leq \frac{\ell_0^2}{2 \norm{\bfB}_1}$ for all $k$, where $\ell_0 = \frac{\norm*{\bfB}_1}{3m} \min_{i \in [n]} \nu_i$ and $\nu_i$ is defined in~\cref{eq:nu defn}. 
    Then, we have 
    \begin{itemize}
        \item $\bfp^k \in \Delta_B$ for all $k \geq 0$, 
        \item there exists an index $k_0 \geq 0$ such that $p^k_j \geq \frac{\ell_0}{2}\; \forall\, j \in [m]$ for all $k \geq k_0$, and 
        \item $p^{k+1}_j > p^k_j + \eta^k\cdot\frac{1}{6m}$ if $p^k_j \leq \ell_0$ for all $k \geq 0$. 
    \end{itemize}
    \label{lem:discrete-time-relative-tatonnement-properties}
\end{lemma}

\begin{proof} 
    The above lemma follows from the following three facts. 
    Informally, these points create a discrete-time threshold that prevents any price from approaching zero, provided the step size is chosen small enough.
    \begin{enumerate}
        \item For any $\bfp^k \in \Delta_B$, $\bfp^{k+1} \in H_B$; 
        \item For any $\bfp^k \in \Delta_B$ and each $j \in [m]$, if $p^k_j \leq \ell_0$ then $p^{k+1}_j > p^k_j + \eta^k\cdot\frac{1}{6m} $; 
        \item For any $\bfp^k \in \RR^m_+$ and each $j \in [m]$, if $p^k_j > \ell_0$, then $p^{k + 1}_j > \frac{\ell_0}{2}$. 
    \end{enumerate} 
    % Note that the set $z(p)$ is bounded if $p \in \Delta_B$, i.e., $\norm{\zeta} < \infty$ for any $\zeta \in z(p)$ if $p \in \Delta_B$. 
    The first fact is true because $\sum_{j = 1}^m \tilde{z}_j = \sum_{j = 1}^m (z_j - \frac{1}{m}\sum_{j' = 1}^m z_{j'}) = 0$ for any $\tilde{\bfz} \in \tilde{Z}(\bfp)$ 
    % whenever $z(p)$ is bounded, 
    where $\bfz$ is the excess demand supporting $\tilde{\bfz}$. 
    The second fact follows from~\cref{lem:excess-demand-barrier}. 
    The third fact is true because of the following: for any $\bfz^k \in Z(\bfp^k)$, we have $z^k_j \leq \frac{\norm{\bfB}_1}{p_j^k} - 1 \leq \frac{\norm{\bfB}_1}{\ell_0} - 1$ and $\frac{1}{m}\sum_{j=1}^m z^k_j \geq -1$. Thereby, $\tilde{z}^k_j \leq z^k_j - (-1) \leq \frac{\norm{\bfB}_1}{\ell_0}$. Thus, $p^{k + 1}_j = p^k_j - \eta^k \tilde{z}^k_j > \ell_0 - \frac{\ell_0^2}{2\norm{\bfB}_1} \frac{\norm{\bfB}_1}{\ell_0} = \frac{\ell_0}{2}$. 
    
    We then prove~\cref{lem:discrete-time-relative-tatonnement-properties}. 
    From the above three facts, we know 
    if $\bfp^k \in \Delta_B$ then $p^{k+1}_j > \min\{ p^k_j, \frac{\ell_0}{2} \}$, which means $\bfp^{k + 1} \in \Hb \cap \RR^m_{++} \subset \Delta_B$. 
    By induction the first statement in~\cref{lem:discrete-time-relative-tatonnement-properties} is correct. 
    % In the first case, $\min_j p^0_j \geq \ell$. 
    % In this case, let $k_0 = 0$. 
    % Recall that $\ell = \frac{\norm{B}_1}{m} \min_{i,j,j'} \frac{d_{ij}}{d_{i j'}}$. 
    % Combining the above facts, we can prove this lemma by induction: 
    % Suppose that for all $0 \leq k' \leq k$, we have the above two facts. Then, because 
    % \begin{equation*}
    %     \sum_{j = 1}^m \tilde{\zeta}_j = 0, \quad \forall\; \tilde{\zeta} \in \tilde{z}(p), \; p \in \mathbb{R}^m_+ 
    % \end{equation*} 
    % and $p^k \in \mathbb{R}^m_+ \cap H_B$, 
    % $p^{k+1}$ lie in $\Delta_B$. 
    Then, let us assume there exists a set of chores $J$ such that $p^0_j \in [0, \ell_0)$ for any $j \in J$. 
    Because $p^{k + 1}_j > p^k_j + \eta^k \cdot \frac{1}{6m}$ if $p^k_j < \ell_0$ and $p^k \in \Delta_B$, 
    there is a time point $\kappa_j$ for each $j \in J$ such that it exceeds the barrier $\ell_0$ for the first time, i.e., $\kappa_j := \min\{ k \,\mid\, p^k_j > \ell_0 \}$.
    The above facts then imply that $p^k_j \geq \frac{\ell_0}{2}$ for all $k \geq \kappa_j$. 
    Hence, letting $k_0 := \max_{j \in J} \kappa_j$, we show the second statement in~\cref{lem:discrete-time-relative-tatonnement-properties} is correct. 
    The third statement was proved as the second fact above.
\end{proof}

Following the general framework of stochastic approximation (see, for example,~\cite{davis2020stochastic} and references therein), 
we next need to isolate conditions under which the discrete-time sequence can be seen as an approximation to the continuous-time trajectory of the dynamical system in~\eqref{chores-tatonnement}. 
A set of such conditions was stipulated by~\cite[Assumption A]{davis2020stochastic} for the general stochastic subgradient descent method. Here we specialize their results~\cite[Corollary 5.5]{davis2020stochastic} to our setting, which allows us to leave out one of their assumptions dealing with the stochastic setting.
\begin{proposition}[Corollary 5.5 of \cite{davis2020stochastic}]
    \label{davis assumptions}
    If the iterates of the subgradient descent method satisfy the following:
    \begin{enumerate}
        \item All limit points of $\{\bfp^k\}$ lie in $\Delta_B$.
        \item The iterates are bounded: $\sup_{k \geq 1} \norm{\bfp^k} < \infty$ and $\sup_{k\geq 1} \sup_{\tilde{\bfz}^k \in \tilde{Z}(\bfp^k)} \norm{\tilde \bfz^k} < \infty$.
        \item The stepsize sequence is nonnegative and square summable, but not summable.
        \item For any unbounded increasing sequence $\{ k_j \}_{j = 0, 1, 2, \ldots} \subset \mathbb N$ such that $\bfp^{k_j}$ converges to some point $\hat{\bfp}$, it holds
        $\lim_{n\rightarrow \infty} \operatorname{dist}\left(\frac{1}{n} \sum_{j=1}^n \tilde{Z}(\bfp^{k_j}), \tilde{Z}(\hat{\bfp}) \right) = 0. 
        $
    \end{enumerate}
    and 
    the objective function satisfies that (i) the function is locally Lipschitz and subdifferentially regular and (ii) the set of its non-stationary values is dense in $\RR$. 
    Then, every limit point of the iterates generated by the subgradient method is a stationary point, and the function values converge. 
\end{proposition}

% We do not list an assumption for stochastic noise as we do not have noise in our setting and that assumption trivially holds. 

% \begin{assumption}[Assumption A in~\cite{davis2020stochastic}] 
%     {\color{white} .} 
%     \begin{enumerate}
%         \item All limit points of $\{ p^k \}$ lie in $H_B$; 
%         \item The iterates are bounded, i.e., $\sup_{k\geq 0} \norm{p^k} < \infty$ and $\sup_{k\geq 0} \max_{z^k \in z(p^k)} \norm{z^k} < \infty$; 
%         \item The sequence $\{\eta^k\}$ is nonnegative, square summable, but not summable: 
%         \begin{equation*}
%             \eta^k \geq 0, \quad \sum_{k = 0}^\infty \eta^k = \infty \quad \text{and} \quad \sum_{k = 0}^\infty (\eta^k)^2 < \infty; 
%         \end{equation*}
%         \item For any unbounded increasing sequence $\{k_j \} \subset \mathbb{N}$ such that $p^{k_j}$ converges to some point $\bar{p}$, it holds: 
%         \begin{equation*}
%             \lim_{n \rightarrow \infty} \textnormal{dist}\left( \frac{1}{n} \sum_{j = 1}^n z^k, z(\bar{p}) \right) = 0. 
%         \end{equation*}
%     \end{enumerate}
% \end{assumption} 

By choosing appropriate stepsizes
and using the properties of $f\vert_{H_B}$ shown in~\cref{subsec:Relative chores t\^atonnement is stable}
, we can show that the assumptions in~\cref{davis assumptions} hold for the iterates of~\eqref{discrete-time-relative-chores-tatonnement}. 
This allows us to conclude convergence of~\eqref{discrete-time-relative-chores-tatonnement}. See~\cref{app:sec:Omitted proofs} for the proof. 

\begin{theorem}
    In a chores Fisher market with CCH disutilities satisfying~\cref{assump:strictly-increasing}, let $\{ \bfp^k \}_{k = 0, 1, \ldots}$ be a sequence of iterates generated by~\cref{discrete-time-relative-chores-tatonnement} with $\{ \eta^k \}_{k = 0, 1, \ldots}$ and any initial point $\bfp^0 \in \Delta_B$, 
    where the stepsizes satisfy 
    \begin{equation*}
        0 \leq \eta^k \leq \frac{\ell_0^2}{2 \norm{\bfB}_1} \; \forall\, k = 0 , 1, \ldots, \quad \sum\nolimits_{k = 0}^\infty \eta^k = \infty \quad \text{and} \quad \sum\nolimits_{k = 0}^\infty (\eta^k)^2 < \infty. 
    \end{equation*} 
    % Moreover, $\sup_{k \geq 0} \{ \eta^k \} \leq \frac{\ell^2}{2 \norm{B}_1}$. 
    Then, every limit point of the iterates $\{ \bfp^k \}_{k = 0, 1, \ldots}$ is a stationary point of $f\vert_{H_B}$ and the function values $\{ f\vert_{H_B} (\bfp^k)\}_{k = 0, 1, \ldots}$ converge. 
    \label{thm:discrete-time-relative-tatonnement-convergence}
\end{theorem} 

\subsection{Convergence rate guarantees for CES disutilities}\label{sec:poly-time}

For a large class of convex CES disutility functions with $\rho \in (1,\infty)$, we show that the relative t\^atonnement dynamics can find an $\varepsilon$-approximate CE in $\mathcal{O}\left( \frac{1}{\varepsilon^2} \right)$ iterations. 
This is the first result guaranteeing a convergence rate that is polynomial in $\frac{1}{\varepsilon}$ for any interesting disutility class beyond linear utilities.

For a proper, nonnegative, convex, $1$-homogeneous function $h: \RR^m \rightarrow \RR_+$ satisfying $h(\mathbf{0}_m) = 0$ (called a \emph{gauge function}), we define the gauge dual $h^\circ$ by 
\begin{equation}
    h^\circ(\bfp) = \sup_{\bfx}\{ \inp*{\bfp}{\bfx} \mid h(\bfx) \leq 1 \}. 
    \label{eq:def-gauge-dual}
\end{equation}

Given a CCH disutility function $d_i$, we define its gauge dual as 
\begin{equation*}
    d^\circ_i(\bfp) := \max_{\bfx_i}\{ \inp{\bfp}{\bfx_i} \mid d_i(|\bfx_i|) \leq 1 \}, 
\end{equation*}
where $| \cdot |$ denotes the vector of the component-wise absolute values of a vector.
Here, $d_i(|\cdot|)$ is the extension of $d_i$ from $\RR_+^m$ to $\RR^m$.
% This dual function coincides with the earning maximization problem when $\bfp \geq 0$.
By leveraging this gauge dual, we establish global and local smoothness of the potential function of~\eqref{discrete-time-relative-chores-tatonnement}.

In the following subsections, we first establish global smoothness and local smoothness for different ranges of disutility functions in~\cref{subsubsec:global-smoothness}  (for $\rho\in (1,2]$), and \cref{subsubsec:local-smoothness} (for $\rho\in (2,\infty)$). 
In~\cref{subsubsec:smooth-convergence}, we show that the smoothness guarantees from the previous subsections leads to the desired iteration complexity.

\subsubsection{Global smoothness properties of disutilities}
\label{subsubsec:global-smoothness}

Before presenting our results, we first formally define the concept of smoothness and some strongly-convex-like concepts.
We say a function $h$ is $L$-smooth if $h$ is differentiable and 
\begin{equation*}
    \lVert \nabla h(\bfx_1) - \nabla h(\bfx_2) \rVert \leq L \norm*{ \bfx_1 - \bfx_2} 
    \hspace{10pt} \forall\, \bfx_1, \bfx_2 \in \textnormal{dom }(h).
\end{equation*}
A function $h$ is \emph{$\mu$-strongly convex} if 
\begin{equation*}
    \lambda h(\bfx_2) + (1 - \lambda) h(\bfx_1) \ge h\left(\lambda \bfx_2 + (1 - \lambda) \bfx_1\right) + \frac{\mu}{2} \norm{\bfx_2 - \bfx_1}^2 \hspace{10pt} \forall\, \lambda \in (0, 1) \hspace{10pt} \forall\, \bfx_1, \bfx_2 \in \textnormal{dom }(h).
\end{equation*}
A function $h$ is \emph{$\mu$-strongly quasi-convex} if
\begin{equation*}
    \max\{ h(\bfx_2), h(\bfx_1) \} \ge h\left(\lambda \bfx_2 + (1 - \lambda) \bfx_1\right) + \frac{\mu}{2} \norm{\bfx_2 - \bfx_1}^2 \hspace{10pt} \forall\, \lambda \in (0, 1) \hspace{10pt} \forall\, \bfx_1, \bfx_2 \in \textnormal{dom }(h).
\end{equation*}
Strong quasi-convexity is a strictly weaker notion than strong convexity; any $\mu$-strongly convex function is $\mu$-strongly quasi-convex by definition.
A closed convex set $C$ is \emph{$\mu$-strongly convex} if, for any $\bfx_1, \bfx_2 \in C$ and $\lambda \in [0,1]$,
\begin{equation*}
    \mathcal{B}_{\frac{\mu}{2} \lambda(1-\lambda) \norm{ \bfx_1 - \bfx_2 }^2}\left(\lambda \bfx_1 + (1-\lambda)\bfx_2\right) \subseteq C.
\end{equation*}
Recall that we use  $\mathcal{B}_r(\bfx)$ to denote the closed Euclidean ball centered at $\bfx$ with radius $r$.

For a proper, closed, convex function 
$h$, strong convexity and smoothness are dual under standard Fenchel conjugacy:
$f$ is $\mu$-strongly convex ($\mu>0$) \emph{if and only if} its conjugate $f^*$ is differentiable with $(1/\mu)$-Lipschitz gradient~\cite[Theorem 5.26]{beck2017first}.
However, 
many CCH functions are not even strongly quasi-convex. 
For example, the gauge functions corresponding to convex CES disutility functions with $\rho \in [1, \infty)$ are not.
% To see this point, consider a one-dimensional convex CES disutility function with $\rho \in [1, \infty)$ and $d_{1} = 1$

Next, we prove global smoothness of the gauge dual for a broad class of CCH disutility functions, even though the disutility functions themselves are not strongly convex.
In particular, we find a specific smoothness constant for any two price vectors $\bfp$ and $\bfp'$ that lie in the whole space of the price simplex $\Delta_B$. 

% \begin{definition}
%     An agent is said to have \emph{smooth demand} if the gauge dual of her disutility function $d(p)$ is smooth on a domain of $H_B$, i.e., 
%     \begin{equation*}
%         \lVert \nabla d^\circ(p_1) - \nabla d^\circ(p_2) \rVert \leq L \norm*{ p_1 - p_2 }, \hspace{30pt} \forall\, p_1, p_2 \in H_B. 
%     \end{equation*}
% \end{definition}

\begin{theorem}
    If there exists some power $a \geq 1$ such that $d(|\bfx|)^a$ is $\mu$-strongly quasi-convex in a region containing $S^1 = \{ \bfx: d(|\bfx|) \leq 1 \}$, 
    and $\norm*{\nabla ( d(|\bfx|) )^a} \leq M_1$ for all $\bfx \in S^1$, 
    then for any $\bfp_1, \bfp_2 \in \Delta_B$, we have 
    \begin{equation}
        \lVert \nabla d^\circ(\bfp_1) - \nabla d^\circ(\bfp_2) \rVert \leq \frac{\sqrt{m} M_1}{\mu \norm{\bfB}_1} \norm{ \bfp_1 - \bfp_2 }. 
        \label{eq:smoothness-for-a-class-of-disutility-functions}
    \end{equation}
    \label{thm:smooth-strongly-quasiconvex}
\end{theorem}
\begin{proof}
    Observe that 
    \begin{equation*}
        S^1 = \{ \bfx: d(|\bfx|) \leq 1 \} = \{ \bfx: d(|\bfx|)^a \leq 1 \}. 
    \end{equation*}
    % Furthermore, $\norm{\nabla (d(|\bfx|)^a)} = \norm*{a d(|\bfx|)^{a-1} \nabla d(|\bfx|)} \leq a M_1$ for all $\bfx \in S^1$. 
    By~\cite[Corollary 1]{vial1982strong}, we have 
    % A convex function is said to be in this class if its unit level set \tianlong{There is no sign constraint in the definition of sublevel set}
    % \begin{equation*}
    %     C = \{ x: d(x) \leq 1 \} 
    % \end{equation*}
    % is $\mu$-strongly convex. 
    the sublevel set $S^1$ is $\frac{M_1}{2\mu}$-strongly convex. 
    By~\cite[Theorem 2.1 (i)]{goncharov2017strong}, 
    the gauge dual $\max_{\bfx \in S^1} \inp*{\bfp}{\bfx}$ is differentiable when $\norm{\bfp} \leq 1$ and 
    \begin{equation*}
        \lVert \nabla d^\circ(\bfp_1) - \nabla d^\circ(\bfp_2) \rVert \leq \frac{M_1}{2\mu} \lVert \bfp_1 - \bfp_2 \rVert 
    \end{equation*}
    when $\norm{\bfp_1} = \norm{\bfp_2} = 1$. 
    Note that 
    $d^\circ(\bfp)$ is $1$-homogeneous w.r.t. $\bfp$, thereby 
    $\nabla d^\circ(\bfp)$ is $0$-homogeneous\footnote{A vector-valued or matrix-valued function $\mathbf{f}(\bfx)$ is said to be $k$-homogeneous for some integer $k$ if $\mathbf{f}(a \bfx) = a^k \mathbf{f}(\bfx)$ for any $a > 0$.}. 
    Hence, for any two $\bfp_1, \bfp_2 \in \Delta_B$, we have 
    \begin{align*}
        \norm{ \nabla d^\circ(\bfp_1) - \nabla d^\circ(\bfp_2) } 
        &= \norm*{ \nabla d^\circ\left(\frac{\bfp_1}{\norm{\bfp_1}}\right) - \nabla d^\circ\left(\frac{\bfp_2}{\norm{\bfp_2}}\right) } \\ 
        &\leq \frac{M_1}{2\mu} \norm*{ \frac{\bfp_1}{\norm{\bfp_1}} - \frac{\bfp_2}{\norm{\bfp_2}} } \\ 
        &\leq \frac{M_1}{2\mu} \norm*{ \frac{\bfp_1}{\norm{\bfp_1}} - \frac{\bfp_2}{\norm{\bfp_1}} } + \frac{M_1}{2\mu} \norm*{ \frac{\bfp_2}{\norm{\bfp_1}} - \frac{\bfp_2}{\norm{\bfp_2}} } \\ 
        % &= \frac{\mu}{\norm{p_1}} \norm*{ p_1 - p_2 } + \mu \left\lvert \frac{1}{\norm{p_1}} - \frac{1}{\norm{p_2}} \right\rvert \norm*{p_2} \\ 
        &= \frac{M_1}{2\mu} \frac{\norm*{ \bfp_1 - \bfp_2 }}{\norm{\bfp_1}} + \frac{M_1}{2\mu} \frac{\left\lvert \norm{\bfp_1} - \norm{\bfp_2} \right\rvert}{\norm*{\bfp_1}} \\ 
        &\leq \frac{M_1}{\mu\norm{\bfp_1}} \norm*{ \bfp_1 - \bfp_2 } \leq \frac{\sqrt{m} M_1 }{\mu \norm{\bfB}_1} \norm*{ \bfp_1 - \bfp_2 }, 
    \end{align*}
    where the last inequality holds because $\norm*{\bfp_1} \geq \frac{1}{\sqrt{m}} \norm*{\bfp_1}_1 = \frac{1}{\sqrt{m}} \norm*{\bfB}_1$. 
    % This implies the $2\sqrt{m}\mu$-smoothness of $d^\circ(p)$. 
    % Note that this class of functions is strictly larger than strongly convex functions~\citet[Corollary 2]{vial1982strong}. 
\end{proof}

Next, we show the unit sublevel set of a CES disutility function is bounded.
\begin{lemma}\label{lem:poly-time-rho-le-2}
    Let $d_i(\bfx_i)$ be a CES disutility function with $\rho \in (1, \infty)$.  
    % and $d_i^\circ(p)$ be defined as in~\cref{eq:def-gauge-dual}. 
    Then, we have $R_i:= \sup\{ \norm*{\bfx_i}_\infty \mid d_i(|\bfx_i|) \leq 1 \} \leq \max_{j \in [m]} d_{ij}^{-{1}/{\rho}} = (\min_{j \in [m]} d_{ij})^{-1/\rho}$. 
    \label{lem:boundedness-x-sublevel}
\end{lemma}
\begin{proof}
    % The smoothness of $d^\circ(p)$ implies $d^\circ(p)$ is differentiable. 
    % Let $x' = \textnormal{argmax}_{x: d(x) \leq 1} \inp{p}{x}$. 
    % The feasibility implies that $d_j | x_j |^\rho \leq 1$ for all $j$. 
    % $\norm*{\nabla d^\circ(p)} = \norm*{x'} \leq $. 
    \begin{align*}
        R_i &= \max\left\{ \max_{j \in [m]} | x_{ij} | \left\vert (\sum\nolimits_{j=1}^m d_{ij} \left| x_{ij} \right|^{\rho})^{\frac{1}{\rho}} \leq 1 \right. \right\} \\ 
        &\leq \max\left\{ \max_{j \in [m]} | x_{ij} | \mid d_{ij} \left| x_{ij} \right|^{\rho} \leq 1, \; \forall\, j \in [m] \right\} \\ 
        &= \max_{j \in [m]} d_{ij}^{-{1}/{\rho}} = (\min_{j \in [m]} d_{ij})^{-1/\rho}.
        \qedhere
    \end{align*} 
\end{proof}

It was shown in \cite{juditsky2008large} that the squared $p$-norm with $p \in (1, 2]$ is strongly convex. 
Analogously, we prove the squared CES (i.e., squared \emph{weighted} $p$-norm) disutility function with $\rho \in (1, 2]$ is strongly convex.
% , which implies the smoothness of the gauge dual on $\Delta_B$ by~\cref{thm:smooth-strongly-quasiconvex}. 
\begin{lemma}
  \label{lem:squared_rho=2}
    The squared CES disutility function $( d_i(\bfx_i) )^2 = (\sum_{j=1}^m d_{ij} x_{ij}^\rho)^\frac{2}{\rho}$ with $\rho \in (1, 2]$ is $2(\rho - 1) (\min_{j} d_{ij})^{\frac{2}{\rho}}$-strongly convex in $C_i := \{ \bfx_i: \norm*{\bfx_i}_\infty \leq R_i \} \supset S^1_i$. 
    Furthermore, $\norm*{\nabla ( d_i(\bfx_i) )^2} \leq 2(\sum_{j=1}^m d_{ij}^{\frac{2}{\rho}})^{\frac{1}{2}}$ for all $\bfx_i \in S_i^1 = \{ \bfx_i: d_i(|\bfx_i|) \leq 1 \}$.
    \label{lem:squared-ces-1-2-disutility-properties}
\end{lemma}
\begin{proof}
    Because $( d_i(\bfx_i) )^2$ is twice continuously differentiable in $\mathbb{R}^m_+$, by~\cite[Theorem 2.1.11]{nesterov2018lectures}, to show $\mu$-strong convexity, it suffices to prove that 
    \begin{equation}
        \inp*{\nabla^2 ( d_i(\bfx_i) )^2 \bfh}{\bfh} \geq \mu \norm*{\bfh}^2 \hspace{30pt} \forall\, \bfx_i \in \mathbb{R}^m_+ \text{ and } \bfh \in \mathbb{R}^m. 
        \label{eq:twice-differential-strong-convexity-equiv}
    \end{equation}
    Note that $(d_i(\bfx_i))^2$ is $2$-homogeneous thus $\nabla^2 (d_i(\bfx_i))^2$ is $0$-homogeneous. 
    Hence, it suffices to prove the relation in~\cref{eq:twice-differential-strong-convexity-equiv} for $\bfx_i \in \mathbb{R}^m_+$ normalized to $\sum_{j=1}^m d_{ij} x_{ij}^\rho = 1$. 
    By calculation, we have that $\nabla_j (d_i(\bfx_i))^2 = 2(\sum_{j'=1}^m d_{ij'} x_{ij'}^\rho)^{\frac{2}{\rho} - 1} d_{ij} x_{ij}^{\rho - 1} \;\forall\, j \in [m]$ 
    and  
    \begin{align*}
        \nabla^2_{jj'} (d_i(\bfx_i))^2 =& 2 d_{ij} x_{ij}^{\rho - 1} \left( \frac{2}{\rho} - 1 \right) (\sum_{\ell=1}^m d_{i\ell} x_{i\ell}^\rho)^{\frac{2}{\rho} - 2} \rho d_{ij'} x_{ij'}^{\rho - 1} \hspace{135pt} \forall\, j, j' \in [m] \\ 
        \nabla^2_{jj} (d_i(\bfx_i))^2 =& 2 \left( \frac{2}{\rho} - 1 \right) (\sum_{\ell=1}^m d_{i\ell} x_{i\ell}^\rho)^{\frac{2}{\rho} - 2} \rho d_{ij}^2 x_{ij}^{2(\rho - 1)} + 2 (\sum_{\ell=1}^m d_{i\ell} x_{i\ell}^\rho)^{\frac{2}{\rho} - 1} d_{ij} (\rho - 1) x_{ij}^{\rho - 2} \hspace{10pt} \forall\, j \in [m]. 
    \end{align*}
    Plugging in $\sum_{\ell=1}^m d_{i\ell} x_{i\ell}^\rho = 1$, we obtain that 
    \begin{align*}
        \inp*{\nabla^2 (d_i(\bfx_i))^2 \bfh}{\bfh} &= \sum_{j=1}^m\sum_{j'=1}^m 2\left(\frac{2}{\rho} - 1\right) \rho (d_{ij} x_{ij}^{\rho - 1} h_j) (d_{ij'} x_{ij'}^{\rho - 1} h_{j'}) + \sum_{j=1}^m 2(\rho - 1) d_{ij} x_{ij}^{\rho - 2} h_j^2 \\ 
        &= 2(2 - \rho) (\sum_{j=1}^m d_{ij} x_{ij}^{\rho - 1} h_j)^2 + 2(\rho - 1) \sum_{j=1}^m d_{ij} x_{ij}^{\rho - 2} h_j^2 \\ 
        &\geq 2(\rho - 1) (\min_j d_{ij}) \left((\min_{j} d_{ij})^{-\frac{1}{\rho}}\right)^{\rho - 2} \norm{\bfh}^2 = 2(\rho - 1) (\min_{j} d_{ij})^{\frac{2}{\rho}}\norm{\bfh}^2,  
    \end{align*}
    where the last inequality follows by the non-negativity of the first term in the second line, $\rho \in (1, 2]$ and~\cref{lem:boundedness-x-sublevel}. 
    Furthermore, for all $\bfx_i \in S_i^1$, we have 
    \begin{equation*}
        \left| \nabla_j (d_i(\bfx_i))^2 \right|
        = \left| 2\left(\sum_{j'=1}^m d_{ij'} x_{ij'}^\rho\right)^{\frac{2}{\rho} - 1} d_{ij} x_{ij}^{\rho - 1} \right| \leq 2 d_{ij} \left( d_{ij}^{-\frac{1}{\rho}} \right)^{\rho - 1} = 2 d_{ij}^{\frac{1}{\rho}},
    \end{equation*}
    where we use $\sum_{j'=1}^m d_{ij'} x_{ij'}^\rho \leq 1$ and $x_{ij} \leq d_{ij}^{-\frac{1}{\rho}}$ for all $j \in [m]$ (implied by $\sum_{j'=1}^m d_{ij'} x_{ij'}^\rho \leq 1$) in the second inequality. 
    Hence, we have $\norm{\nabla (d_i(\bfx_i))^2} = \left(\sum_{j=1}^m \left| \nabla_j (d_i(\bfx_i))^2 \right|^2\right)^{\frac{1}{2}} \leq 2 \big( \sum_{j=1}^m d_{ij}^{\frac{2}{\rho}} \big)^{\frac{1}{2}}$.
\end{proof}

As a result, we can characterize the smoothness of the gauge dual if an agent has convex CES disutility with $\rho \in (1, 2]$. 
\begin{corollary}
  \label{coro:rho=2}
    Let $d_i(\bfx) = (\sum_{j=1}^m d_{ij} x_{ij}^\rho)^\frac{1}{\rho}$ with $\rho \in (1, 2]$. 
    For any $\bfp_1, \bfp_2 \in \Delta_B$, 
    it holds that 
    \begin{equation*}
        \norm*{\nabla d_i^\circ(\bfp_1) - \nabla d_i^\circ(\bfp_2)} \leq \frac{2 m(\max_j d_{ij})^{\frac{1}{\rho}}}{(\rho - 1) (\min_{j} d_{ij})^{\frac{2}{\rho}} \norm{\bfB}_1} \norm{\bfp_1 - \bfp_2}. 
    \end{equation*}
    \label{crl:smooth-CES-rho-1-2}
\end{corollary}
\begin{proof}
    By~\cref{lem:squared-ces-1-2-disutility-properties}, we have~\cref{eq:smoothness-for-a-class-of-disutility-functions} in~\cref{thm:smooth-strongly-quasiconvex} holds with $a=2$, $\mu = 2(\rho - 1) (\min_{j} d_{ij})^{\frac{2}{\rho}}$ and $M_1 = 2(\sum_{j=1}^m d_{ij}^{\frac{2}{\rho}})^{\frac{1}{2}}$, yielding 
    \begin{align*}
        \norm*{\nabla d_i^\circ(\bfp_1) - \nabla d_i^\circ(\bfp_2)} \leq \frac{2\sqrt{m}(\sum_{j=1}^m d_{ij}^{2/\rho})^{1/2}}{2(\rho - 1) (\min_{j} d_{ij})^{\frac{2}{\rho}} \norm*{\bfB}_1} \norm*{\bfp_1 - \bfp_2} \leq \frac{m(\max_j d_{ij})^{\frac{1}{\rho}}}{(\rho - 1) (\min_{j} d_{ij})^{\frac{2}{\rho}} \norm{\bfB}_1} \norm{\bfp_1 - \bfp_2}.
    \end{align*}
\end{proof}

\begin{remark}
    We emphasize that 
    our result in~\cref{thm:smooth-strongly-quasiconvex} does capture a broader range of CCH disutility functions than the convex CES functions with $\rho \in (1, 2]$. 
    For example, 
    it covers another interesting class of CCH functions $d_i(\bfx_i) = \norm{A \bfx_i}$ where $A$ is a full rank matrix.
    In this case, an agent's disutility is given by the $\ell_2$ norm of a load vector, where each component of this vector is a linear combination of the chores assigned to the agent.
\end{remark}

\subsubsection{Local smoothness for convex CES function with $\rho \in (2, \infty)$}
\label{subsubsec:local-smoothness}

For the convex CES function with $\rho \in (2, \infty)$, however, even the squared $p$-norm function is not strongly quasi-convex. 
That is, we cannot derive the smoothness of the gauge dual of the disutility functions with $\rho \in (2, \infty)$ by leveraging~\cref{thm:smooth-strongly-quasiconvex}. 
This aligns with our observation in~\cref{fig:strongly-convex-set-and-smoothness}, as the contour of the unit sublevel set tends to be very flat when one of the prices is close to the zero.

In this case, we first establish a ``local'' smoothness in a region where all the prices are greater than a strictly positive constant.
Even though, the initial point can lie in a nonsmooth region, where it is difficult to guarantee non-asymptotic convergence.
To handle this, we leverage the ``discrete-time threshold'' in~\cref{lem:discrete-time-relative-tatonnement-properties} and show that one can employ a small but non-diminishing stepsize to ``push'' the prices in the nonsmooth region into the local smooth region; after that, we can ensure the iterates are controlled within the smooth region and then find an approximate CE in a polynomial time.

% Without loss of generality, we set $\min_{j \in [m]} d_{ij} = 1$ and let $D = \max_{j \in [m]} d_{ij}$. 
Denote 
\begin{equation*}
    \sigma = \frac{\rho}{\rho - 1} \in (1, 2) \hspace{6pt} \textnormal{ when } \rho \in (2, \infty). 
\end{equation*}

We first show the follow lemma for the sake of establishing the local smoothness.
\begin{lemma}
  \label{lem:smooth-region}
    Let $d_i(\bfx_i) = (\sum_j d_{ij} |x_{ij}|^\rho)^{{1}/{\rho}}$ with $\rho > 2$. 
    % Then, $C_i = \{ x_i: d_i(x_i) \leq 1 \}$ is a strongly convex set. 
    For any $\bfp_1, \bfp_2 \in \Delta_\bfB$ and $\bfp_1, \bfp_2 \geq r > 0$, 
    it holds that 
    \begin{equation*}
        \norm*{\nabla d_i^\circ(\bfp_1) - \nabla d_i^\circ(\bfp_2)} \leq L(r) \norm*{\bfp_1 - \bfp_2}, 
    \end{equation*}
    where $L(r) := \frac{\sigma - 1}{\norm*{\bfB}_1^{\sigma - 1}}(\max_j d_{ij})^{\frac{(1-\sigma)^2}{\sigma}} (\min_j d_{ij})^{1-\sigma} \frac{1}{r} + \frac{\sigma - 1}{\norm*{\bfB}_1} (\max_j d_{ij})^{\frac{(1-\sigma)(1-2\sigma)}{\sigma}} (\min_j d_{ij})^{2(1-\sigma)}$. 
    \label{lem:smooth-CES-rho-2-infty}
\end{lemma}
\begin{proof}
    % Note that $(\sum_j d_{ij} |x_{ij}|^\rho)^{\frac{2}{\rho}} (\rho > 1)$ is $2(\rho - 1) m^{2(2/\rho - 1)}$-strongly convex. 
    For any $\bfp \geq 0$, the $\max$ in $d_i^\circ(\bfp)$ can be achieved at some $\bfx_i \in \mathbb{R}^m_+$, 
    thereby 
    \begin{equation}
        d_i^\circ(\bfp) 
        = \underset{\bfx_i: (\sum_j d_{ij} |x_{ij}|^\rho)^{{1}/{\rho}} \leq 1}{\max} \inp{\bfp}{\bfx_i} 
        = \underset{\bfx_i \geq 0: \sum_j d_{ij} x_{ij}^\rho \leq 1}{\max} \inp{\bfp}{\bfx} \geq 0. 
    \end{equation}
    By the first-order optimality, we can obtain a closed-form expression of $d_i^\circ(\bfp) = (\sum_{j=1}^m d_{ij}^{1-\sigma} p_j^\sigma)^{1/\sigma}$. 
    % Let $p, q \in \Delta_B$ and $p_j, q_j \geq r$ for $j \in [m]$. 
    Denote $\bfd_i = (d_{i1}, \ldots, d_{im})$. 
    We can then derive that (the power is taken element-wise) 
    \begin{align}
        & \norm*{\nabla d^\circ_i(\bfp_1) - \nabla d^\circ_i(\bfp_2)} \nonumber \\ 
        =& \norm*{d_i^\circ(\bfp_1)^{1 - \sigma} \bfd_i^{1-\sigma} \odot \bfp_1^{\sigma - 1} - d_i^\circ(\bfp_2)^{1 - \sigma} \bfd_i^{1-\sigma} \odot \bfp_2^{\sigma - 1}} \nonumber \\ 
        \leq& \norm*{d_i^\circ(\bfp_1)^{1 - \sigma} \bfd_i^{1-\sigma} \odot \bfp_1^{\sigma - 1} - d_i^\circ(\bfp_1)^{1 - \sigma} \bfd_i^{1-\sigma} \odot \bfp_2^{\sigma - 1}} \nonumber \\ 
        & \hspace{60pt} + \norm*{d_i^\circ(\bfp_1)^{1 - \sigma} \bfd_i^{1-\sigma} \odot \bfp_2^{\sigma - 1} - d_i^\circ(\bfp_2)^{1 - \sigma} \bfd_i^{1-\sigma} \odot \bfp_2^{\sigma - 1}} \nonumber \\ 
        \leq& d_i^\circ(\bfp_1)^{1 - \sigma} (\min_j d_{ij})^{1-\sigma} \norm*{\bfp_1^{\sigma - 1} - \bfp_2^{\sigma - 1}} + \norm*{\bfd_i^{1-\sigma} \odot \bfp_2^{\sigma - 1}} \left| d_i^\circ(\bfp_1)^{1 - \sigma} - d_i^\circ(\bfp_2)^{1 - \sigma} \right|. 
        \label{eq:total-upper-bound}
    \end{align}
    For any $\bfp \in \Delta_B$, we have 
    \begin{equation}
        d_i^\circ(\bfp) = \left( \sum_j d_{ij}^{1-\sigma} p_j^\sigma \right)^{\frac{1}{\sigma}} \geq \left(\max_j d_{ij}\right)^{\frac{1-\sigma}{\sigma}} \norm*{\bfB}_1, 
        \label{eq:d-circ-lower-bound}
    \end{equation}
    since the lower bound is achieved when $p_j = \norm{\bfB}_1$ if $j$ maximizes $d_{ij}$ and $= 0$ otherwise. 
    By a similar argument, we have 
    % \begin{equation*}
    %     d_i^\circ(p) \leq (\min_j d_{ij})^{\frac{1-\sigma}{\sigma}} \norm*{B}_1  
    % \end{equation*}
    % and 
    \begin{equation}
        \norm*{\bfd_i^{1-\sigma} \odot \bfp^{\sigma - 1}} \leq (\min_j d_{ij})^{1-\sigma} \norm{\bfB}_1^{\sigma - 1}. 
        \label{eq:norm-d-q-lower-bound}
    \end{equation}
    
    By the Mean Value Theorem, there exists a $\bfq \in \Delta_B$ on the line segment between $\bfp_1$ and $\bfp_2$ such that 
    \begin{equation*}
        \left| d_i^\circ(\bfp_1)^{1 - \sigma} - d_i^\circ(\bfp_2)^{1 - \sigma} \right| \leq \norm*{(1-\sigma)\left(\sum_j d_{ij}^{1-\sigma} q_j^\sigma\right)^{\frac{1-2\sigma}{\sigma}} \bfd_i^{1-\sigma} \bfq^{\sigma - 1}} \norm*{\bfp_1 - \bfp_2}. 
    \end{equation*}
    The right hand side can be further bounded by 
    \begin{align}
        % & \left| d_i^\circ(p_1)^{1 - \sigma} - d_i^\circ(p_2)^{1 - \sigma} \right| \nonumber \\ 
        % \leq& 
        & \norm*{(1-\sigma)\Big(\sum_j d_{ij}^{1-\sigma} q_j^\sigma\Big)^{\frac{1-2\sigma}{\sigma}} \bfd_i^{1-\sigma} \bfq^{\sigma - 1}} \norm*{\bfp_1 - \bfp_2} \nonumber \\ 
        \leq& (\sigma - 1) (\max_j d_{ij})^{\frac{(1-\sigma)(1-2\sigma)}{\sigma}} \norm{\bfB}_1^{1-2\sigma} (\min_j d_{ij})^{1-\sigma} \norm{\bfB}_1^{\sigma - 1} \norm*{\bfp_1 - \bfp_2} \tag{by~\cref{eq:d-circ-lower-bound,eq:norm-d-q-lower-bound}} \\ 
        =& (\sigma - 1) (\max_j d_{ij})^{\frac{(1-\sigma)(1-2\sigma)}{\sigma}} (\min_j d_{ij})^{1-\sigma} \norm{\bfB}_1^{-\sigma} \norm*{\bfp_1 - \bfp_2}. 
        \label{eq:difference-d-circ-p1msigma-upper-bound}
    \end{align}
    Combining~\cref{eq:difference-d-circ-p1msigma-upper-bound,eq:norm-d-q-lower-bound}, we have 
    \begin{equation}
        \norm*{\bfd_i^{1-\sigma} \odot \bfp_2^{\sigma - 1}} \left| d_i^\circ(\bfp_1)^{1 - \sigma} - d_i^\circ(\bfp_2)^{1 - \sigma} \right| \leq (\sigma - 1) (\max_j d_{ij})^{\frac{(1-\sigma)(1-2\sigma)}{\sigma}} (\min_j d_{ij})^{2(1-\sigma)} \norm{\bfB}_1^{-1} \norm*{\bfp_1 - \bfp_2}. 
        \label{eq:term-2-upper-bound}
    \end{equation}

    By~\cref{eq:d-circ-lower-bound}, we have 
    \begin{equation}
        d_i^\circ(\bfp)^{1 - \sigma} \leq (\max_j d_{ij})^{\frac{(1-\sigma)^2}{\sigma}} \norm*{\bfB}_1^{1 - \sigma}. 
        \label{eq:d-circ-p1msigma-upper-bound}
    \end{equation}
    % Therefore, 
    % \begin{equation*}
    %     \norm*{\nabla d^\circ_i(p) - \nabla d^\circ_i(q)} \leq c_1 \norm*{p^{\sigma - 1} - q^{\sigma - 1}} + c_2 \norm*{p - q}, 
    % \end{equation*}
    % where $$c_1' = (\max_j d_{ij})^{\frac{(1-\sigma)^2}{\sigma}} \norm*{B}_1^{1 - \sigma} (\min_j d_{ij})^{1-\sigma}$$ and $$c_2 = (\min_j d_{ij})^{1-\sigma} \norm{B}_1^{\sigma - 1} (\sigma - 1) (\max_j d_{ij})^{\frac{(1-\sigma)(1-2\sigma)}{\sigma}} (\min_j d_{ij})^{1-\sigma} \norm{B}_1^{-\sigma} = (\sigma - 1) (\max_j d_{ij})^{\frac{(1-\sigma)(1-2\sigma)}{\sigma}} (\min_j d_{ij})^{2(1-\sigma)} \norm{B}_1^{-1}.$$
    Moreover, by the Mean Value Theorem, we have 
    \begin{align}
        \norm*{\bfp_1^{\sigma - 1} - \bfp_2^{\sigma - 1}} = \sqrt{\sum_j \left| (\bfp_1)_j^{\sigma - 1} - (\bfp_2)_j^{\sigma - 1}\right|^2} &= \sqrt{\sum_j \left| (\sigma - 1) q_j^{\sigma - 2} \right|^2 \left| (\bfp_1)_j - (\bfp_2)_j \right|^2} \nonumber \\ 
        &\leq (\sigma - 1) r^{\sigma - 2} \norm*{\bfp_1 - \bfp_2}, 
        \label{eq:difference-p-psigmam1-upper-bound}
    \end{align}
    where $q_j \in [(\bfp_1)_j, (\bfp_2)_j]$ for all $j \in [m]$. 
    Combining~\cref{eq:d-circ-p1msigma-upper-bound,eq:difference-p-psigmam1-upper-bound}, we have 
    \begin{equation}
        d_i^\circ(\bfp_1)^{1 - \sigma} (\min_j d_{ij})^{1-\sigma} \norm*{\bfp_1^{\sigma - 1} - \bfp_2^{\sigma - 1}} 
        \leq (\sigma - 1) (\max_j d_{ij})^{\frac{(1-\sigma)^2}{\sigma}} (\min_j d_{ij})^{1-\sigma} \norm*{\bfB}_1^{1 - \sigma}  r^{\sigma - 2} \norm*{\bfp_1 - \bfp_2}. 
        \label{eq:term-1-upper-bound}
    \end{equation}

    Since $\sigma \in (1, 2)$ and $r \leq 1$, $r^{\sigma - 2}$ can be further upper bounded by $\frac{1}{r}$.
    Then, the lemma follows by combining~\cref{eq:total-upper-bound,eq:term-1-upper-bound,eq:term-2-upper-bound}. 
\end{proof}

We then define a proper local region as follows: 
\begin{equation*}
    P_{s} := \{ \bfp \in \Delta_B \mid p_j \geq \frac{\ell_0}{2} \textnormal{ for all } j \in [m] \}. 
\end{equation*}

By~\cref{lem:smooth-CES-rho-2-infty}, the gauge dual is $L(\frac{\ell_0}{2})$-smooth within $P_{s}$, where $L(r)$ is defined in~\cref{lem:smooth-CES-rho-2-infty}.

Leveraging the ``discrete-time'' threshold established in~\cref{lem:discrete-time-relative-tatonnement-properties}, we can make sure that the discrete-time t\^atonnement dynamics enter the ``local smooth region'' in a polynomial time and stay within it. 

\begin{lemma}
    \label{lem:non-smooth-to-smooth}
    Consider a chores Fisher market with CES disutility functions with $\rho \in (2, \infty)$.
    Let $\{ \bfp^k \}_{k = 0, 1, 2, \ldots}$ be a sequence of iterates generated by~(\ref{discrete-time-relative-chores-tatonnement}) with any initial point $\bfp^0 \in \Delta_B$ 
    and $\eta^k = \frac{\norm{\bfB}_1}{18 m^2} (\min_{i \in [n]} \nu_i )^2$ for all $k$, the iterates enter $P_{s}$ in $\lceil \frac{18m^3}{\min_{i \in [n]} \nu_i} \rceil$ iterations, where $\nu_i$ is defined in~\cref{eq:nui defn CES} for CES functions.
    Furthermore, 
    at any price $\bfp \in P_{s}$, the price iterates initialized from $\bfp$ stay within $P_{s}$ once $\eta^k \leq \frac{\norm{\bfB}_1}{18 m^2} (\min_{i \in [n]} \nu_i)^2$.
\end{lemma}
\begin{proof}
    Recall that $\ell_0 = \frac{\norm{\bfB}_1}{3m} \min_{i \in [n]} \nu_i$, where $\nu_i$ is defined in~\cref{eq:nui defn CES} for CES functions.
    By~\cref{lem:discrete-time-relative-tatonnement-properties}, for any $k \geq 0$ and any $j \in [m]$ such that $p^k_j \leq \ell_0$, we have we have $p^{k+1}_{j} > p^k_j + \frac{\ell_0^2}{2 \norm{\bfB}_1} \cdot \frac{1}{6m}$. 
    For each chore $j$, this ``jump'' can happen at most $\left\lceil \frac{\ell_0}{2} \Big/ \frac{\ell_0^2}{12 m \norm{\bfB}_1} \right\rceil = \left\lceil \frac{6 m \norm{\bfB}_1}{\ell_0}  \right\rceil$ times before the chore $j$ become greater than $\frac{\ell_0}{2}$.
    This upper bound can be simplified as $\frac{6 m \norm{\bfB}_1}{\ell_0} = \frac{6m \norm{\bfB}_1}{(\norm{\bfB}_1/ 3m) \min_{i \in [n]} \nu_i} = \frac{18m^2}{\min_{i \in [n]} \nu_i}$.
    Therefore, the total number of iterations before the iterates enter $P_s$ is at most $\lceil \frac{18m^3}{\min_{i \in [n]} \nu_i} \rceil$.
    By~\cref{lem:discrete-time-relative-tatonnement-properties}, the price iterates will stay above $\frac{\ell_0}{2}$ after that.
    Therefore, 
    the price iterates enter the local region $P_{s}$ in polynomial time.
\end{proof}

The above lemma prove the polynomial-time termination of the first phase. 
After that, we can leverage the local smoothness to prove the $\tilde{\mathcal{O}}({1}/{\varepsilon^2})$ convergence to an approximate CE, as shown in the next part.

\subsubsection{Convergence to an approximate CE at a $\tilde{\mathcal{O}}({1}/{\varepsilon^2})$ rate}
\label{subsubsec:smooth-convergence}

\begin{lemma}
    If the $i$-th agent has $L_i$-smooth demand, 
    and the gauge dual of her disutility is lower bounded by $r_i > 0$ once $\bfp \in \Delta_B$, 
    then $f$ defined in~\cref{eq:potential-CCH} is $L$-smooth where $$L = \sum_{i=1}^n B_i \left( \frac{L_i}{r_i} + \frac{R_i(L_i \norm*{\bfB}_1 + R_i)}{r_i^2} \right).$$
    \label{lem:smooth-f-HB}
\end{lemma}

\begin{proof} 
    Since $\nabla d_i^\circ(\bfp) = \textnormal{argmax}_{\bfx: d(\bfx) \leq 1} \inp{\bfp}{\bfx}$, we have $\norm*{\nabla d_i^\circ(\bfp)} \leq R_i$ for all $\bfp$. 
    If $\lVert \nabla d_i^\circ(\bfp_1) - \nabla d_i^\circ(\bfp_2) \rVert \leq L_i \norm*{ \bfp_1 - \bfp_2 }$, 
    then for all $\bfp_1, \bfp_2, \in \Delta_B$, 
    we have 
    \begin{align}
        \left\lvert d_i^\circ(\bfp_1) - d_i^\circ(\bfp_2) \right\rvert &\leq \left| \inp{\bfp_1}{\nabla d_i^\circ(\bfp_1)} - \inp{\bfp_2}{\nabla d_i^\circ(\bfp_2)} \right| \nonumber \\ 
        &\leq \left| \inp{\bfp_1}{\nabla d_i^\circ(\bfp_1)} - \inp{\bfp_1}{\nabla d_i^\circ(\bfp_2)} \right| + \left| \inp{\bfp_1}{\nabla d_i^\circ(\bfp_2)} - \inp{\bfp_2}{\nabla d_i^\circ(\bfp_2)} \right| \nonumber \\ 
        &\leq \norm{\bfp_1} \norm{\nabla d_i^\circ(\bfp_1) - \nabla d_i^\circ(\bfp_2)} + \norm{\nabla d_i^\circ(\bfp)} \norm{\bfp_1 - \bfp_2} \nonumber \\ 
        &\leq (L_i \norm*{\bfB}_1 + R_i) \norm*{ \bfp_1 - \bfp_2 }. 
    \end{align}
    On the other hand, $d_i^\circ(\bfp) \geq \inp{\bfp}{\bfx_i'}$ for any $\bfx_i'$ such that $d_i(\bfx_i') \leq 1$. 
    For any appropriate $d_i$, we can see there is $r_i > 0$ such that $d_i^\circ(\bfp) \geq r_i$ for all $\bfp \in \Delta_B$. 
    Next, we can upper bound 
    \begin{align*}
        \norm*{\nabla f(\bfp_1) - \nabla f(\bfp_2)} &= \norm*{\sum_{i=1}^n B_i \left( \frac{\nabla d_i^\circ(\bfp_1)}{d_i^\circ(\bfp_1)} - \frac{\nabla d_i^\circ(\bfp_2)}{d_i^\circ(\bfp_2)} \right)} \nonumber \\ 
        &\leq \sum_{i=1}^n B_i \norm*{\frac{\nabla d_i^\circ(\bfp_1)}{d_i^\circ(\bfp_1)} - \frac{\nabla d_i^\circ(\bfp_2)}{d_i^\circ(\bfp_2)}} \nonumber \\ 
        &\leq \sum_{i=1}^n B_i \Big( \norm*{\frac{\nabla d_i^\circ(\bfp_1)}{d_i^\circ(\bfp_1)} - \frac{\nabla d_i^\circ(\bfp_2)}{d_i^\circ(\bfp_1)}} + \norm*{\frac{\nabla d_i^\circ(\bfp_2)}{d_i^\circ(\bfp_1)} - \frac{\nabla d_i^\circ(\bfp_2)}{d_i^\circ(\bfp_2)}} \Big) \nonumber \\ 
        &\leq \sum_{i=1}^n B_i \left( \frac{L_i}{r_i} \norm{\bfp_1 - \bfp_2} + \norm{\nabla d_i^\circ(\bfp_2)} \frac{\left| d_i^\circ(\bfp_1) - d_i^\circ(\bfp_2) \right|}{\left| d_i^\circ(\bfp_1) \right| \left| d_i^\circ(\bfp_2) \right|} \right) \nonumber \\ 
        &\leq \sum_{i=1}^n B_i \left( \frac{L_i}{r_i} + \frac{R_i(L_i \norm*{\bfB}_1 + R_i)}{r_i^2} \right) \norm{\bfp_1 - \bfp_2}. 
    \end{align*}
\end{proof}

Because $\norm*{\mathbf{a} - \frac{1}{n}\sum_{i=1}^n a_i} \leq \norm*{\mathbf{a}}$ for any vector $\mathbf{a} \in \mathbb{R}^n$, 
$L$-smoothness of $f$ implies $L$-smoothness of $f\vert_{H_B}$, i.e., 
\begin{equation}
    \norm*{\nabla f\vert_{H_B}(\bfp_1) - \nabla f\vert_{H_B}(\bfp_2)} \leq \norm*{\nabla f(\bfp_1) - \nabla f(\bfp_2)} \leq L \norm*{\bfp_1 - \bfp_2} \hspace{30pt} \forall\, \bfp_1, \bfp_2 \in \Delta_B. 
\end{equation}

Following standard derivations, we can show that the gradient descent method can find an $\varepsilon$-stationary point of a smooth function $f$ in $\tilde{\mathcal{O}}(\frac{1}{\varepsilon^2})$ iterations. 
Here, $f$ is considered to be a general smooth function, not the one defined in the previous sections.
\begin{theorem}
    If a function $f$ is $L$-smooth in a closed convex region $S \subset \textnormal{dom} (f)$, 
    and gradient descent with a constant stepsize $\eta \leq \frac{1}{L}$ generates a sequence of iterates within $S$, 
    then gradient descent can find an $\varepsilon$-stationary point of $f$ in $({f(\bfp^0) - \underline{f}})\big/\left( {(\eta - \frac{\eta^2 L}{2}) \varepsilon^2} \right)$ iterations. 
    \label{thm:1-over-squaredeps-from-smoothness}
\end{theorem}
\begin{proof}
    % Because $f$ has $L$-Lipschitz gradients on $\Delta_B$, i.e., 
    % \begin{equation*}
    %     \norm*{\nabla f(p_1) - \nabla f(p_2)} \leq L\norm*{p_1 - p_2} \hspace{30pt} \forall\, p_1, p_2 \in \Delta_B, 
    % \end{equation*}
    Note that $f$ is continuously differentiable, then for any two consecutive iterates $\bfp_1, \bfp_2 \in S$, we have 
    \begin{equation*}
        f(\bfp_2) - f(\bfp_1) = \int_0^1 \inp*{\nabla f(\bfp_1 + t(\bfp_2 - \bfp_1))}{\bfp_2 - \bfp_1} dt. 
    \end{equation*}
    This yields the following inequality: 
    \begin{align}
        f(\bfp_2) - f(\bfp_1) - \inp*{\nabla f(\bfp_1)}{\bfp_2 - \bfp_1} &= \int_0^1 \inp*{\nabla f(\bfp_1 + t(\bfp_2 - \bfp_1)) - \nabla f(\bfp_1)}{\bfp_2 - \bfp_1} dt \nonumber \\ 
        &\leq \int_0^1 \norm*{\nabla f(\bfp_1 + t(\bfp_2 - \bfp_1)) - \nabla f(\bfp_1)} \norm*{\bfp_2 - \bfp_1} dt 
        \nonumber \\ 
        &\leq \int_0^1 t L \norm*{\bfp_2 - \bfp_1}^2 dt \nonumber \\ 
        &= \frac{L}{2} \norm*{\bfp_2 - \bfp_1}^2, 
    \end{align}
    where the first inequality follows from the Cauchy-Schwarz inequality, and the second inequality follows from $L$-Lipschitz gradients of $f$ on $\Delta_B$. 

    Let $\bfp^k$ denote the $k$-th iterate, we have 
    \begin{equation}
        f(\bfp^{k + 1}) - f(\bfp^k) - \inp*{\nabla f(\bfp^k)}{\bfp^{k + 1} - \bfp^k} \leq \frac{L}{2} \norm*{\bfp^{k + 1} - \bfp^k}^2. 
        \label{eq:L-smooth-inequality}
    \end{equation}
    Plugging $\bfp^{k + 1} - \bfp^k = - \eta \nabla f(\bfp^k)$ into~\cref{eq:L-smooth-inequality} and rearranging terms, we have 
    \begin{equation}
        (\eta - \frac{\eta^2 L}{2}) \norm*{\nabla f(\bfp^k)}^2 \leq f(\bfp^k) - f(\bfp^{k + 1}). 
        \label{eq:L-smooth-inequality-gradient-descent}
    \end{equation}
    As $0 < \eta \leq \frac{1}{L}$, the left-hand side of~\cref{eq:L-smooth-inequality-gradient-descent} is positive. 
    By telescoping~\cref{eq:L-smooth-inequality-gradient-descent} over $k = 0, \ldots, T - 1$, we have 
    \begin{equation}
        (\eta - \frac{\eta^2 L}{2}) (T + 1) \min_{0 \leq k \leq T} \norm*{\nabla f(\bfp^k)}^2 \leq (\eta - \frac{\eta^2 L}{2}) \sum_{k = 0}^T \norm*{\nabla f(\bfp^k)}^2 \leq f(\bfp^0) - f(\bfp^{T + 1}) \leq f(\bfp^0) - \underline{f}. 
    \end{equation}
    This yields that 
    \begin{equation*}
        \left( \min_{0 \leq k \leq T} \norm*{\nabla f(\bfp^k)} \right)^2 \leq \frac{f(\bfp^0) - \underline{f}}{(\eta - \frac{\eta^2 L}{2}) (T + 1)} \leq \frac{f(\bfp^0) - \underline{f}}{(\eta - \frac{\eta^2 L}{2}) T}. 
    \end{equation*}
    This means, to achieve $\min_{0 \leq k \leq T} \norm*{\nabla f(\bfp^k)} \leq \varepsilon$, it suffices to set $T = \frac{f(\bfp^0) - \underline{f}}{(\eta - \frac{\eta^2 L}{2}) \varepsilon^2}$. 
    which can be telescoped and leads to the desired results. 
\end{proof}

Combining the above results, we are ready to present the following iteration complexity of the dynamics~\eqref{discrete-time-relative-chores-tatonnement}. 

\begin{theorem}
    If every agent has a CES disutility function with $\rho \in (1,\infty)$, 
    then~(\ref{discrete-time-relative-chores-tatonnement}) can find an approximate CE in $\tilde{\mathcal{O}}\left( \frac{1}{\varepsilon^2} \right)$ iterations. 
\end{theorem}
\begin{proof}
    Recall that we define 
    \begin{equation}
        \delta_i = \sup\{ \delta > 0 \mid d(\delta \cdot \mathbf{1}_m) \leq 1 \} \hspace{30pt} \forall\, i \in [n]. 
    \end{equation}
    We know that $\delta_i > 0$ because of the continuity and $d_i(\mathbf{0}_m) = 0$. 
    Then, 
    for any $\bfp \in \Delta_B$, $d_i^\circ(\bfp) = \sup\{ \inp*{\bfp}{\bfx} \mid d_i(\bfx) \leq 1 \}$ is lower bounded by $\delta_i \norm*{\bfB}_1$. 
    This leads to a lower bound of $f\vert_{\Hb}(\bfp)$ over its domain: 
    \begin{equation}
        f\vert_{\Hb}(\bfp) = - \sum_{j = 1}^m p_j + \sum_{i = 1}^n B_i \log\left( d^\circ(\bfp) \right) \geq - \norm*{\bfB}_1 + \sum_{i = 1}^n B_i \log{\left( \delta_i \norm*{\bfB}_1 \right)}. 
        \label{eq:fHB-lower-bound}
    \end{equation}
    
    Recall that $R_i = \sup\{ \norm*{\bfx_i}_\infty \mid d_i(\bfx_i) \leq 1 \} < \infty$. 
    This yields an upper bound of $f\vert_{\Hb}(\bfp)$ over its domain: 
    \begin{equation}
        f\vert_{\Hb}(\bfp) = - \sum_{j = 1}^m p_j + \sum_{i = 1}^n B_i \log\left( d^\circ(\bfp) \right) \leq - \norm*{\bfB}_1 + \sum_{i = 1}^n B_i \log{\left( R_i \norm*{\bfB}_1 \right)}. 
        \label{eq:fHB-upper-bound}
    \end{equation}
    
    Thus, the difference between $f\vert_{\Hb}(\bfp^0)$ and $(f\vert_{\Hb})^*$ is at most $\sum_{i = 1}^n B_i \log{\left( R_i / \delta_i \right)}$. 
    For CES disutility functions with $\rho \in (1, \infty)$, we have $R_i = (\min_{j \in [m]} d_{ij})^{-1/\rho}$ by~\cref{lem:boundedness-x-sublevel}, and $\delta_i \geq \left(\frac{1}{m \max_{j \in [m]} d_{ij}}\right)^{1/\rho}$ because $d_i\left( \left(\frac{1}{m \max_{j \in [m]} d_{ij}}\right)^{1/\rho} \mathbf{1}_m \right) = \sum_{j=1}^m d_{ij} \left( \left(\frac{1}{m \max_{j \in [m]} d_{ij}}\right)^{1/\rho} \right)^\rho \leq 1$.
    Therefore, 
    \begin{equation*}
        f\vert_{\Hb}(\bfp^0) - (f\vert_{\Hb})^* \leq \sum_{i=1}^n B_i \log(\frac{R_i}{\delta_i}) \leq \frac{1}{\rho} \sum_{i=1}^n B_i \log\left( \frac{m \max_{j \in [m]} d_{ij}}{\min_{j \in [m]} d_{ij}} \right).
    \end{equation*}

    Let $\eta \leq \frac{\ell_0^2}{2\norm*{\bfB}_1}$. 
    For $\rho \in (1, 2]$, by~\cref{crl:smooth-CES-rho-1-2} we obtain the $L$-smoothness of the gauge dual of the disutility functions.
    For $\rho \in (2, \infty)$, 
    by~\cref{lem:non-smooth-to-smooth,lem:smooth-region} we can obtain the $L$-smoothness of the gauge dual of the disutility functions after a constant (independent of $\frac{1}{\varepsilon}$) number of warm-start iterations bounded by $\textnormal{poly}(n, m, \max_{i \in [n]} B_i, \frac{1}{\ell_0})$.
    By~\cref{lem:smooth-f-HB}, we compute the smoothness modulus of the restriction $f\vert_{\Hb}$. 
    Note that the gauge duals of the CES disutility functions are lower bounded by $\delta_i \norm*{\bfB}_1$.
    Letting $\eta = \min\{ \frac{\ell_0^2}{2\norm*{\bfB}_1}, \frac{1}{2L} \}$. 
    by~\cref{thm:1-over-squaredeps-from-smoothness},~\cref{eq:fHB-lower-bound,eq:fHB-upper-bound}, we obtain an $\varepsilon$-stationary point after $\tilde{\mathcal{O}}(\frac{1}{\varepsilon^2})$ iterations. 
    
    Furthermore, note that an $\varepsilon$-stationary point of $f\vert_{\Hb}(\bfp)$ means that there is a (unique) relative excess demand $\tilde{\bfz} \in \tilde{Z}(\bfp)$ such that $\norm*{\tilde{\bfz}} \leq \varepsilon$. 
    Consequently, there exists a $\bfz \in Z(\bfp)$ such that $-\varepsilon \leq z_j - \frac{1}{m}\sum_{j=1}^m z_j \leq \varepsilon$ for all $j \in [m]$ which leads to $-2\varepsilon \leq z_j \leq 2\varepsilon$ for all $j \in [m]$ if $\bfp \in \Hb$. 
    This is because: 
    for each $\bfz \in Z(\bfp)$, we have $\{ b_{ij} \}_{i \in [n], j \in [m]}$ such that $z_j = \sum_{i = 1}^n b_{ij} / p_j - 1, \; \forall\, j \in [m]$ and $\sum_{j=1}^m b_{ij} = B_i, \; \forall\, i \in [n]$; this then implies that $\sum_{i=1}^n B_i \left( \frac{1}{m}\sum_{j=1}^m z_j - \varepsilon + 1 \right) \leq \sum_{i=1}^n B_i = \sum_{j=1}^m \sum_{i=1}^n b_{ij} = \sum_{j=1}^m p_j \left( z_j + 1 \right) \leq \sum_{j=1}^m p_j \left( \frac{1}{m}\sum_{j=1}^m z_j + \varepsilon + 1 \right) = \sum_{i=1}^n B_i \left( \frac{1}{m}\sum_{j=1}^m z_j + \varepsilon + 1 \right)$, yielding $-\varepsilon \leq \frac{1}{m}\sum_{j=1}^m z_j \leq \varepsilon$.
\end{proof}

% \begin{remark}
%     \tianlong{For other related convergence results}
% \end{remark}

% Then, we adapt results from~\cite{bianchi2022convergence} to show the convergence of chores t\^atonnement with constant stepsizes. Again, as we do not consider any stochastic noise, many assumptions for their results hold trivially. 

% \begin{theorem}[Adapted from Theorem 3 in~\cite{bianchi2022convergence}] 
%     Constant-stepsize Convergence 
% \end{theorem}

% \begin{proof}
%     First, we need to verify their Assumptions 1-5 hold true. 
%     Assumption 1-2 in~\cite{bianchi2022convergence} is true because of locally Lipschitz continuity and the absence of stochastic noise. 
%     Assumption 3 holds for $f\vert_{H_B}$ by Proposition 4 in~\cite{bianchi2022convergence}. 
% \end{proof}

\section{Stability of Competitive Equilibria in Chores Markets}
\label{sec:stability-main}

In this section we investigate the stability properties of CE in chores markets with \emph{linear disutilities}. %, by connecting the stability of the individual CE to its minimum pain-per-buck graph structures. % under linear disutilities.
In the goods case, $\bfp^*$ is known to be unique for linear Fisher markets, which implies global stability.
In contrast, for chores there may be multiple equilibria, and we already showed in \cref{fig:instances-stability-NW} that the landscape of stable and unstable equilibria can be quite interesting.
% It can be shown that none of them is globally stable once we have more than two CE [To prove], and there could exist locally unstable CE. 

% Beyond the definition of CE stability in~\cite{arrow1958stability}, we have some alternative definitions: 
% \begin{definition}
%     Let $p^*$ be a CE in a chores Fisher market $\left( n, m, B, D \right)$. 
%     Denote $Z(p)$ as the set of the excess demand in this market given a price vector $p$. 
%     Then, we say a CE is stable if for some $\epsilon > 0$ and \emph{any} $p \in \mathcal{B}_\epsilon(p^*)$, we have 
%     \begin{equation}
%         \inp*{\tilde{z}}{p - p^*} > 0 \quad \forall\; \tilde{z} \in \tilde{Z}(p), 
%     \end{equation}
%     where the set of relative excess demand $\tilde{Z}(p) = \{ z - \frac{1}{m} \sum_{j=1}^m z_j \,|\, z \in Z(p) \}$. 
% \end{definition} 
% Another interesting way to define or characterize the stability of CE is connecting it to the change of MPB structures. 

% """ 
% [To show: the definitions are equivalent or one implies another]
% """ 

% \subsection{Concrete instances} 

% \ck{The below is what we hope for. Is it true? Need to redefine global stability to be for any non-CE starting point, otherwise the theorem is trivially true.}
% \begin{theorem}
%     If the number of competitive equilibria is exactly one, then that equilibrium is globally stable. If there is more than one equilibrium then there is no globally stable equilibrium. 
% \end{theorem}
% \ck{No time for this?}
We now characterize stable CE in terms of properties of our potential function, the structure of the MPB graph, and the direction of the excess demand correspondence. 

For linear Fisher market with chores, we introduce the concepts of MPB graph. 
To satisfy the optimal bundles condition, every agent only takes those chores with the \emph{minimum pain-per-buck} (MPB), i.e., $\frac{d_{ij}}{p^*_j} = \underset{l \in [m]}{\min} \frac{d_{il}}{p^*_l}$, that is, $x^*_{ij} > 0$ only if $\frac{p^*_j}{d_{ij}} = \underset{l \in [m]}{\max} \frac{p^*_l}{d_{il}}$. We denote the MPB set for agent $i$ given a price vector $\bfp$ as 
\begin{equation*}
    J^*_i(\bfp) = \Big\{ j \,\left\vert\, \frac{p_j}{d_{ij}} = \underset{l \in [m]}{\max} \frac{p_l}{d_{il}} \right. \Big\}. 
\end{equation*}

% \tianlong{MPB graph: may not exactly equivalent statement to characteristics of stability of CE}

\begin{theorem}
    Let $\bfp^*$ be an equilibrium price. 
    The following statements are equivalent: 
    \begin{itemize}
        \item[(i)] $\bfp^*$ is locally stable; 
        % \item[(ii)] There exists a neighborhood $\mathbf{N}(p^*)$ of $p^*$ such that 
        % \begin{equation*}
        %     \inp*{\zeta}{p - p^*} > 0 \quad \forall\; \zeta \in \tilde{z}(p), p \in \mathbf{N}(p^*) \cap H_B \setminus \{ p^* \}; 
        % \end{equation*} 
        \item[(ii)] $\bfp^*$ is a strict local minimum of the function $f\vert_{H_B}$; 
        % \item[(iii)] For all $\nu \in \mathbf{H}_0$, $\left(f\vert_{H_B}\right)'\left( p^*; \nu \right) > 0$; 
        \item[(iii)] 
        % The MPB graph corresponding to $p^*$ is connected, i.e., 
        There exists an equilibrium allocation $\bfx^*$ coupled with $\bfp^*$ satisfying: 
        given any nonempty $J \subset [m]$, there exist some $j \in J$ and $j' \in J^c$ such that 
        % For each $j \in [m]$, there exists at least one $j' \neq j$ such that for some $i$ 
        \begin{equation*}
            \exists\; i \in [n]: j, j' \in J^*_i(\bfp^*) \quad \text{and} \quad x^*_{i j'} > 0. 
        \end{equation*} 
        % or equivalently, there exists some $i \in [n]$ such that $j \in J^*_i(p^*)$ and $\tilde{J}^*_i(p^*, x^*) \cap J^c \neq \emptyset$. 
        \item[(iv)]
        There exists a neighborhood $\mathbf{N}(\bfp^*) \subset \mathbb{R}^m_+$ of $\bfp^*$ such that 
        \begin{equation*}
            \inp*{\bfz}{\bfp - \bfp^*} \geq 0 \quad \forall\; \bfz \in Z(\bfp), \bfp \in \mathbf{N}(\bfp^*) \cap H_B. 
        \end{equation*} 
    \end{itemize}
    \label{thm:characterize-locally-stable-chores-ce}
\end{theorem}

Before presenting the proof, we introduce additional tools in the non-smooth analysis we will use in the subsequent proof. 

The directional derivative can also be generalized. The \emph{Clarke generalized directional derivative} of $f$ at $\bfx$ in the direction $\nu$ is defined as follows: 
\begin{equation*}
    f^\circ(\bfx; \nu) = \limsup_{\bfy \rightarrow \bfx, h \downarrow 0} \frac{f(\bfy+h\nu) - f(\bfy)}{h}. 
\end{equation*} 
A useful fact is~\cite[Proposition 1.5(c)]{clarke2008nonsmooth}: 
\begin{equation}
    f^\circ(\bfp; \nu) = \max\{ \inp*{\zeta}{\nu}: \zeta \in \partial f(\bfp) \}, \quad \forall\; \nu 
    \label{fact:Clarke-directional-derivatives}
\end{equation} 
for a locally Lipschitz function $f$. 
If $f$ is subdifferentially regular, 
the (original) directional derivative coincides with the Clarke generalized directional derivative.

\begin{proof}[Proof of~\cref{thm:characterize-locally-stable-chores-ce}]
    First, we show $\textnormal{(i)} \Leftrightarrow \textnormal{(ii)}$. 
    For $\textnormal{(ii)} \Rightarrow \textnormal{(i)}$, 
    because the equilibrium prices are disconnected, and $p^*$ is a strict local minimum of the function $f\vert_{H_B}$,  
    we know there is a neighborhood of $\bfp^*$ defined as $\tilde{\mathbf{N}}(\bfp^*) := \{ \bfp \in \mathcal{B}_r(\bfp^*) \cap \Delta_B \; \vert \; f\vert_{H_B}(\bfp) \leq f\vert_{H_B}(\bfp^*) + \epsilon \}$ for some $r, \epsilon > 0$ such that $\bfp^*$ is the unique minimizer in $\tilde{\mathbf{N}}(\bfp^*)$. 
    Moreover, there must exist some $\epsilon' > 0$ such that $f\vert_{\Hb}(\bfp^*)+\epsilon'$ is the smallest value attained on $\tilde{\mathbf{N}}(\bfp^*)\cap \partial \mathcal{B}_r(\bfp^*)$, where $\partial \mathcal{B}_r(\bfp^*)$ is the boundary of the ball.
    Since $f\vert_{H_B}$ is locally Lipschitz continuous, there exists another neighborhood $\mathbf{N}(\bfp^*) := \{ \bfp \in \mathcal{B}_r(\bfp^*) \cap \Delta_B \; \vert \; f\vert_{H_B}(\bfp) \leq f\vert_{H_B}(\bfp^*) + \epsilon'/2 \} \subset \tilde{N}(\bfp^*)$. 
    Now if we initialize the continuous-time dynamics anywhere in $\mathbf{N}(\bfp^*)$ then we have by the strict descent property of the continuous trajectory in~\cref{lem:descent} that the trajectory cannot escape from $\tilde{\mathbf{N}}(\bfp^*)$. 
    Because the dynamical system \textnormal{(\ref{chores-tatonnement})} is stable, there has to be a CE which any trajectory converges to. 
    It follows that any trajectory starting from any initial point in $\mathbf{N}(\bfp^*)$ converges to $\bfp^*$. By definition, $\bfp^*$ is locally stable. 
    % \tianlong{Updated}\ck{updated more}
    For $\textnormal{(i)} \Rightarrow \textnormal{(ii)}$, we prove it by contraposition. 
    % Because CE are disconnected, $p^*$ is the only stationary point in a small enough neighborhood around itself.
    Suppose $\bfp^*$ is a strict local maximum or a saddle point (the only other options besides being a local minimum). 
    In either case, we can find a point $\bfp'$ in any neighborhood of $\bfp^*$ such that $f\vert_{H_B}(\bfp^*) > f\vert_{H_B}(\bfp^0)$. 
    Again, by the strict descent property in~\cref{lem:descent}, any continuous trajectory starting from $\bfp'$ will not converge to $\bfp^*$. 
    Thus $\bfp^*$ is not locally stable. 

    $\textnormal{(ii)}$ $\implies$ $\textnormal{(iv)}$. 
    Assume that $\bfp^*$ is a strict local minimum of the function $f\vert_{H_B}$. 
    Suppose that for any neighborhood $\mathbf{N}(\bfp^*) \subset \RR^m_+$ of $\bfp^*$, there exists a $\bfp' \in \mathbf{N}(\bfp^*) \cap H_B$ and a $\bfz \in Z(\bfp')$ such that $\inp{\bfz}{\bfp' - \bfp^*} < 0$. 
    Let $\tilde{\bfz} = \bfz - \frac{1}{m} \mathbf{1}_m^\top \bfz$. 
    By the subdifferential regularity of $f\vert_{H_B}$, we have 
    \begin{equation*}
        f\vert_{H_B}(\bfp) \geq f\vert_{H_B}(\bfp') + \inp{\tilde{\bfz}}{\bfp - \bfp'} + o\left( \norm{\bfp - \bfp'} \right) \quad \text{as } \bfp \rightarrow \bfp'. 
    \end{equation*}
    Since $\inp{\mathbf{1}_m}{\bfp - \bfp'} = 0$ for any $\bfp, \bfp' \in H_B$, 
    we further have 
    \begin{equation*}
        f\vert_{H_B}(\bfp) \geq f\vert_{H_B}(\bfp') + \inp{\bfz}{\bfp - \bfp'} + o\left( \norm{\bfp - \bfp'} \right) \quad \text{as } \bfp \rightarrow \bfp'. 
    \end{equation*}
    Let $\bfp' \rightarrow \bfp^*$ and pick the $\bfz \in Z(\bfp')$ such that $\inp{\bfz}{\bfp' - \bfp^*} < 0$, then we obtain 
    \begin{align*}
        f\vert_{H_B}(\bfp^*) &\geq f\vert_{H_B}(\bfp') + \inp{\bfz}{\bfp^* - \bfp'} + o\left( \norm{\bfp^* - \bfp'} \right) \\ 
        &> f\vert_{H_B}(\bfp') + o\left( \norm{\bfp^* - \bfp'} \right) \\ 
        &\geq f\vert_{H_B}(\bfp') + o\left( \lvert f\vert_{H_B}(\bfp^*) - f\vert_{H_B}(\bfp') \rvert \right), 
    \end{align*}
    where the last inequality is because $f\vert_{H_B}$ is locally Lipschitz. This leads to $f\vert_{H_B}(\bfp^*) > f\vert_{H_B}(\bfp')$ and thus a contradiction. 
    
    $\textnormal{(iv)}$ $\implies$ $\textnormal{(iii)}$. 
    We prove this by contraposition. 
    Suppose $\textnormal{(iii)}$ does not hold
    and let $\bfx^*$ be any equilibrium allocation associated with $\bfp^*$. 
    % Let $b^*\in \RR^{n\times m}_+$ be such that $b^*_{ij} = p^*_j x^*_{ij}$ is the amount that agent $i$ earns from chore $j$. 
    %We let $b^*\in \RR^{n\times m}_+$ be the earning at $(p^*, x^*)$. 
    Suppose that, fixing $\bfx^*$, there is a nonempty subset $J \subset [m]$ such that for every $j \in J$ and $j' \in J^c$ one of the following cases happens for each $i \in [n]$: 
        (1) $\frac{p^*_j}{d_{ij}} \neq \frac{p^*_{j'}}{d_{i j'}}$; 
        (2) $\frac{p^*_j}{d_{ij}} = \frac{p^*_{j'}}{d_{i j'}} < \frac{p^*_k}{d_{ik}}$ for some $k \neq j,j'$; 
        (3) $\frac{p^*_j}{d_{ij}} = \frac{p^*_{j'}}{d_{i j'}} \geq \frac{p^*_k}{d_{ik}},\;\forall\,k \neq j, j'$ but $x^*_{i j'} = 0$. 
    % Since $p^*$ is an equilibrium price, in any of the above three cases, for \emph{every} $j \in J$ there is a set $I_j^*(p^*)$ such that 
    % % $\frac{p^*_j}{d_{i j}} > \frac{p^*_k}{d_{i k}}$ for all $k \in [m] \setminus \{ j \}$ for all $i \in I_j^*$ 
    % $\tilde{J}^*_i(p^*, x^*) \subseteq J$ for all $i \in I_j^*(p^*)$. 
    Let $\nu \in \mathbb{R}^m$ be a vector where 
    \begin{equation*}
        \nu_k = \begin{cases}
            \frac{p^*_k}{\sum_{l \in J} p^*_l} & k \in J \\ 
            -\frac{p^*_k}{\sum_{l \in J^c} p^*_l} & k \in J^c. 
        \end{cases}
    \end{equation*}
    It is easy to check that for any $h > 0$, $\bfp^* + h \nu \in H_B$. 
    % Then, consider a trajectory of \textnormal{(\ref{chores-tatonnement})} with initial point $p(0) = p^0 = p^* + h \nu$ where $h > 0$ can be arbitrarily small. 
    % One can check, if $p^*$ described as in $\textnormal{(iii)}$, there is a neighborhood of $p^*$ such that for any $p$ in that neighborhood $b^*$ still aligns with the MPB structure. 
    % This can be explained as follows: 
    Consider each agent's MPB chores. Since we assumed $\textnormal{(iii)}$ does not hold, the MPB chores are either (1) all in $J$; (2) all in $J^c$; or (3) in both $J$ and $J^c$ but $x^*_{ij} > 0$ only if $j \in J$. 
    If the prices of chores in $J$ increase proportionally, and the prices of chores in $J^c$ decrease proportionally, in cases (1) and (2) there will not be new ties because the price changes are small; and in case (3) $\bfx^*$ is still valid because only $j \in J^c$ leaves some agents' MPB bundles, and we know $x^*_{ij} = 0$ for $j \in J^c$. 
    % In this case, one can verify the budget allocation for all $j$ will keep the same. 
    % It follows that $z_j(p^0) < 0$ for all $j \in J$ and $z_{j'}(p^0) > 0$ for all $j' \in J^c$. 
    % As a result, $\tilde{z}_j(p^0) < 0$. 
    % In this case, we have $p_j(t) = p_j(0) - \int_0^t \tilde{z}_j(p(s)) d s > p_j(0),\; \forall\, j \in J$ for any $t \geq 0$, which means the prices in the subset $J$ will keep increasing and move away from $p^0$. 
    This shows that there exists a $\bfz \in Z(\bfp^* + h\nu)$ such that $z_j < 0$ if $j \in J$ and $z_j > 0$ otherwise. 
    This implies $\inp{\bfz}{\bfp - \bfp^*} < 0$ for some $\bfp$ in the neighborhood of $\bfp^*$ and thus $\textnormal{(iv)}$ does not hold. 
    % On the other hand, 
    % by~\cref{lem:descent}, we know the objective value of $f\vert_{H_B}$ is strictly decreasing. This implies $f\vert_{H_B}(p^*)$ is greater than that of prices along the trajectory, i.e., $f\vert_{H_B}(p^*) > f\vert_{H_B}(p(t))$ for any $t > 0$, which can be verified by taking $\lim_{h \downarrow 0}$ and the locally Lipschitz continuity of $f\vert_{H_B}$. 
    % Thus, again since the objective value of $f\vert_{H_B}$ is strictly decreasing along the trajectory, we know trajectory $p(t)$ will not move back to $p^*$, which contradicts the statement that $p^*$ is locally stable. 

    Finally, we show $\textnormal{(iii)}$ $\implies$ $\textnormal{(ii)}$. 
    To do this, we need the following fact~\cite[Proposition 1.5(c)]{clarke2008nonsmooth}
    \begin{equation}
        (f\vert_{H_B})^\circ(\bfp; \nu) = \max\{ \inp*{\bfz}{\nu}: \bfz \in \partial f\vert_{H_B}(\bfp) \}, \quad \forall\; \nu \in \mathbf{H}_0. 
        \label{fact:Clarke-directional-derivatives}
    \end{equation} 
    For any $\nu \in \mathbf{H}_0 \setminus \{ 0 \}$, 
    we know there is $\nu_{\max} = \max_j \nu_j > 0$ since $\sum_j \nu_j = 0$. 
    Let $J_{\max}(\nu) = \textnormal{argmax}_j \{ \nu_j \} \subset [m]$. 
    It follows that $\nu_j > \nu_{j'}$ if $j \in J_{\max}(\nu)$ and $j' \in J_{\max}^c(\nu)$. 
    % we sort all $\nu_i,\; \forall\, i$ from the greatest number to the smallest number, and denote the rank of the $\nu_i$ as $r(\nu_i) \in \{ 1, \ldots, m \}$, that is, $\nu_i \geq \nu_{i'}$ if $r(\nu_i) < r(\nu_{i'})$. Since $\sum_j \nu_j = 0$, $\nu_i > 0$ if $r(\nu_i) = 1$. 
    Let $\bfx^*$ be the equilibrium allocation satisfying the condition in~$\textnormal{(iii)}$. 
    % Note that $x^*$ is an allocation corresponding to $\tilde{z}(p^*) = 0$, so $0\in \partial \fHb$, and thus 0 is an attainable value in~\cref{fact:Clarke-directional-derivatives} for all $\nu \in \mathbf{H}_0$. 
    % This implies $(f\vert_{H_B})^\circ(p; \nu) \geq 0, \; \forall\, \nu \in \mathbf{H}_0$. 
    Using the condition in~$\textnormal{(iii)}$, we have that
    there exists a pair $j \in J_{\max}(\nu)$ and $j' \in J^c_{\max}(\nu)$ such that $x^*_{ij}, x^*_{ij'} > 0$ for some $i \in [n]$. 
    Consider another feasible demand $x'$ where 
    \begin{equation*}
        x'_{k l} = \begin{cases}
            x^*_{ij} + \frac{p_{j'}^*}{p_{j}^*}x^*_{ij'} & k = i, l = j \\ 
            0 & k = i, l = j' \\ 
            x^*_{k l} & \textnormal{otherwise} 
        \end{cases}. 
    \end{equation*} 
    Since $\sum_j \nu_j = 0$, 
    \begin{equation*}
        \inp*{\tilde{z}(\bfp^*)}{\nu} = \inp*{z(\bfp^*) + \mathbf{1}_m}{\nu} = \nu_j \sum\nolimits_i x_{ij}^* + \nu_{j'} \sum\nolimits_i x_{i j'} + \sum\nolimits_{l \neq j, j'} \nu_l \sum\nolimits_i x_{i l}^*. 
    \end{equation*}
    Then, we can check that $\bfx'$ corresponds to a relative excess demand $\tilde{\bfz}'(\bfp^*)$ that leads to a strict increase in the objective value of $\inp*{\tilde{\bfz}(p^*)}{\nu}$, that is, $(f\vert_{H_B})^\circ(\bfp^*; \nu) > 0$. 
    % As in~\citet[Definition 2]{li2020understanding}, for a locally Lipschitz function that is subdifferentially regular, the original diretional deriative exists and coincides with Clarke generalized directional derivative. 
    For a locally Lipschitz function that is subdifferentially regular, the original directional derivative exists and coincides with the Clarke generalized directional derivative~\cite[Definition 2]{li2020understanding}. 
    Therefore, $(f\vert_{H_B})'(\bfp^*; \nu) = (f\vert_{H_B})^\circ(\bfp^*; \nu) > 0$ for all directions $\nu \in \mathbf{H}_0$. 
    Thus, $\bfp^*$ is a strict local minimum of $f\vert_{H_B}$. 
\end{proof} 

As a corollary of~\cref{thm:characterize-locally-stable-chores-ce}, we can characterize locally unstable CE as well. 
\begin{corollary}
    Let $\bfp^*$ be an equilibrium price. The following statements are equivalent: 
    \begin{itemize}
        \item[(i)] $\bfp^*$ is locally unstable; 
        \item[(ii)] $\bfp^*$ is either a strict local maximum or a saddle point of the function $f\vert_{H_B}$; 
        % \item[(iii)] There exists at least one $\nu \in \mathbf{H}_0$ such that $\left(f\vert_{H_B}\right)'\left( p^*; \nu \right) = 0$; 
        \item[(iii)] 
        Let $\bfx^*$ be any equilibrium allocation coupled with $\bfp^*$. 
        There exists some nonempty $J \subset [m]$ such that 
        for every $j \in J$ and $j' \in J^c$: 
        % for each $i$ one of the following cases happens 
        \begin{align*}
            \forall\; i \in [n]: 
            % \frac{p^*_j}{d_{ij}} \neq \frac{p^*_{j'}}{d_{i j'}} \quad \text{ or } \quad 
            % \frac{p^*_j}{d_{ij}} = \frac{p^*_{j'}}{d_{i j'}} < \max_k \frac{p^*_k}{d_{ik}} \quad 
            j \notin J_i^*(\bfp^*) 
            \text{ or/and }
            j' \notin J_i^*(\bfp^*) 
            \quad 
            % & \frac{p^*_j}{d_{ij}} = \frac{p^*_{j'}}{d_{i j'}} < \frac{p^*_k}{d_{ik}} \text{ for some } k \neq j, j'; \\ 
            \text{ or } 
            \quad 
            j, j' \in J^*_i(\bfp^*) \text{ but } x^*_{i j'} = 0. 
            % & \frac{p^*_j}{d_{ij}} = \frac{p^*_{j'}}{d_{i j'}} \geq \frac{p^*_k}{d_{ik}} \quad \forall\; k \neq j, j' \text{ but } x^*_{i j'} = 0. 
        \end{align*}
        \item[(iv)] For any neighborhood $\mathbf{N}(\bfp^*)$ of $\bfp^*$, we have 
        \begin{equation*}
            \exists\; \bfz \in \tilde{Z}(\bfp), \bfp \in \mathbf{N}(\bfp^*) \cap \Delta_B \setminus \{ \bfp^* \} \quad \text{ such that } \quad 
            \inp*{\bfz}{\bfp - \bfp^*} < 0. 
        \end{equation*}
    \end{itemize}
    \label{thm:unstable}
\end{corollary}

\newpage

\section*{Acknowledgements}

The research of Bhaskar Ray Chaudhury was supported by NSF CAREER Grant CCF-2441580.
The research of Christian Kroer and Tianlong Nan was supported by 
the Office of Naval Research awards N00014-22-1-2530 and N00014-23-1-2374, and the National Science Foundation awards IIS-2147361 and IIS-2238960.
The research of Ruta Mehta was supported by NSF Grant CCF-2334461.

\bibliographystyle{alphaurl}
\bibliography{refs} 

\newpage 
\appendix 

\section{Omitted Proofs} 
\label{app:sec:Omitted proofs}

\subsection{Proof of~\cref{lem:descent}} 

First, we introduce an additional definition. A locally Lipschitz function $f: \RR^m \rightarrow \RR$ admits a \emph{chain rule} if for \emph{any} arc $\bfp: \RR_+ \rightarrow \RR^m$, equality 
    \begin{equation*}
        \frac{d}{d t} (f \circ \bfp)(t) = \inp*{\nu}{\dot{\bfp}(t)} \quad \textnormal{holds for all $\nu \in \partial f(\bfp(t))$ and for a.e. $t \geq 0$.}
    \end{equation*} 
% holds for $t \geq 0$ almost everywhere. 
% See, for example,~\citet[Definition 5.1]{davis2020stochastic}. 

To prove~\cref{lem:descent}, we show the following lemma for a more general case: we consider an arbitrary locally Lipschitz function $f: \mathbb{R}^m \rightarrow \mathbb{R}$ that admits a chain rule, and prove a descent lemma for a general Bregman differential inclusion. 
% \begin{equation}
%     \frac{\partial \nabla \phi(p)}{\partial t} \in -\partial f(p). \label{eq:differential-inclusion-general}
% \end{equation}
Then, both~\cref{lem:descent} can be seen as corollaries of this lemma. 
This lemma can be extracted from the proof of Proposition 4.1 of  of~\cite{ding2024stochastic}. 
We give a formal proof here for completeness. 

\begin{lemma}
    Consider a locally Lipschitz function $f: \mathbb{R}^m \rightarrow \mathbb{R}$ that admits a chain rule. 
    Let $\bfp: \mathbb{R}_+ \rightarrow \mathbb{R}^m$ be any arc satisfying the Bregman differential inclusion 
    \begin{equation}
        \frac{\partial}{\partial t} \nabla \phi(\bfp) \in -\partial f(\bfp), \label{eq:differential-inclusion-general}
    \end{equation}
    where $\nabla \phi$ is differentiable almost everywhere and $\nabla^2 \phi(\bfp(t))$ to be \emph{positive definite} for any arc $\bfp(t)$ for a.e. $t \geq 0$. Then, if $\bfp(0)$ is not a stationary point of~\cref{eq:differential-inclusion-general}, there exists a real number $T > 0$ such that 
    \begin{equation}
        f(\bfp(T)) < \sup_{t \in [0, T]} f(\bfp(t)) \leq f(\bfp(0)). 
        \label{eq:descent-general}
    \end{equation} 
    \label{lem:descent-general}
\end{lemma}

\begin{proof} 
    % The following proof applies to both~\cref{lem:descent,lem:multiplicativedescent}. 
    % In particular, 
    % we show the descent lemma for a Bregman differential inclusion: 
    % \begin{equation}
    %     \frac{\partial \nabla \phi(p)}{\partial t} \in -\partial f(p). \label{eq:differential-inclusion-general}
    % \end{equation} 
    
    % \begin{definition}[Chain rule]
    %     Consider a locally Lipschitz function $f$ on $\RR^m$. 
    %     We will say that $f$ admits a \emph{chain rule} if for any arc $p: \RR_+ \rightarrow \RR^m$, equality 
    %     \begin{equation*}
    %         (f \circ p)'(t) = \inp*{\partial f(p(t))}{\dot{z}(t)} \; \text{holds for a.e. } t \geq 0. 
    %     \end{equation*}
    % \end{definition}
    
    % Let us say ``$f$'' is a function which admits a chain rule, and it can be $f\vert_{\Hb}$ in the case of~\cref{eq:differential-inclusion-general} and $f$ (the sum-max-log function in the main texts) in the case of~\cref{lem:multiplicativedescent}. 
    
    % Consider any trajectory $p(t)$ of the differential inclusion~\cref{eq:differential-inclusion-general}. 
    % % \begin{equation*}
    % %     \frac{\partial \nabla\phi(p)}{\partial t} \in - \partial f(p). 
    % % \end{equation*} 
    % The condition tells us the initial point of the trajectory is a non-stationary point of $f$. 

    By the chain rule, 
    for almost every $s \geq 0$, we have 
    \begin{equation*}
        \frac{d}{d s} f(\bfp(s)) 
        % = \inp*{\partial f(p(t))}{\dot{p}(t)} 
        = \inp*{-\frac{\partial }{\partial s} \nabla\phi(\bfp(s))}{\dot{\bfp}(s)} = \inp*{-\nabla^2\phi(\bfp(s)) \dot{\bfp}(s)}{\dot{s}(t)}. 
    \end{equation*}
    Then, for any $t \geq 0$, it holds that 
    \begin{equation*}
        f(\bfp(t)) - f(\bfp(0)) = \int_0^t \inp*{-\nabla^2\phi(\bfp(s)) \dot{\bfp}(s)}{\dot{\bfp}(s)} ds \leq -\int_0^t \lambda_{\min}\big(\nabla^2\phi(\bfp(s))\big) \norm{\dot{\bfp}(s)}^2 \leq 0, 
    \end{equation*}
    where $\lambda(\cdot)$ denotes the minimum eigenvalue of a matrix. 
    Then, by contradiction we can show~\cref{eq:descent-general}. Suppose $f(\bfp(t)) = f(\bfp(0))$ for all $t \geq 0$.
    This means that $\norm{\dot{\bfp}(t)} = 0$ for almost all $t \geq 0$. 
    Since $\bfp(\cdot)$ is absolutely continuous, $\bfp(t) \equiv \bfp(0)$ for any $t \geq 0$. 
    By the differential inclusion, for any $t > 0$ we have $0 = \frac{\partial}{\partial t} \nabla \phi\left( \bfp(t) \right) \in - \partial f(\bfp(t)) = - \partial f(\bfp(0))$. This contradicts the assumption that the initial point of the trajectory is non-stationary. 
    Therefore, \cref{eq:descent-general} follows. 
\end{proof} 

% As a consequence, there exists some $T > 0$, such that 
% \begin{equation*}
%     f(p(T)) < \sup_{t \in [0, T]} f(p(t)) \leq f(p(0)). 
% \end{equation*}
% This shows that the objective $f$ strictly decreases along every trajectory. 
% This turns out to be very useful in this work. 

As shown in~\cite[Lemma 5.4]{davis2020stochastic}, a subdifferentially regular function admits a chain rule. 
Thus, both $f$ (defined in~\cref{eq:potential-CCH}) and $f\vert_{H_B}$ admit a chain rule. 

Thus, \cref{lem:descent} can be seen as a corollary of~\cref{lem:descent-general} when $\nabla \phi(\bfp) = \bfp$ and the potential function is defined as in~\cref{eq:potential-CCH}.
% ; and~\cref{lem:multiplicativedescent} can be seen as a corollary of~\cref{lem:descent-general} when $\nabla \phi(p) = \log{p}$ (element-wise) and the potential function is $f|_{H_B}$. 

% Both of $f$ and $h$ admit the chain rule, 
% \begin{equation*}
%     - \inp{\nabla^2 h(p(t)) \dot{p}(t)}{\dot{p}(t)} \in \inp{\partial f(p(t))}{\dot{p}(t)}
% \end{equation*} 

% \subsection{Proof of~\cref{thm:relative-tatonnement-stable,eq:multiplicative-tatonnement-stable}} 

\subsection{Proof of~\cref{thm:relative-tatonnement-stable}} 

We state a general theorem to show the convergence of the continuous-time trajectories under appropriate conditions. \cref{thm:relative-tatonnement-stable} follows from this theorem because of~\cref{lem:descent}. 

Consider a general Bregman differential inclusion: 
\begin{equation}
    \begin{aligned}
        \frac{\partial }{\partial t} \nabla \phi({\bfp}) \in D(\bfp), 
    \end{aligned}
    \label{eq:differential-inclusion}
\end{equation} 
where $D: \mathcal{X} \rightrightarrows \mathbb{R}^m$ is a set-valued map, and $\nabla \phi: \mathcal{X} \rightarrow \mathbb{R}^m$ is a function which is differentiable almost everywhere. 

\begin{theorem}
    Assume that 
    \begin{enumerate} 
        \item $\mathcal{X}$ is compact. 
        \item The set $D^{-1}(\mathbf{0})$ is disconnected. 
        \item There exists a Lyapunov function $\varphi$ which is lower bounded, i.e., $\inf_{x \in \mathbb{R}^m} \varphi(x) > - \infty$. 
        Additionally, it satisfies $\mathcal{X} \subseteq \textnormal{dom}\, \varphi$ and every point in $D^{-1}(\mathbf{0})$ is a stationary point of $\varphi$. 
        Furthermore, 
        for any trajectory $\bfp: \mathbb{R}_+ \rightarrow \mathbb{R}^m_+$ of~\cref{eq:differential-inclusion} such that $\bfp(0) \in \mathcal{X}$ that is not a stationary point of~\cref{eq:differential-inclusion}, there exists a real number $T > 0$ satisfying 
        \begin{equation}
            \varphi(\bfp(T)) < \sup_{t \in [0, T]} \varphi(\bfp(t)) \leq \varphi(\bfp(0)). 
            \label{eq:descent}
        \end{equation} 
    \end{enumerate} 
    Let $\bfp: \mathbb{R}_+ \rightarrow \mathbb{R}^m_+$ be any trajectory of~\cref{eq:differential-inclusion} with an initial point in $\mathcal{X}$. 
    Then, $\bfp(t)$ converges to some stationary point of~\cref{eq:differential-inclusion}, i.e., $\bfp(t) \rightarrow \bfp^* \in D^{-1}(\mathbf{0})$ as $t \rightarrow \infty$. 
\end{theorem} 

\begin{proof}
    Because the function value of the Lyapunov function $\varphi$ is eventually strictly decreasing along the trajectory, and $\varphi$ is lower bounded, the function value of $\varphi$ converges to some value, i.e., $\lim_{t \rightarrow \infty} \varphi(\bfp(t))$ exists. 
    % Next, we consider two cases. 

    Because $\bfp(t)$ is bounded and $\varphi(\bfp(t))$ converges, there is a sequence $\tau_0, \tau_1, \ldots, \tau_k, \ldots$ such that $\tau_k \rightarrow \infty$ as $k \rightarrow \infty$ and $\bfp(\tau_k)$ converges to some point $\bfp^*$ as $k \rightarrow \infty$. 
    % First, $p(t)$ converges. Let us say $p(t) \rightarrow p^*$ as $t \rightarrow \infty$. 
    Then, we can show $\bfp^* \in D^{-1}(0)$ by contradiction. 
    % Suppose that $p^* \notin D^{-1}(0)$, equivalently, $p^*$ is not a stationary point of~\cref{eq:differential-inclusion}. 
    Suppose that 
    % $\lim_{t \rightarrow \infty} p(t) = p^*$ and 
    $\bfp^*$ is not a stationary point. 
    Consider a sequence of arcs $\bfp^{\tau_k}(t) = \bfp(\tau_k + t), \; k = 0, 1, 2, \ldots$. 
    It is easy to see $\bfp^{\tau}(t)$ is a trajectory of~\cref{eq:differential-inclusion} for any $\tau \geq 0$. 
    Then, it can be shown that $\hat{\bfp}(t) = \lim_{k \rightarrow \infty} \bfp^{\tau_k}(t)$ is also a trajectory of~\cref{eq:differential-inclusion}. 
    By~\cref{lem:descent-general}, we have there exists a $T > 0$ such that 
    \begin{equation}
        \varphi(\hat{\bfp}(T)) < \sup_{t \in [0, T]} \varphi(\hat{\bfp}(t)) \leq \varphi(\hat{\bfp}(0)) = \varphi(\bfp^*). 
    \end{equation} 
    On the other hand, we have 
    \begin{equation}
        \varphi(\hat{\bfp}(T)) 
        % = \varphi(\lim_{k \rightarrow \infty} p^{\tau_k}(t)) 
        = \lim_{k \rightarrow \infty} \varphi(\bfp^{\tau_k}(T)) 
        = \lim_{k \rightarrow \infty} \varphi(\bfp(\tau_k + T)) = \lim_{t \rightarrow \infty} \varphi(\bfp(t)) = \varphi(\bfp^*), 
    \end{equation} 
    where the last two inequalities follow from the existence of $\lim_{t \rightarrow \infty} \varphi(\bfp(t))$ and the continuity of $\varphi$.  
    This leads to a contradiction. Therefore, $\bfp^* \in D^{-1}(\mathbf{0})$. 

    Since $\varphi(\bfp(t))$ converges, $\varphi(\bfp(t)) \rightarrow \varphi(\bfp^*)$ which is a value corresponding to the stationary point $\bfp^*$. Moreover, because the stationary points of~\cref{eq:differential-inclusion} are disconnected and the arc $\bfp(t)$ is continuous, we can rule out the case that the arc cycles around multiple stationary points. This is because there has to be non-stationary point on the path between different stationary points. Consider such a non-stationary point and a trajectory of~\cref{eq:differential-inclusion} starting from this point. By~\cref{lem:descent-general}, there will be a contradiction to $\lim_{t \rightarrow \infty} \varphi(\bfp(t)) \rightarrow \varphi(\bfp^*)$. Therefore, we can conclude $\bfp(t) \rightarrow \bfp^* \in D^{-1}(\mathbf{0})$.
\end{proof}

\subsection{Proof of~\cref{thm:discrete-time-relative-tatonnement-convergence}}

\begin{proof} 
    First, by~\cref{lem:discrete-time-relative-tatonnement-properties} we show the sequence of iterates $\left\{ \bfp^k \right\}_{k \geq 0}$ lie in $\Delta_B$ and enjoy a positive lower bound, if the stepsizes are properly upper bounded. 
    
    Thus, Assumption 1 in~\cref{davis assumptions} are satisfied. 
    % Without loss of generality, we assume $\min_j p^0_j > 0$. 
    % If this does not hold then we can reindex and consider iterates starting from $p^1$. It is trivial to show $\min_j p^1 > 0$. 
    % % and the convergence properties of $\{ p^k \}_{k \geq 1}$ and $\{ p^k \}_{k \geq 0}$ are the same. 
    % From here, we show $p^k_j \geq \frac{\ell}{2}$ for all $j \in [m]$ and $k \geq 0$. 
    % The proof is similar to that in~\cref{lem:well-defined-equilibrium}: 
    % let $\kappa$ be a time point such that 
    % $p_j^{\kappa - 1} \geq \ell$ and $p_j^\kappa < \ell$ for some $j$. 
    % Then, we have $p^\kappa_j \geq \frac{\ell}{2}$ since $z^\kappa_j \leq \frac{\norm{B}_1}{p_j^{\kappa - 1}} - 1 \leq \frac{\norm{B}_1}{\ell} - 1$ and thus $\tilde{z}^\kappa_j \leq z^\kappa_j + 1 \leq \frac{\norm{B}_1}{\ell}$. 
    % As shown in the proof of~\cref{lem:well-defined-equilibrium}, it follows that $\tilde{z}_j(p^k) = \{ -1 \}$ for all $k = \kappa, \kappa + 1, \ldots$ such that $p^k_j \leq \ell$. 
    % Thus, $p^{k + 1} > p^k$ for each of these $k$. 
    % This shows a price lower bound for the discrete-time relative t\^atonnement under small stepsizes. 
    Moreover, it follows from~\cref{lem:discrete-time-relative-tatonnement-properties} that the prices and excess demands are bounded, and therefore Assumption 2 in~\cref{davis assumptions} is satisfied. 
    Assumption 3 in~\cref{davis assumptions} corresponds to the choice of stepsizes. 
    Assumption 4 in~\cref{davis assumptions} follows from the fact that $\partial f$ is outer-semicontinuous and compact-convex valued. 
    %\tianlong{Consider using the explanation in~\cite{ding2024stochastic}}
    % ; we leave the details to the reader. 
% As a result, chores t\^atonnement is equivalent an unconstrained minimization problem 
% \begin{equation}
%     \begin{aligned}
%         \min_p \quad & f(p) |_L. 
%     \end{aligned}
% \end{equation} 
    
    In~\cref{subsec:Relative chores t\^atonnement is stable}, we have already shown that $f\vert_{H_B}$ is a locally Lipschitz and subdifferentially regular function. 
    Furthermore, it is direct that the non-stationary values of $f\vert_{H_B}$ is dense as the corresponding finite number of chores CE are disconnected. 
    Therefore, by Corollary 5.5 (1) of~\cite{davis2020stochastic}, the theorem follows. 
\end{proof}

\section{Generalized Subdifferential}
\label{app:sec:generalized-subdifferential}

A similar point of view on differentiation of nonsmooth functions was introduced by Norkin~\cite{mikhalevich2024methods}. 
Based on their notions, a function $f$ is said to be generalized differentiable in its domain, if there is a set-valued map $G_f: \mathbb{R}^m \rightrightarrows \mathbb{R}^m$ such that for every $\bfx$, 
$G_f(\bfx)$ is nonempty, convex, compact valued, the graph of $G_f$ is closed, and 
\begin{equation*}
    f(\bfy) = f(\bfx) + \inp{\bfg}{\bfy - \bfx} + o(\bfx, \bfy, \bfg), \quad \bfg \in G_f(\bfy) \; \text{and} \; \lim_{\bfy \rightarrow \bfx} \sup_{\bfg \in G_f(\bfy)} \frac{o(\bfx, \bfy, \bfg)}{\norm{\bfx - \bfy}} = 0. 
\end{equation*} 
In the literature, $G_f$ is called \emph{generalized subdifferential}. 
An element $\bfg \in G_f(\bfx)$ is a \emph{(Norkin) generalized gradient} or \emph{pseudo-gradient} of $f$ at $\bfx$. 
Generalized-differentiable functions satisfy the local Lipschitz condition. 
They, generally, have no directional derivatives but are continuously differentiable almost everywhere. 

An important point to mention is that the generalized subdifferential is a whole family of set-valued mappings, and the \emph{minimal inclusion} of Norkin subdifferentials coincides with the Clarke subdifferential. 

\textbf{Verifying $\partial f\vert_{H_B}$.} 
In~\cite{mikhalevich2024methods}, they show 
\begin{equation*}
    G_{f\vert_{H_B}}(\bfp) = \left\{ \bfg_0 \in H_0 \, \left\vert \, \bfg_0 = \Pi_{H_0}(\bfg), \bfg \in G_f(\bfp) \right. \right\}. 
\end{equation*}
Based on the fact that the minimal inclusion of generalized subdifferential coincides with Clarke subdifferential, we know  
\begin{equation*}
    \partial f\vert_{H_B}(p) = \left\{ \bfg_0 \in H_0 \, \left\vert \, \bfg_0 = \Pi_{H_0}(\bfg), \bfg \in \partial f(\bfp) \right. \right\}. 
\end{equation*}
This is because if we consider any other genralized subdifferential $G'_f$ in the family of pseudo-gradient mappings, the corresponding $G'_{f\vert_{H_B}}$ satisfies $G'_{f\vert_{H_B}}(p) \supset \partial f\vert_{H_B}(p)$ because $G'_f(p) \supset G_f(p)$ and $\Pi_{H_0}$ is a one-to-one operator. The inclusion relation can also be verified from the proof of~\cite[Theorem 1.10]{mikhalevich2024methods}. 
Therefore, $\partial f\vert_{H_B}$ is still the minimal inclusion among the family of pseudo-gradient mappings of $f\vert_{H_B}$, and thus we verify it is the Clarke subdifferential we are looking for.

\end{document}